\def\pversion{extended}
\ifx\pversion\undefined\def\pversion{main}\fi
\def\dotextendedversion{extended}
\ifx\pversion\dotextendedversion\PassOptionsToClass{nonacm}{acmart}\fi
\documentclass[sigplan,screen]{acmart}

\usepackage{multirow}
\usepackage[ruled,vlined,linesnumbered]{algorithm2e}
\usepackage{float}
\usepackage{stfloats}
\usepackage{subcaption}
\usepackage{enumitem}
\usepackage{amsmath}
\usepackage{mathtools}
\usepackage{xcolor}
\usepackage{balance}  

\newlist{compactitem}{itemize}{3}
\setlist[compactitem]{nosep}
\setlist[compactitem,1]{label=\textbullet}
\setlist[compactitem,2]{label=\textendash}
\setlist[compactitem,3]{label=\textasteriskcentered}

\def\dotmainversion{main}   
\ifx\pversion\dotmainversion
  \newcommand{\dotextref}[2]{\expandafter\def\csname dotextref@#1\endcsname{#2}}
  \InputIfFileExists{extrefs}{}{\PackageWarning{main}{extrefs.tex not found;
    appendix references will show ??. Run make}}
  \DeclareRobustCommand{\appref}[1]{%
    \ifcsname dotextref@#1\endcsname
      \csname dotextref@#1\endcsname
    \else
      \textbf{??}\PackageWarning{main}{Appendix reference `#1' missing from
        extrefs.tex; run make}%
    \fi}
  \DeclareRobustCommand{\extdoc}{the full version of the paper~\cite{dot-arxiv}}
  \DeclareRobustCommand{\appsec}[1]{the full version of the paper~\cite{dot-arxiv}}
  \DeclareRobustCommand{\Appsec}[1]{The full version of the paper~\cite{dot-arxiv}}
\else\ifx\pversion\dotextendedversion
  \newcommand{\appref}[1]{\ref{#1}}
  \DeclareRobustCommand{\extdoc}{the Appendix}
  \DeclareRobustCommand{\appsec}[1]{Appendix~\ref{#1}}
  \DeclareRobustCommand{\Appsec}[1]{Appendix~\ref{#1}}
\else
  \PackageError{main}{Unknown \string\pversion: \pversion\space
    (expected main or extended)}{}
\fi\fi

\newcommand{\update}[1]{\textcolor{black}{#1}}
\newcommand{\remove}[1]{}

\let\dotorigappendix\appendix
\renewcommand{\appendix}{\dotorigappendix\global\let\appendix\relax}

\makeatletter
\newcommand{\tk@framegray}{0.35}  
\newcommand{\tk@backgray}{0.95}   
\definecolor{takeawayframe}{gray}{\tk@framegray}
\definecolor{takeawayback}{gray}{\tk@backgray}
\newlength{\tk@W}\newlength{\tk@H}\newlength{\tk@inH}
\newsavebox{\tk@body}\newsavebox{\tk@title}
\newcommand{\tk@rule}{0.5pt}  
\newcommand{\tk@pad}{4pt}     
\newcommand{\tk@num}[1]{\strip@pt\dimexpr#1\relax}
\ifdefined\pdfliteral
  \newcommand{\tk@rrect}[8]{%
    \begingroup
    \edef\tk@w{\tk@num{#5}}\edef\tk@h{\tk@num{#6}}%
    \edef\tk@r{\tk@num{#7pt/2}}\edef\tk@k{\tk@num{#7pt*2239/10000}}
    \edef\tk@wr{\tk@num{#5-#7pt/2}}\edef\tk@hr{\tk@num{#6-#7pt/2}}%
    \edef\tk@wk{\tk@num{#5-#7pt*2239/10000}}\edef\tk@hk{\tk@num{#6-#7pt*2239/10000}}%
    \ifnum#8=1 \edef\tk@htr{\tk@hr}\edef\tk@htk{\tk@hk}\edef\tk@wtr{\tk@wr}\edef\tk@wtk{\tk@wk}%
    \else\edef\tk@htr{\tk@h}\edef\tk@htk{\tk@h}\edef\tk@wtr{\tk@w}\edef\tk@wtk{\tk@w}\fi
    \put(#3,#4){\pdfliteral{q 0.99626 0 0 0.99626 0 0 cm #2 g
      \tk@r\space 0 m \tk@wr\space 0 l
      \tk@wk\space 0 \tk@w\space \tk@k\space \tk@w\space \tk@r\space c
      \tk@w\space \tk@htr\space l
      \tk@w\space \tk@htk\space \tk@wtk\space \tk@h\space \tk@wtr\space \tk@h\space c
      \ifnum#8=1 \tk@r\else 0\fi\space \tk@h\space l
      \ifnum#8=1 \tk@k\space \tk@h\space 0 \tk@hk\space 0 \tk@hr\space c \fi
      0 \tk@r\space l
      0 \tk@k\space \tk@k\space 0 \tk@r\space 0 c h f Q}}%
    \endgroup}
\else
  \newcommand{\tk@rrect}[8]{%
    \put(#3,#4){{\color{#1}\setlength{\unitlength}{1pt}%
      \begin{picture}(0,0)
        \put(\dimexpr#7pt/2\relax,0pt){\rule{\dimexpr#5-#7pt\relax}{#6}}%
        \ifnum#8=1
          \put(0pt,\dimexpr#7pt/2\relax){\rule{#5}{\dimexpr#6-#7pt\relax}}%
          \put(\dimexpr#7pt/2\relax,\dimexpr#6-#7pt/2\relax){\circle*{#7}}%
          \put(\dimexpr#5-#7pt/2\relax,\dimexpr#6-#7pt/2\relax){\circle*{#7}}%
        \else
          \put(0pt,\dimexpr#7pt/2\relax){\rule{#5}{\dimexpr#6-#7pt/2\relax}}%
        \fi
        \put(\dimexpr#7pt/2\relax,\dimexpr#7pt/2\relax){\circle*{#7}}%
        \put(\dimexpr#5-#7pt/2\relax,\dimexpr#7pt/2\relax){\circle*{#7}}%
      \end{picture}}}}
\fi
\newcommand{\tk@begin}[1]{%
  \setlength{\tk@W}{\linewidth}%
  \sbox{\tk@title}{\small\itshape\color{white}\strut#1}%
  \begin{lrbox}{\tk@body}%
  \begin{minipage}[b]{\dimexpr\tk@W-2\dimexpr\tk@rule\relax-2\dimexpr\tk@pad\relax\relax}\ignorespaces
}
\newcommand{\tk@end}{%
  \end{minipage}\end{lrbox}%
  \setlength{\tk@inH}{\dimexpr\ht\tk@body+\dp\tk@body+2\dimexpr\tk@pad\relax\relax}%
  \setlength{\tk@H}{\dimexpr\tk@inH+\tk@rule+\ht\tk@title+\dp\tk@title+4pt\relax}%
  \noindent\begin{picture}(\tk@W,\tk@H)
    \tk@rrect{takeawayframe}{\tk@framegray}{0pt}{0pt}{\tk@W}{\tk@H}{8}{1}%
    \tk@rrect{takeawayback}{\tk@backgray}{\tk@rule}{\tk@rule}%
      {\dimexpr\tk@W-2\dimexpr\tk@rule\relax\relax}{\tk@inH}{7}{0}%
    \put(\dimexpr\tk@rule+\tk@pad\relax,\dimexpr\tk@H-2pt-\ht\tk@title\relax){\usebox{\tk@title}}%
    \put(\dimexpr\tk@rule+\tk@pad\relax,\dimexpr\tk@rule+\tk@pad+\dp\tk@body\relax){\usebox{\tk@body}}%
  \end{picture}}
\newenvironment{takeaway}[1][Key Takeaway]%
  {\begin{figure}[!htb]\tk@begin{#1}}{\tk@end\end{figure}}
\newenvironment{takeaway*}[1][Key Takeaway]%
  {\par\addvspace{5pt}\tk@begin{#1}}{\tk@end\par\addvspace{5pt}}
\makeatother

\ifx\pversion\dotmainversion
\copyrightyear{2027}
\acmYear{2027}
\setcopyright{cc}
\setcctype{by}
\acmConference[ASPLOS '27]{Proceedings of the 32nd ACM International Conference on Architectural Support for Programming Languages and Operating Systems, Volume 1}{April 11--15, 2027}{Heraklion, Greece}
\acmBooktitle{Proceedings of the 32nd ACM International Conference on Architectural Support for Programming Languages and Operating Systems, Volume 1 (ASPLOS '27), April 11--15, 2027, Heraklion, Greece}
\acmDOI{10.1145/3845814.3850035}
\acmISBN{979-8-4007-3012-2/2027/04}
\fi

\begin{document}

\ifx\pversion\dotmainversion
\title{DigitsOnTurbo: Leveraging SIMD for Accelerating Large-number Arithmetic}
\else
\title{Leveraging SIMD for Accelerating Large-number Arithmetic}
\fi

\author{Subhrajit Das}
\orcid{0009-0009-0757-5206}
\affiliation{%
  \institution{Indian Institute of Technology Gandhinagar}
  \city{Gandhinagar}
  \country{India}}
\email{subhrajit.das@iitgn.ac.in}

\author{Abhishek Bichhawat}
\orcid{0000-0002-3075-2743}
\affiliation{%
  \institution{Indian Institute of Technology Gandhinagar}
  \city{Gandhinagar}
  \country{India}}
\email{abhishek.b@iitgn.ac.in}

\author{Yuvraj Patel}
\orcid{0009-0003-2059-6403}
\affiliation{%
  \institution{The University of Edinburgh}
  \city{Edinburgh}
  \country{United Kingdom}}
\email{yuvraj.patel@ed.ac.uk}

\ccsdesc[500]{Mathematics of computing~Arbitrary-precision arithmetic}
\ccsdesc[300]{Mathematics of computing~Mathematical software performance}
\ccsdesc[500]{Computer systems organization~Single instruction, multiple data}
\ccsdesc[300]{Computing methodologies~Vector / streaming algorithms}

\keywords{Multi-precision arithmetic; Algorithm restructuring; SIMD; AVX-512; Carry propagation; Vertical and crosswise multiplication}

\begin{abstract}
  
  Large-number arithmetic, widely used in scientific computing and cryptography, has seen limited adoption of \emph{single instruction, multiple data} (SIMD) parallelism on modern CPUs due to the inherent dependencies in traditional algorithms. We present \texttt{Digits\-On\-Turbo} (DoT), which restructures the computation around independent, data-parallel operations, rather than vectorizing the standard algorithms, thereby leveraging the benefits provided by SIMD. Over prior SIMD implementations, DoT achieves up to $1.85\times$ speedups for addition and subtraction, and $2.3\times$ for multiplication. When integrated into state-of-the-art libraries (e.g., GMP, OpenSSL), DoT yields up to $4\times$ speedup for addition and subtraction, and up to $2\times$ speedup for multiplication, cascading into end-to-end throughput gains of up to $15.8\%$ for scientific computations, and up to $5.6\%$ latency and $6.2\%$ throughput improvements on cryptographic implementations.

\end{abstract}

\maketitle
\renewcommand{\shortauthors}{Subhrajit Das, Abhishek Bichhawat, \& Yuvraj Patel}

\section{Introduction}
\label{sec:intro}

\emph{Single instruction, multiple data} (SIMD) allows modern CPUs to exploit
data-parallel execution by running a single (same) instruction on multiple data
elements, simultaneously~\cite{flynn66, hennessy2011computer}. Different
applications like linear algebra libraries (e.g., OpenBLAS~\cite{openblas},
Intel MKL~\cite{intelMKL}), and computer vision and AI/ML
frameworks~\cite{armPL,ermig1979SimdLibrary} have adopted SIMD by using
specialized instructions to maximize CPU throughput. Interestingly, this has
not been the case with applications performing \emph{large-number arithmetic}
that involve operations on operands that range from hundreds to thousands of
bits. 

Large-number arithmetic is extensively used in high-preci\-sion scientific
computing~\cite{HPFloating, HPArithPhys, BAILEY201210106,
sagemathSageMathMathematical, wolframWolframMathematica,
sourceforgeMaximaBased, MATLABSymbolic, ubordeauxPARIGPDevelopment,mpfrMPFRLibrary,boostBoost}, and cryptographic
applications~\cite{MathematicsPublicKey, RSA, ECCC, ECDSA, frey2010arithmetic,
OpenSSL, Ersben, ibmCommonCryptographic, sslMinimumSize,
forbesCurrentEncryption, diffie2022new, nistKeySize}; however, the libraries
used for performing these operations do not capitalize on the benefits provided
by SIMD architectures. The GNU Multiple Precision Arithmetic Library
(GMP)~\cite{gmplibBignumLibrary}, one of the most widely deployed
arbitrary-pre\-cis\-ion libraries, avoids the use of SIMD for accelerating
these operations because ``\textit{[SIMD does not provide] much support for
propagating the sort of carries that arise in GMP}''~\cite{gmplibAssemblySIMD}.

The cost of operating on large numbers shows up directly in workloads that
depend on these primitives. For instance, computing $\pi$ to $1M$ digits takes
$18\times$ longer than for $100K$ digits~\cite{gmplibGMPbenchResults}, and
decryption using a $4096$-bit RSA key takes $5\times$ more time than using a
$2048$-bit key~\cite{javamexLengths}. While Intel and AMD are progressively investing 
more die area to support wider SIMD units to expand the horizon of applications
utilizing the data parallelism~\cite{intelAVX512, amd_epyc_5gen_hpc}, workloads 
involving large-number arithmetic operations are falling behind in terms of
performance due to their reliance on the scalar pipeline. 

Large-number arithmetic operations carry strong sequential dependencies between
data elements, or \emph{limbs}. For instance, when adding two $256$-bit numbers
(or $4$ limbs on a $64$-bit machine), each $64$-bit limb is added separately,
and the generated carry-bit is propagated to the next limb. As SIMD does not
propagate the carry-bit across data elements in the same run, the overhead 
increases. Prior work~\cite{Ren, numberworldIntegerAddition,
gueron2016accelerating} attempted to break these dependencies, but with limited
success. As we show later, when adding two large numbers, managing carry-bits
still requires $9\times$ to $12\times$ more time than the actual addition
operation~\cite{Ren, numberworldIntegerAddition}. Similarly, multiplying two
large numbers~\cite{gueron2016accelerating} introduces long read-after-write
chains, resulting in substantial throughput reduction.

Through quantitative and qualitative analysis, we observe that the
state-of-the-art approaches directly apply the existing sequential algorithm to
SIMD, thereby inheriting the original dependency structure. Furthermore, we
observe that the SIMD computation itself is never the performance bottleneck.
Instead, the overhead of preparing and routing data for the intermediate SIMD
operations imposed by the algorithm's dependencies erodes the potential
benefits of SIMD. 

In this paper, our objective is to address these shortcomings by asking the
question: \emph{Can we reconstruct the algorithms to compute the more
independent, data-parallel operations separately from the dependent operations,
thereby improving effective SIMD utilization?} We answer this question by
proposing \texttt{DigitsOnTurbo} (DoT), an approach that leverages SIMD for
accelerating large-number arithmetic operations. As the major bottleneck for
addition and subtraction is the propagation of carry-bits across limbs, we
divide the operation into four phases, such that the more common but faster
tasks like limb addition can be done in parallel, while isolating the more
performance-heavy but rarer task of handling cascading carries. For
multiplication, we use the \emph{vertical and crosswise} multiplication
technique~\cite{multStrategies, asmj2006, archiveVedicMathematics}, which
allows us to perform independent crosswise multiplication operations before
dependencies are accounted for.

We implement DoT in C with x86-64 SIMD intrinsics~\cite{intelIntrinsics}. For
addition and subtraction, DoT reduces the carry management overhead by over
$2\times$ compared with prior SIMD approaches \cite{Ren,
numberworldIntegerAddition}, delivering effective speedups of
up to $1.85\times$ over both those works and an optimized scalar
add-with-carry baseline on an Intel Emerald Rapids CPU. DoT multiplication
achieves a $2.31\times$ speedup over prior work~\cite{gueron2016accelerating}.

We integrate DoT into the GMP library~\cite{gmplibBignumLibrary} (DoTMP) and OpenSSL~\cite{OpenSSL} (DoTSSL) by replacing
the libraries' addition and subtraction primitives and the base-case
multiplication operation. \update{We evaluate the performance of DoT across both generic C (\emph{generic}) and native assembly (\emph{asm}) library configurations.}\remove{.} 

\update{On microbenchmarks, DoTMP improves GMP's generic build addition and subtraction by $3.7\times$, and multiplication by $1.5\times$, and DoTSSL accelerates OpenSSL's generic build addition by $3.0\times$, subtraction by $4.0\times$, and multiplication by $1.2\times$. On the same microbenchmarks over the \emph{asm} builds, DoTMP performs better than GMP's hand-tuned assembly by about $1.5\times$ on both addition and subtraction while it slows down multiplication to $0.95\times$, and DoTSSL improves OpenSSL's assembly performance by $2.8\times$, $1.8\times$ and $1.1\times$ on addition, subtraction and multiplication, respectively.}

On GMPbench~\cite{gmplibGMPbenchResults}, \update{over generic baselines,} DoTMP improves the overall score by \update{$7.2\%$}\remove{$7.8\%$}, with additional gains in other operations: division improves by \update{$8.6\%$}\remove{$8.4\%$}, $\pi$ computation by \update{$15.8\%$}\remove{$19.3\%$}, and multiplication for larger operands by \update{$51.2\%$}\remove{$48.7\%$}. \update{On OpenSSL speed,} DoTSSL improves \remove{OpenSSL} throughput by 
up to \update{$6.2\%$}\remove{$5.9\%$} for \textit{Finite Field Diffie–Hellman} (FFDH) and \update{$5.4\%$}\remove{$5.2\%$} for \textit{Digital Signature Algorithm} (DSA) \update{over the generic build}, and reduces latency by up to \update{$4.7\,\%$}\remove{$7.9\,\%$} for DSA, \update{$5.6\,\%$}\remove{$6.1\,\%$} for FFDH, and \update{$4.7\,\%$}\remove{$5.5\,\%$} for RSA. 
\update{On the hand-tuned assembly baselines, i.e., the \emph{asm} builds, DoTMP and DoTSSL perform similar to the libraries, and stay within $\pm 2\%$ of GMP (asm) on GMPbench and $\pm 2.1\%$ of OpenSSL (asm) on throughput. We believe that assembly-level optimizations of DoT may result in better performance by leveraging the benefits of SIMD.}

\update{We further evaluate DoTMP on scientific tools like PARI/\-GP~\cite{ubordeauxPARIGPDevelopment}, MPFR~\cite{mpfrMPFRLibrary}, and Boo\-st.\-Multi\-pre\-ci\-sion~\cite{boostBoost} (that offer GMP as a backend) and observe substantial gains. Over the \emph{generic} baselines, the performance on PARI/GP improves by up to $578\%$ for addition and up to $16.6\%$ for operations like square, Boost improves by up to $194\%$ for addition and up to $52\%$ for multiplication, and MPFR improves addition and subtraction by $6.8\%$ and $13.4\%$, respectively. Over the \emph{asm} baselines, except for addition and subtraction where DoT provides significant speedups over GMP's assembly backend (e.g., $73.0\%$ on PARI over GMP (asm)), other operations yield similar results with performance around $\pm 1$--$2\%$.}

In summary, the contributions of this work are:

\begin{compactitem}

  \item \textbf{DoT addition and subtraction:} A 4-phase algorithm that reduces carry management
overhead by nearly $2\times$ and delivers strong SIMD speedups ($1.8\times$),
with the carry adjustment phase provably negligible for random inputs.

  \item \textbf{DoT multiplication:} A SIMD multiplication routine that exposes
all partial products as independent operations, eliminating RAW hazards and
delivering $2\times$ speedups for 256-bit multiplication.

  \item \textbf{End-to-end impact:} \update{Over generic baselines,} DoTMP improves GMPbench's overall score by \update{$7.2\%$}\remove{$7.8\%$} ($\pi$ up to \update{$15.8\%$}\remove{$+19.3\%$}, multiplication up to \update{$51.2\%$}\remove{$+48.7\%$}) \update{and drives major real-world gains in PARI/GP (up to $578\%$), Boost (up to $194\%$), and MPFR (up to $13.4\%$)}; DoTSSL
  improves RSA, FFDH, and DSA throughput by up to \update{$6.2\%$}\remove{$5.9\%$} and reduces latency by up to 
  \update{$5.6\%$}\remove{$7.9\%$}. All gains \update{of DoT} come from $1{,}013$ lines of portable C with intrinsics, 
  \update{with a total of $134$}\remove{$\sim$40} lines of integration changes across GMP and OpenSSL, and no
  hand-written assembly.

\end{compactitem}


\section{Background and Motivation}
\label{sec:background}

\subsection{Large numbers and their representation}

Large numbers, a.k.a. big numbers, refer to values beyond the standard 64-bit
capabilities of common programming languages. These numbers play a fundamental
role in scientific computing~\cite{HPFloating,HPArithPhys,BAILEY201210106},
cryptography~\cite{MathematicsPublicKey,ECCC,RSA,ECDSA,frey2010arithmetic,OpenSSL,Ersben},
and various mathematical software
tools~\cite{sagemathSageMathMathematical,wolframWolframMathematica,sourceforgeMaximaBased}.
Applications in these domains deal with operands ranging from a few hundred
bits to thousands of bits~\cite{ECCC, ibmCommonCryptographic, sslMinimumSize,
OpenSSL, diffie2022new, nistKeySize}.

As standard hardware cannot support operands wider than 64
bits~\cite{hennessy2011computer}, software-based approaches such as the GNU
Multiple Precision Arithmetic Library (GMP)~\cite{gmplibBignumLibrary}, GNU
Multiple Precision Floating-Point Reliable library
(MPFR)~\cite{mpfrMPFRLibrary}, OpenSSL~\cite{OpenSSL}, Fast Library for Number
Theory (Flint)~\cite{Flint}, mpmath~\cite{mpmathMpmathPython},
Apfloat~\cite{apfloatApfloatArbitrary}, BigInt~\cite{mozillaBigIntJavaScript},
SageMath~\cite{sagemathSageMathMathematical} and
Mathematica~\cite{wolframWolframMathematica} represent large numbers as arrays
of fixed-size \emph{limbs}. A limb is typically a 32-bit or 64-bit unsigned
word aligned with the machine width. Two radix styles are common: a
\emph{saturated} representation uses radix $2^{32}$ or $2^{64}$ to maximize
bits per limb~\cite{Ren, Ersben}, while an \emph{unsaturated} (reduced-radix)
representation keeps headroom in each limb (e.g., base $2^{51}$ for
P-256~\cite{gueron2015fast}) to simplify carry handling and field-specific
reduction.

\subsection{SIMD and Large-number Arithmetic}
\label{sec:bg_simd_struggle}
\begin{table*}[ht]
\centering
\footnotesize
\resizebox{\linewidth}{!}{%
\begin{tabular*}{\linewidth}{@{\extracolsep{\fill}}l r  l r  l r  l r  l r@{}}
\toprule
\multicolumn{2}{c}{\textbf{Naive SIMD Add}} &
\multicolumn{2}{c}{\textbf{Ren et al.~\cite{Ren}}} &
\multicolumn{2}{c}{\textbf{Two-level KSA~\cite{numberworldIntegerAddition}}} &
\multicolumn{2}{c}{\textbf{Naive SIMD Mul}} &
\multicolumn{2}{c}{\textbf{Gueron et al.~\cite{gueron2016accelerating} (IFMA)}} \\
\cmidrule(r){1-2} \cmidrule(lr){3-4} \cmidrule(lr){5-6} \cmidrule(lr){7-8} \cmidrule(l){9-10}
Operations & \% & Operations & \% & Operations & \% & Operations & \% & Operations & \% \\
\midrule
Load                  &  2.6 & Load              &  5.8 & Load                &  8.0 & Radix conversion        &  8.9 & Radix conversion       &  6.7 \\
Add                   &  1.8 & Add               &  6.7 & Add                 &  8.5 & Load \& permute         & 12.7 & Load \& permute        &  9.7 \\
Initial carry detect  &  1.0 & Generate carry    & 69.5 & Generate carry      & 43.8 & IFMA compute            & 73.8 & IFMA compute           & 36.7 \\
Sequential carry prop & 92.8 & Add carry         & 13.5 & Add carry           & 33.6 & Store \& normalize      &  4.6 & Store \& normalize     & 46.9 \\
Store                 &  1.8 & Store             &  4.5 & Store               &  6.1 &                         &      &                        &      \\
\midrule
\textbf{Carry/Add ratio}& 52.1 &                   & 12.4 &                     &  9.1 &                         &   -- &                        &   -- \\
\bottomrule
\end{tabular*}%
}

\caption{\small{Operation-wise cycle breakdown (\%) of prior SIMD arithmetic routines: 512-bit addition and 256-bit multiplication. For addition routines, the carry-to-add overhead ratio (carry-handling cycles\,/\,addition cycles) is also shown. Carry-handling includes initial carry detection, carry generation, sequential carry propagation, and add carry steps.}}
\label{tab:ablation-existing}
\end{table*}

On x86-64, SIMD has evolved from 128-bit SSE~\cite{intelSSE} to 256-bit
AVX2~\cite{intelAVX2} and 512-bit AVX-512~\cite{intelAVX512}. A single AVX-512
instruction processes eight 64-bit operands in parallel; modern CPU cores
provide multiple dedicated SIMD execution units~\cite{intel2024optimization},
delivering substantially higher data-parallel throughput. The workloads, where
SIMD has paid off (linear algebra, computer vision, AI/ML, and modular
arithmetic~\cite{intelMKL,openblas,armPL,gueron2012software,DidierModulo,Buhrow21,ghosslrsaz4kavxifma}),
share one property: their sub-tasks are largely independent and map cleanly
onto vector \emph{lanes}\footnote{Lanes are the independent 128-bit segments
that compose SIMD registers (1 for SSE, 2 for AVX2, and 4 for AVX-512), each
executing packed data elements in parallel without cross-lane dependencies.}.
Large-number addition, subtraction, and multiplication do not. Their textbook
formulations carry strong sequential dependencies between limbs; these
dependencies, not SIMD compute throughput, bound the achievable lane
parallelism. 

\subsubsection{Addition and Subtraction.} The textbook limb-by-limb addition
computes $S_i = A_i + B_i + C_{i-1}$~\cite{knuth1997art,brent2010modern}, where
$A_i$ and $B_i$ are the $i$-th limbs of the inputs, $C_{i-1}$ is the carry-in
from the previous limb, and $S_i$ is the resulting limb. The carry-out $C_i$ is
generated based on whether the sum exceeds the limb base, thereby creating a
carry chain where the carry-out at position $i$ becomes the carry-in at
position $i{+}1$. Thus, an $m$-limb addition reduces to a chain whose latency
grows linearly with $m$. On the scalar pipeline, the chain hides inside
hardware add-with-carry instructions
(\texttt{ADC}/\texttt{SBB}~\cite{intelSDM,armADDC}) and serves as the basis for
the hand-tuned assembly that ships with GMP and OpenSSL.

SIMD lanes, on the other hand, are isolated: AVX-512 does not propagate carries
across lanes natively~\cite{intelSDM,intelavx512perm}, and ARM SVE2's
\texttt{ADCLB}/\texttt{ADCLT}~\cite{armSVE2} only propagate a carry between
adjacent even--odd element pairs within a single instruction, so a full
$m$-limb carry chain still requires a sequence of dependent operations. A
direct SIMD port of the carry loop, therefore, has to rebuild the hardware
carry chain in software, using compares, masks, and conditional adds. 
In our measurements, an operation-wise breakdown of the cycles spent
(Table~\ref{tab:ablation-existing}) shows that the naive SIMD approach has a
carry-to-add overhead ratio of $52.1$: for every cycle spent on actual
addition, over fifty cycles go to sequential carry propagation. Subtraction shows the same
pattern with borrows.

\subsubsection{Multiplication.} Multiplying two $m$-limb numbers needs $m^2$
word-level products $a_i \times b_j$ that have to be accumulated into a
$2m$-limb result. The two textbook organizations, row-wise schoolbook, and
column-wise Comba~\cite{knuth1997art,brent2010modern,comba90}, differ in how
they fold these products together, but both run the inner routine through a
single accumulator. For larger operands, libraries fall back to Karatsuba
(Algorithm~\appref{algo:karatsuba} in \appsec{sec:appendix-bg-hw}),
Toom-Cook, or FFT-based
methods~\cite{karatsuba1963multiplication,toom1963complexity,schonhage1971fast},
but the recursion eventually bottoms out at this same base-case routine.
Crucially, Karatsuba's identity for multiplying two limbs, $A$ and $B$: $A
\times B = (A_H \times B_H) \times \beta^2 + \bigl[(A_H{+}A_L)\times(B_H{+}B_L)
- A_H \times B_H - A_L \times B_L\bigr] \times \beta + A_L \times B_L$, where
  $\beta = 2^{k}$ is the limb base, trades one full-size multiplication for
three half-size ones plus \emph{additions and subtractions} at every level. As
recursion deepens, the aggregate addition/subtraction work grows rapidly;
faster addition and subtraction therefore improves not only the primitives themselves but also
the recursive multiplication, and through it, division, modular exponentiation,
and computations like $\pi$. 

A direct SIMD port of the schoolbook inherits the inner-loop dependency: each
iteration broadcasts a $b_j$, performs a \emph{Fused Multiply-Add} (FMA) on the
resulting row into the accumulator, and has to finish that fold before the next
iteration can start. The routine is not throughput-bound on the multiplier, but
latency-bound on the long RAW chain through the accumulator, which leaves the
multiple available CPU FMA execution ports
underutilized~\cite{intelSDM} (detailed breakdown in
Table~\ref{tab:ablation-existing}). 

\Appsec{sec:appendix-bg-hw} covers the scalar hardware and library
context (ALU width, hardware add-with-carry chains, base-case routines in GMP
and OpenSSL, and recursive multiplication thresholds) in detail.

\subsection{Prior SIMD-based works}
\label{sec:bg_existing}

Prior work attempted to break the dependency we discussed above. However, we
observe that across all solutions, the sequential chain is removed in one place
only to reappear, in a different form, somewhere else: usually at the
SIMD-scalar boundary or inside a shared accumulator.

\subsubsection{Carry-select} Ren et al.~\cite{Ren} target cross-lane carry
propagation by simulating carry-select adders~\cite{carrySelect} to mitigate
the long carry dependency chain. Their method breaks down the addition of large
integers into smaller, parallel additions of 8-bit operands.  Each limb's
addition falls into one of three cases: no-carry, propagate, or generate,
determined solely by the smaller operands. The smaller operands are arranged to
match the carry behavior of the original operands, enabling fast determination
of carries across all limbs.

However, \emph{preparing} those packed states needs a
substantial amount of SIMD-to-scalar conversion. The resulting carry-to-add
overhead ratio is $12.4$ (Table~\ref{tab:ablation-existing}), meaning carry
management still dominates runtime by over an order of magnitude. While the
authors report gains on isogeny workloads, they do not directly compare their
approach against GMP or OpenSSL, while acknowledging the high overheads.

\subsubsection{Two-level Kogge-Stone} The tool, y-cruncher~\cite{Ycruncher},
used for computing the value of mathematical constants like $\pi$ uses a two-level Kogge-Stone-style addition~\cite{KoggeStone}.
While the original KSA is a recursive radix-2 approach with a depth of $\log n$
for $n$-bit operands, it adopts a two-level approach with radices of
$\frac{n}{k}$ and $k$. In this approach, the operands are divided into $k$
groups, each of size $\frac{n}{k}$ bits, and processed independently at the
first level. The second level aggregates results across the $k$ groups to
compute final carry-bits and adjust sums using carry and max-sum information.
The second level sidesteps the per-limb serialization of the naive SIMD path;
however, the second-level resolution (carry preparation, mask manipulation, and
SIMD-to-scalar transitions) now dominates the runtime. The carry-to-add
overhead ratio is $9.1$ (Table~\ref{tab:ablation-existing}), which is better
than naive or Ren et al., but the carry chain's cost has migrated rather than disappearing. 

\subsubsection{FMA schoolbook with a shared accumulator} For multiplication,
prior SIMD work ranges from early SSE2 routines~\cite{bigmulSSE2} to
AVX-512IFMA\footnote{AVX-512IFMA (Integer Fused Multiply-Add) provides mainly
two instructions: \texttt{vpmadd52luq} and \texttt{vpmadd52huq}, for performing
integer fused multiply-add on 52-bit operands.}
implementations~\cite{keliris2014investigating,gueron2016accelerating,Edamatsu1,Edamatsu2}.
The state-of-the-art IFMA routine by Gueron et
al.~\cite{gueron2016accelerating} adapts the schoolbook organization to
AVX-512IFMA: each iteration broadcasts a $b_j$ limb, FMAs the resulting row
into a single shared accumulator vector, and drains one output limb per
iteration through \texttt{alignr\_epi64}. The scalar carry chain is gone: all
the work now stays in AVX-512 registers. But the accumulator is shared across
iterations, creating a long RAW dependency chain through madd52 $\to$ alignr
$\to$ madd52. Although newer x86-64 CPUs expose two IFMA execution ports
\cite{intel2024optimization}, only one can issue per cycle along this chain, so
the routine is latency-bound rather than throughput-bound, achieving an IPC of
only $4.2$ despite the availability of two IFMA ports. A full phase-wise
breakdown is in Table~\ref{tab:ablation-existing}. We could not directly
evaluate the other IFMA-based works~\cite{Edamatsu0,Edamatsu1,Edamatsu2} due to
the lack of source code and insufficient detail in the papers to reimplement
them faithfully.

\subsubsection{Takeaway} Across all methods in
Table~\ref{tab:ablation-existing}, the SIMD compute phase is never what limits
performance. The surrounding work needed to feed the intermediate computations
around a dependency imposed by the sequential algorithm dominates the
computation. None of the prior work goes back and restructures the operation
around better independent, data-parallel work in the first place. For addition,
they try to eliminate the sequential carry dependency, either by paying
a heavy SIMD-to-scalar transition for preparing the carry states or through
complex two-level reduction, diminishing the SIMD benefits. On the
multiplication side, the dependency through the shared accumulator creates a
long RAW chain that starves the available CPU execution ports.

Two natural questions follow: (1) Are carry cascades common enough in practice
to justify paying the high overhead of complex carry-handling costs? (2) Can
multiplication's partial products be reorganized to eliminate the
shared-accumulator dependency?

\begin{table}[!t]
  \centering
  \small
\renewcommand{\arraystretch}{1.3}
\begin{tabular}{|c|p{0.8\columnwidth}|}
\hline
\textbf{Symbol} & \textbf{Definition} \\
\hline
$k$ & Bit-width of a single limb \\
$n$ & Total bit length of a large integer operand \\
$m$ & Total number of limbs needed to represent an operand\\
$w$ & SIMD vector width (no. of limbs processed in parallel) \\
\textbf{a, b, r, ...} & SIMD vector registers \\
\update{$\mathbf{p0}$, $\mathbf{p1}$, ...} & \update{SIMD mask (or predicate) registers} \\
$A_i, B_i$ & Individual $k$-bit limb values \\
\hline
\end{tabular}
\caption{\small{Notation used throughout the DoT algorithm
    descriptions.\remove{ Boldface lowercase letters (\textbf{a}, \textbf{b}, \textbf{r},
    \ldots) denote SIMD vector registers; uppercase subscripted
    letters ($A_i$, $B_i$) denote individual $k$-bit limb values.}
  }}
\label{tab:notations}
\end{table}

\begin{table}[!htbp]
  \centering
  \small
\renewcommand{\arraystretch}{1.3}
\begin{tabular}{|l|p{0.7\columnwidth}|}
\hline
\textbf{\update{Primitive}} & \textbf{\update{Description (per limb $i$ across lanes)}} \\
\hline
\texttt{\update{\footnotesize simd\_load}} & \update{Load $w$ limbs from memory into SIMD register} \\
\texttt{\update{\footnotesize simd\_store}} & \update{Store $w$ limbs of SIMD register to memory} \\
\texttt{\update{\footnotesize simd\_add}} & \update{$r_i = (a_i + b_i) \bmod 2^k$} \\
\texttt{\update{\footnotesize simd\_cmp\_lt}} & \update{$\mathbf{p}_i = [\, x_i < y_i \,]$} \\
\texttt{\update{\footnotesize simd\_cmp\_eq}} & \update{$\mathbf{p}_i = [\, x_i = y_i \,]$} \\
\texttt{\update{\footnotesize simd\_add\_mask}} & \update{$r_i = (x_i + \mathbf{p}_i) \bmod 2^k$} \\
\texttt{\update{\footnotesize simd\_mul\_lo}} & \update{$r_i = (a_i \cdot b_i) \bmod 2^k$} \\
\texttt{\update{\footnotesize simd\_mul\_hi}} & \update{$r_i = \lfloor (a_i \cdot b_i)/2^k \rfloor$} \\
\hline
\end{tabular}
\caption{\small{\update{SIMD primitives used in DoT; these operate on $w$ limbs in parallel, where $w$ is the SIMD width. Arithmetic and comparison primitives perform operations lane-wise across the SIMD register.}}}
\label{tab:simd-primitives}
\end{table}
\begin{algorithm}[!htbp]
\caption{DoT Addition}
\label{algo:dot_add}
\footnotesize
\SetAlgoLined

\KwData{Arrays $A$, $B$ each with $m$ limbs of $k$ bits}
\KwResult{Sum array $S$ and final carry $c_{out}$}

\textcolor{gray}{\tcp{$n$-bit addition: process limbs in chunks of $w$}}
\textbf{Function:} \textsc{DoT-Add-Words}$(A, B, m)$\\
\Begin{
  $c_{out} \leftarrow 0$\;
  \For{$i \leftarrow 0$ \textbf{to} $m - 1$ \textbf{step} $w$}{
    $(X, Y) \leftarrow (A[i .. i+w-1], B[i .. i+w-1])$\;
    $(\mathbf{r}, c_{out}) \leftarrow \textsc{Add-W-Limbs}(X, Y, c_{out}, w)$\;
    \textcolor{gray}{\tcp{Store Results}}
    $\texttt{simd\_store}(w \text{ limbs of }\mathbf{r} \text{ into } S[i] \text{ to } S[i+w-1])$\; \label{alg:store2}
    
  }
  \Return{$(S, c_{out})$}\;
}

\textcolor{gray}{\tcp{Function for adding $1 \leq w \leq 8$ limbs simultaneously}}
\textbf{Function:} \textsc{Add-W-Limbs}$(A, B, c_{in}, w)$\\
\Begin{
    $\mathbf{max} \leftarrow$ vector of $w$ limbs, each with value $2^{k}-1$\;

    \textcolor{gray}{\tcp{Phase 1: Parallel Addition}}
    $\mathbf{a} \leftarrow \texttt{simd\_load}(w \text{ limbs from } A[0] \text{ to } A[w-1])$\;\label{alg:ld1}
    $\mathbf{b} \leftarrow \texttt{simd\_load}(w \text{ limbs from } B[0] \text{ to } B[w-1])$\;\label{alg:ld2}
    $\mathbf{r} \leftarrow \texttt{simd\_add}(\mathbf{a}, \mathbf{b})$ \;\label{alg:addlimbs}

    \textcolor{gray}{\tcp{Phase 2: Carry Generation}}
    $\update{\mathbf{p0}} \text{\remove{$\mathbf{c}$}} \leftarrow \texttt{simd\_cmp\_lt}(\mathbf{r}, \mathbf{a})$\ \textcolor{gray}{\tcp*{Generate mask} \label{alg:cmplt1}}
    $c_{out} \leftarrow \update{\mathbf{p0}} \text{\remove{$\mathbf{c}$}} \gg (w-1)$\; \label{alg:extractcout1}
    $\update{\mathbf{p0}} \text{\remove{$\mathbf{c}$}} \leftarrow (\update{\mathbf{p0}} \text{\remove{$\mathbf{c}$}} \ll 1) \mid c_{in}$\; \label{alg:shiftorcin}

    \textcolor{gray}{\tcp{Phase 3: Carry Propagation}}
    $\mathbf{r'} \leftarrow \texttt{\update{simd\_add\_mask} \remove{simd\_add\_carry}}(\mathbf{r}, \update{\mathbf{p0}} \text{\remove{$\mathbf{c}$}})$\; \label{alg:addwithcarry1}
    $\update{\mathbf{p1}} \text{\remove{$\mathbf{c}$}} \leftarrow \texttt{simd\_cmp\_lt}(\mathbf{r'}, \mathbf{r})$\ \textcolor{gray}{\tcp*{Generate mask}} \label{alg:cmplt2}

    \If{$\update{\mathbf{p1}} \text{\remove{$\mathbf{c}$}} \neq 0$}{
        \textcolor{gray}{\tcp{Phase 4: Carry Adjustment}}
        $\update{\mathbf{p1}} \text{\remove{$\mathbf{c}$}} \leftarrow \update{\mathbf{p1}} \text{\remove{$\mathbf{c}$}} \ll 1$\; \label{alg:shiftleft}
        $\update{\mathbf{p2}} \text{\remove{$\mathbf{m}$}} \leftarrow \texttt{simd\_cmp\_eq}(\mathbf{r'}, \mathbf{max})$\ \textcolor{gray}{\tcp*{Generate mask}} \label{alg:cmpeq}
        $\update{\mathbf{p3}} \text{\remove{$\mathbf{c}$}} \leftarrow \update{\mathbf{p1}} \text{\remove{$\mathbf{c}$}} + \update{\mathbf{p2}} \text{\remove{$\mathbf{m}$}}$\; \label{alg:addmask}
        $c_{out} \leftarrow c_{out} \mid (\update{\mathbf{p3}} \text{\remove{$\mathbf{c}$}} \gg w)$\; \label{alg:extractcout2}
        $\update{\mathbf{p4}} \text{\remove{$\mathbf{m}$}} \leftarrow \update{\mathbf{p3}} \text{\remove{$\mathbf{c}$}} \oplus \update{\mathbf{p2}} \text{\remove{$\mathbf{m}$}}$\; \label{alg:xormask}
        $\mathbf{r'} \leftarrow \texttt{\update{simd\_add\_mask} \remove{simd\_add\_carry}}(\mathbf{r'}, \update{\mathbf{p4}} \text{\remove{$\mathbf{m}$}})$\; \label{alg:addwithcarry2}
    }

    \Return{$(\mathbf{r'}, c_{out})$}\;
}
\end{algorithm}

\section{Design}
\label{sec:design}

We describe the design of \texttt{DigitsOnTurbo} (DoT), a SIMD-based algorithm
for large-number addition (and subtraction\footnote{Without loss of generality,
we focus on addition in this paper; differences for subtraction are noted where
they arise.}), and multiplication. In both cases, the design starts by
restructuring the computation around independent, data-parallel operations,
rather than vectorizing an existing sequential implementation, and builds SIMD
phases around that structure. Our design focuses on maximizing parallelism
while guaranteeing correctness under all inputs and microarchitecture
portability. 

\subsection{DoT Addition and Subtraction}

As discussed earlier, addition requires carry propagation, which introduces
necessary dependencies and limits the SIMD utilization. Carefully analyzing the
random nature of large number inputs and the limb representation, we observe
that propagating carry-chains is rare in practice; a fact not accounted for by
prior work. 

Using a \textit{limb-based approach}, addition can be limited to a single
neighboring carry-propagation. In other words, corresponding limbs of both
operands can be added in parallel, followed by the generated carry-bits to be
propagated to preceding intermediate sums in parallel, too. In the majority of
cases, propagating a carry-bit to preceding  intermediate sums may \textit{not}
generate an additional carry-bit. A new carry-bit is generated only when the
earlier carry of 1 is added to an intermediate sum that equals the maximum
value for the base (e.g., $2^{64}-1$ for 64-bit limbs). Such propagation
cascades only when all intermediate limbs compute to the maximum base-value.
However, the probability of each of the sums computing to the maximum possible
value is quite less ($\frac{1}{2^k}$ where $k$ is the number of bits in a
limb), which allows us to avoid these cascading propagations in the majority of
cases and defer them to a slow path that only executes when such a cascade
occurs. \Appsec{sec:appendix-carry} provides a formal analysis of the
same, showing that the probability of a carry cascade is negligible for random
inputs.
\begin{figure}[!bp]
    \centering
    \includegraphics[width=0.45\textwidth]{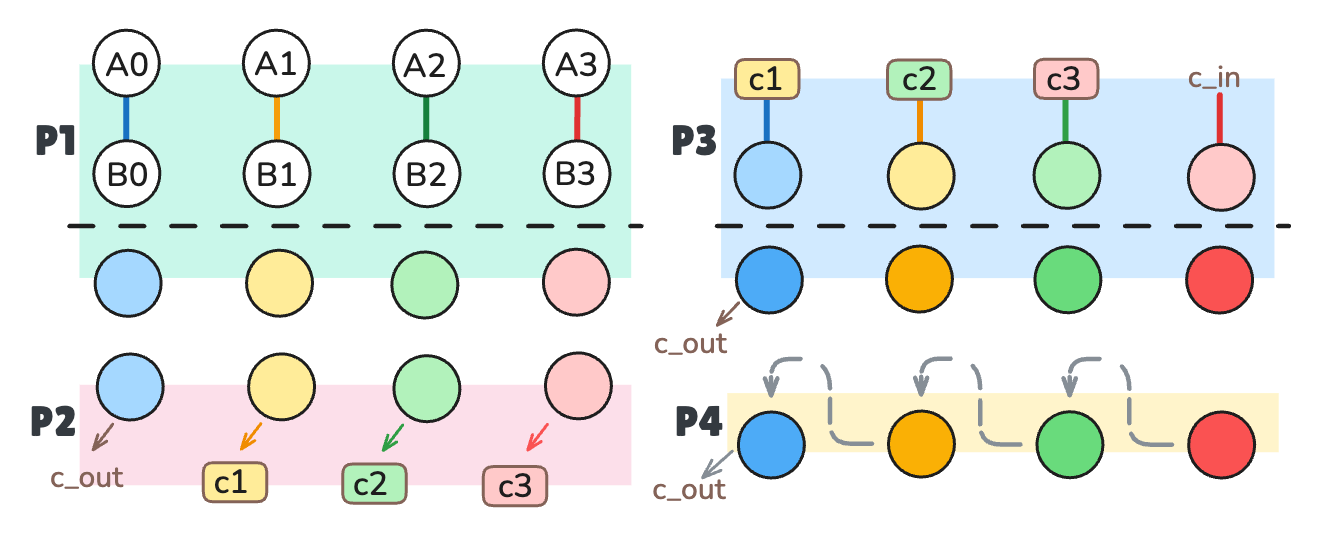}
    \caption{\small{Illustration of DoT addition for a 4-limb example. Phase 1 (P1) and Phase 3 (P3) perform \texttt{SIMD ADD} in parallel; Phase 2 (P2) generates and shifts carry-bits on scalar/mask registers; Phase 4 (P4) handles the rare carry-cascade case via the slow path.}}
    \label{fig:dot_addition}
\end{figure}

Considering the above observation, we propose a four-phase \textsc{DoT
Addition} algorithm. \textbf{Phase~1} (P1) performs limb-wise addition of $A_i$
and $B_i$ in parallel, producing intermediate sums without any carry
management. \textbf{Phase~2} (P2) detects the carries from P1 and shifts them
one position to align with the limbs they must propagate into, with the
top-limb carry extracted as $c_{out}$. \textbf{Phase~3} (P3) adds these aligned
carries (and any incoming $c_{in}$) to the intermediate sums in a single
parallel step. \textbf{Phase~4} (P4) handles the cascading case where P3 itself
generates a new carry by adjusting the sums with a carry mask.
Figure~\ref{fig:dot_addition} illustrates the flow for a 4-limb example.

Algorithm~\ref{algo:dot_add} shows the pseudocode. We walk through the
algorithm with $n{=}512$-bit operands and $k{=}64$-bit limbs, so $m{=}8$ limbs
fit in a single AVX-512 call with $w{=}8$. \textsc{DoT-Add-Words} calls
\textsc{Add-W-Limbs} for $A_0$--$A_7$, $B_0$--$B_7$, with $c_{in}=0$.

\noindent\textbf{Phase 1.} All eight limbs of $A$ and $B$ are loaded into SIMD
registers and added in parallel (lines~\ref{alg:ld1}--\ref{alg:addlimbs}),
computing $R_i = A_i + B_i$ for $i=0..7$ in one instruction. Deferring carry
detection ensures no inter-lane dependency.

\noindent\textbf{Phase 2.} A carry at position $i$ is detected by
$\update{\mathbf{p0_i}} \text{\remove{$c_i$}} = (R_i <
A_i)$ via a SIMD compare across all eight lanes (line~\ref{alg:cmplt1}). The
carry out of the top limb ($\update{\mathbf{p0_7}} \text{\remove{$c_7$}}$) is saved as $c_{out}$
(line~\ref{alg:extractcout1}); the remaining bits are shifted left by one and
the incoming $c_{in}$ is inserted at bit 0 (line~\ref{alg:shiftorcin}), giving
${\update{\mathbf{p0}} \text{\remove{$c$}}}=\update{\mathbf{p0_6}}
\text{\remove{$c_6$}} \ldots \update{\mathbf{p0_0}}
\text{\remove{$c_0$}} \,c_{in}$ and aligning each carry-bit with the
limb it must propagate into. 

\noindent\textbf{Phase 3.} The aligned carries are added to the intermediate
sums in parallel: $R_i' = R_i + \update{\mathbf{p0_i}} \text{\remove{$c_i$}}$ (line~\ref{alg:addwithcarry1}). The
comparison on line~\ref{alg:cmplt2} checks whether any of the carry additions
overflow. If not (the common case), the result is stored and the function
returns. Otherwise, Phase 4 handles the rare cascade.

\noindent\textbf{Phase 4.} The secondary carry-bits are shifted left
(line~\ref{alg:shiftleft}) and combined with a mask of limbs equal to $2^k{-}1$
(lines~\ref{alg:cmpeq}--\ref{alg:xormask}), using the carry-adjustment trick
from the Kogge-Stone adder~\cite{KoggeStone}. The adjusted carry is added back
into the sums (line~\ref{alg:addwithcarry2}) to give the final result. Because
Phase 4 operates on values that have already been carry-propagated in Phase 3;
the adjustment is bounded and correct.

When $m$ is not a multiple of $w$, the final call to \textsc{Add-W-Limbs} uses
masked SIMD load/store variants to process only the remaining limbs; smaller
operands can also use a narrower $w$ (e.g., $w{=}4$ for AVX2 or $w{=}2$ for
SSE) by selecting the corresponding intrinsic width.

\subsubsection{How DoT Reduces Carry Overhead.} Prior SIMD-based solutions do not
handle common and rare cases separately; consequently, they pay heavy overhead
for carry management even when no cascade occurs. In contrast, DoT's Phase 2 performs three simple tasks: detect the carry-bits using a SIMD compare,
shift the carry by one position to align with the limbs they must propagate
into, and extract the top-limb carry ($c_{out}$). For common cases, only the
first three phases are executed. Phase 4 is executed rarely when a carry
cascade arises. By isolating the carry propagation to a separate phase that
only executes when necessary, DoT significantly improves common case
performance. 

\subsubsection{Proof of correctness.} We now show that \textsc{DoT-Add-Words}
is correct and produces the same result as normal
addition with carry across limbs. Theorem~\ref{thm:correctness} states
the correctness theorem; the proof is provided in \extdoc.
\begin{theorem}[Correctness of DoT-Addition]
  \label{thm:correctness}
  Suppose two large integers $A$ and $B$,
  $A = \sum_{i=0}^{m-1} A_i X^i = A_{m-1}X^{m-1} + \dots + A_1X^1 +
  A_0X^0$,  $B = \sum_{j=0}^{m-1} B_j X^j = B_{m-1}X^{m-1} + \dots +
  B_1X^1 + B_0X^0$,  where $X = 2^k$ (where $k$ is the bit
  size of a limb), the number of limbs for the two integers is $m$,
  and $0 \le A_i, B_j < X$. \\
If $\textsc{DoT-Add-Words}(A,B,m) = (S, c_{\text{out}})$, then
 $S = \sum_{i=0}^{m-1} (r_i X^i)$ such that 
$r_i = (A_i + B_i + c_i) \pmod{X}, c_0 = 0$ and
$c_{i+1} = \left\lfloor (A_i + B_i + c_i)/X \right\rfloor$, 
and $c_{\text{out}} = c_m$.
\end{theorem}

\subsubsection{Subtraction.} DoT subtraction mirrors addition, replacing carries
with borrows. Limbs are subtracted in parallel in Phase 1; borrow bits are
generated and aligned in Phase 2; borrows are propagated in Phase 3; and Phase
4 handles the case where a limb that reached zero in Phase 1 receives a borrow
in Phase 3, triggering further borrow generation. The probability of Phase 4 is
similar to the addition case.

\subsection{DoT Multiplication} 
\label{sec:urdhva}

Large-number multiplication presents a different challenge of computing and
accumulating partial products efficiently without serialization.  To overcome
the challenge, we adopt the \emph{vertical and crosswise} multiplication
technique~\cite{multStrategies, asmj2006, archiveVedicMathematics}, which
performs multiplication \emph{column-by-column}, i.e., for each output position
$c$, all cross-products $A_i \times B_j$, such that $i+j=c$, are computed and
summed together, e.g., for $c=0$, we have the cross-product
$A_0X^0. B_0X^0$, for $c=1$, we have the cross-product
$(A_1X^1. B_0X^0 + A_0X^0. B_1X^1)$ and so on. 

\begin{figure}[!bp]
    \centering
    \includegraphics[width=0.45\textwidth]{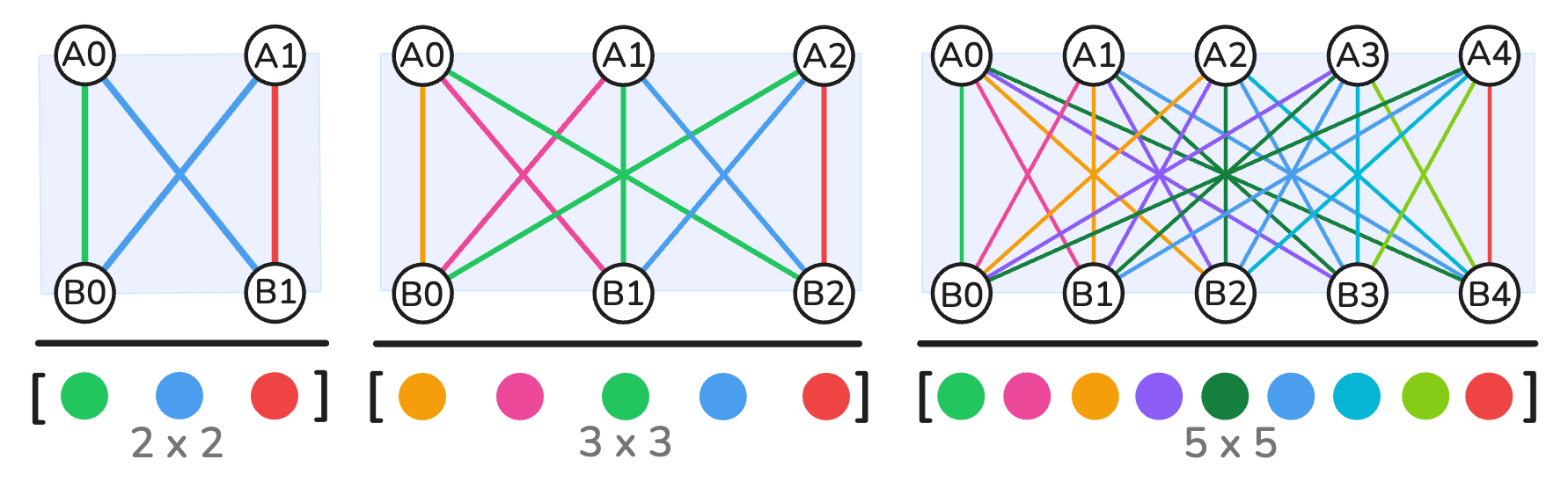}
    \caption{\small{``Vertical and Crosswise'' partial product
organization for $2{\times}2$, $3{\times}3$, and $5{\times}5$ limb
multiplication. Each line represents one cross-product $A_i \times B_j$; lines
of the same color belong to output column $c = i{+}j$ and are summed together.
A $2m{-}1$-column structure exposes all $m^2$ partial products as independent
computations.}}
    \Description{Diagram showing three multiplication grids for 2x2, 3x3, and
5x5 limb operands. In each grid, colored lines connect pairs of input limbs
A\_i and B\_j whose indices sum to the same output column c=i+j. Lines of the
same color represent independent cross-products that are accumulated together.
The number of output columns is 2m-1, and all m-squared cross-products can be
computed in parallel before any carry adjustment.}
    \label{fig:UT-Multiplication}
\end{figure}

Crucially, all cross-products are independent of one another, so they can all
be computed before any summation or carry adjustment.
Figure~\ref{fig:UT-Multiplication} illustrates this for $2{\times}2$,
$3{\times}3$, and $5{\times}5$ limb multiplications. A final carry-adjustment
pass propagates overflows from each column into the next to produce the final
product. Since all partial products are independent, $w$ limb pairs can be
gathered into SIMD registers per output column and computed in a single
instruction. A horizontal reduction sums the partial products per column,
followed by the carry-adjustment pass. The vertical and crosswise
multiplication is defined as:
\begin{multline*}
\textsc{VnC}(A, B, n, k) = A_0\!\cdot\! B_0 + (A_12^k\!\cdot\! B_0 + A_0\!\cdot\! B_12^k) + \cdots \\
{} + (A_{n-1}2^{k(n-1)}\!\cdot\! B_{n-2}2^{k(n-2)} + A_{n-2}2^{k(n-2)}\!\cdot\! B_{n-1}2^{k(n-1)}) \\
{} + (A_{n-1}2^{k(n-1)}\!\cdot\! B_{n-1}2^{k(n-1)}).
\end{multline*}

\textsc{DoT Multiplication} algorithm comprises five phases. The pseudocode is
provided in Algorithm~\ref{algo:ut_mul}. \textbf{Phase~1} gathers the limb pairs for
each output column $c=i+j$ into flat arrays $M_a$ and $M_b$. 
\textbf{Phase~2} computes every partial product in
parallel using SIMD vector multiplies with a zero accumulator, eliminating the
serialized FMA chain of prior IFMA routines~\cite{gueron2016accelerating} and
fully utilizing both IFMA execution ports. \textbf{Phase~3} propagates each high half to
the next column's accumulator via a single flat-index traversal. \textbf{Phase~4}
performs a horizontal reduction per column, and \textbf{Phase~5} executes the
carry-adjustment pass, which performs a single sequential pass over the $2m-1$
column sums.

\begin{algorithm}[!t]
  \footnotesize
  \SetAlgoLined
  \caption{DoT Multiplication}
  \label{algo:ut_mul}
  \KwData{Arrays $A$, $B$ each with $m$ limbs of $k$ bits}
  \KwResult{Product array $P = A \times B$}
  \BlankLine
  \textbf{Function:} \textsc{DoT-Mul-Words}$(A, B, m, k)$\\
  \Begin{
    $M_a[0 \dots m^2-1]$, $M_b[0 \dots m^2-1]$, $P\_lo[0 \dots m^2-1]$, $P\_hi[0 \dots m^2-1]$, $col\_lo[0 \dots 2m-1] \leftarrow 0$\; \label{alg:init}

    $idx\_a \leftarrow 0$, $idx\_b \leftarrow 0$\; \label{alg:initidx}

    \textcolor{gray}{\tcp{Phase 1: Gather limbs into columns}}
    \For{$c \leftarrow 0$ \KwTo $2m - 2$}{ \label{alg:col_loop}
      \For{all pairs $(i, j)$ such that $i + j = c$}{ \label{alg:pair_loop}
        $M_a[idx\_a++] \leftarrow A[i]$\; \label{alg:ga}
        $M_b[idx\_b++] \leftarrow B[j]$\; \label{alg:gb}
      }
    }

    \textcolor{gray}{\tcp{Phase 2: Compute partial products with SIMD}}
    \For{$i \leftarrow 0$ \KwTo $m^2 - 1$ \textbf{step} $w$}{
      $\mathbf{a} \leftarrow \texttt{simd\_load}(w \text{ limbs from } M_a[i] \text{ to } M_a[i+w-1])$\; \label{alg:lda}
      $\mathbf{b} \leftarrow \texttt{simd\_load}(w \text{ limbs from } M_b[i] \text{ to } M_b[i+w-1])$\; \label{alg:ldb}
      $\update{\mathbf{r\_lo}} \text{\remove{$\mathbf{p\_lo}$}} \leftarrow \texttt{simd\_mul\_lo}(\mathbf{a}, \mathbf{b})$\; \label{alg:lp}
      $\update{\mathbf{r\_hi}} \text{\remove{$\mathbf{p\_hi}$}} \leftarrow \texttt{simd\_mul\_hi}(\mathbf{a}, \mathbf{b})$\; \label{alg:hp}
      $\texttt{simd\_store}(\update{\mathbf{r\_lo}} \text{\remove{$\mathbf{p\_lo}$}} \text{ into } P\_lo[i] \text{ to } P\_lo[i+w-1])$\; \label{alg:stlo}
      $\texttt{simd\_store}(\update{\mathbf{r\_hi}} \text{\remove{$\mathbf{p\_hi}$}} \text{ into } P\_hi[i] \text{ to } P\_hi[i+w-1])$\; \label{alg:sthi}
    }

    \textcolor{gray}{\tcp{Phase 3: Align hi parts to neighboring columns}}
    $idx \leftarrow 0$\;
    \For{$c \leftarrow 0$ \KwTo $2m-2$}{ \label{alg:align_loop}
      $num\_pairs \leftarrow \min(c+1, m, 2m-1-c)$\; \label{alg:numpairs}
      \For{$p \leftarrow 0$ \KwTo $num\_pairs-1$}{ \label{alg:pair_align_loop}
        $t \leftarrow idx + p$\; \label{alg:pair_idx}
        $col\_lo[c]   \leftarrow col\_lo[c]   + P\_lo[t]$\; \label{alg:col_lo}
        $col\_lo[c+1] \leftarrow col\_lo[c+1] + P\_hi[t]$\; \label{alg:col_hi}
      }
      $idx \leftarrow idx + num\_pairs$\;
    }

    \textcolor{gray}{\tcp{Phase 4: Reduce partial products in each column}}
    \For{$c \leftarrow 0$ \KwTo $2m-1$}{
      $P[c] \leftarrow col\_lo[c]$\; \label{alg:col_init}
    }

    \textcolor{gray}{\tcp{Phase 5: Carry-over adjustment}}
    $\mathit{carry} \leftarrow 0$\;
    \For{$c \leftarrow 0$ \KwTo $2m-1$}{ \label{alg:carry_loop}
      $P[c] \leftarrow P[c] + \mathit{carry}$\;
      $\mathit{carry} \leftarrow P[c] \gg k$\;
      $P[c] \leftarrow P[c] \bmod 2^k$\; \label{alg:carry_adjust}
    }
    \If{$\mathit{carry} > 0$}{
      Prepend $\mathit{carry}$ as an extra limb to the result\; \label{alg:carry_prepend}
    }
    \Return{$P$}\;
  }
\end{algorithm}
We walk through Algorithm~\ref{algo:ut_mul} with $m{=}5$ and $k{=}52$: two
260-bit operands $A{=}A_4\ldots A_0$ and $B{=}B_4\ldots B_0$ produce a 520-bit
product $P$ across $2m{-}1{=}9$ columns with $1,2,3,4,5,4,3,2,1$ pairs per
column respectively (the triangular pattern in
Figure~\ref{fig:UT-Multiplication}). \textsc{DoT-Mul-Words} is invoked on
$A_0$--$A_4$, $B_0$--$B_4$ and returns a 10-limb product.

\noindent\textbf{Phase 1.} The gather loop
(lines~\ref{alg:col_loop}--\ref{alg:gb}) enumerates each output column $c$
and every $(i,j)$ with $i{+}j{=}c$, packing operands into flat arrays $M_a$ and
$M_b$ in column-wise groups. For $c{=}4$, it emits five pairs:
${(A_0,B_4),(A_1,B_3),(A_2,B_2),(A_3,B_1),(A_4,B_0)}$; after the loop, $M_a$ and $M_b$ each hold 25
entries with every cross-product exposed as an independent operand pair with no
dependencies between them. Intuitively, by restructuring the computation around
the output columns, all pairs contributing to a column can be computed in
parallel in Phase 2 directly without shuffling or permuting. 

\noindent\textbf{Phase 2.} The 25 entries are consumed in chunks of $w{=}8$ via
SIMD multiplies (lines~\ref{alg:lda}--\ref{alg:sthi}): three full iterations
cover 24 pairs plus one trailing pair handled as the remainder step. Each
vector issues \texttt{simd\_mul\_lo/simd\_mul\_hi} against a \emph{zero}
accumulator, producing eight independent low halves and eight high halves per
call. Because no iteration reads another's accumulator, both IFMA ports stay
busy in parallel with no RAW chain through the sequence, unlike Gueron et al.'s
shared-accumulator chain~\cite{gueron2016accelerating}. 

\noindent\textbf{Phase 3.} The alignment pass
(lines~\ref{alg:align_loop}--\ref{alg:col_hi}) folds each product into its
column: $P\_lo[t]$ is added into $col\_lo[c]$ while $P\_hi[t]$ is promoted into
$col\_lo[c{+}1]$. The pair $(2,3)$, for instance, deposits its low half at
column 5 and its high half at column 6. This phase is a simple flat loop with
no dependencies between iterations, and thus can be parallelized with SIMD or
left as a scalar pass. The result is that each column $c$ holds the full sum of
its cross-products, albeit with potential overflow above $2^k$ that must be
handled in Phase 5.

\begin{table*}[!htbp]
\centering
\footnotesize
\resizebox{\linewidth}{!}{%
\begin{tabular*}{\linewidth}{@{\extracolsep{\fill}}l r r  l r  l r@{}}
\toprule
\multicolumn{3}{c}{\textbf{DoT Addition ($m=8$, 64-bit limbs)}} &
\multicolumn{2}{c}{\textbf{DoT Mul ($m=4$, 64-bit limbs)}} &
\multicolumn{2}{c}{\textbf{DoT Mul ($m=5$, 52-bit limbs)}} \\
\cmidrule(r){1-3} \cmidrule(lr){4-5} \cmidrule(l){6-7}
Operations & Random \% & Pathological \% & Operations & \% & Operations & \% \\
\midrule
Load \textit{(Phase 1)}              & 18.7 &  7.3 & Gather \& radix conversion \textit{(Phase 1)} & 22.6 & Gather \textit{(Phase 1)}                   &  8.0 \\
Add \textit{(Phase 1)}               & 13.8 &  5.4 & IFMA compute \textit{(Phase 2)}               & 19.4 & IFMA compute  \textit{(Phase 2)}            & 22.9 \\
Generate carry \textit{(Phase 2)}    & 25.8 & 10.1 & Align hi parts \textit{(Phase 3)}              & 23.5 & Align hi parts \textit{(Phase 3)}            & 27.3 \\
Add carry \textit{(Phase 3)}         & 20.7 &  8.1 & Reduce columns \textit{(Phase 4)}             & 24.4 & Reduce columns \textit{(Phase 4)}            & 31.5 \\
Store \& check \textit{(Phase 3)}    & 21.0 &  8.2 & Carry-over \& store \textit{(Phase 5)}              & 10.1 &  Carry-over \& store \textit{(Phase 5)}          & 10.3 \\
Overflow handling \textit{(Phase 4)}          &  --- & 60.9 &                             &      &                           &      \\
\midrule
Carry/Add ratio                      & 4.9 & 16.2 &                             &  --- &                           &  --- \\
\bottomrule
\end{tabular*}%
}
\caption{\small{Phase-wise timing breakdown (\%) of DoT for 512-bit addition ($m{=}8$, 64-bit limbs) and 256-bit multiplication ($m{=}4$ and $m{=}5$ limbs). For addition, the random column shows percentages for random inputs (Phase~4 never fires); the pathological column shows percentages for pathological inputs (Phase~4 fires).}}
\label{tab:ablation-dot}
\end{table*}

\noindent\textbf{Phase 4.} Each column's accumulated value is moved into the
output array $P[c]$ (line~\ref{alg:col_init}).

\noindent\textbf{Phase 5.} A single scalar pass over the $2m-1{=} 9$ column
sums (lines~\ref{alg:carry_loop}--\ref{alg:carry_adjust}) propagates each
column's overflow above $2^{52}$ into the next, yielding the final 9-limb
product and an optional overflow limb (line~\ref{alg:carry_prepend}). Thus,
only Phase 5 is sequential; Phases~1--4 are fully data-parallel and costlier
serial work is limited to this short tail pass.

\subsubsection{Proof of correctness} The correctness of vertical and
crosswise multiplication is shown as
Theorem~\ref{thm:mul-correct}. Algorithm~\ref{algo:ut_mul} 
is an adaptation of the vertical and crosswise multiplication approach 
designed to work with the SIMD architecture and the available
registers. The proof is provided in \extdoc. 
\begin{theorem}[Correctness of vertical and crosswise multiplication]
  \label{thm:mul-correct}
  Suppose two large integers $A$ and $B$,
  $A = \sum_{i=0}^{m-1} A_i X^i = A_{m-1}X^{m-1} + \dots + A_1X^1 +
  A_0X^0$, $B = \sum_{j=0}^{m-1} B_j X^j = B_{m-1}X^{m-1} + \dots + B_1X^1 + B_0X^0$,
  where $X = 2^k$ (where $k$ is the bit size of a limb), the
  number of limbs for the two integers is $m$,  and $0 \le A_i, B_j < X$. 
  Then, $A \times B = \textsc{VnC}(A, B)$.
\end{theorem}

\subsection{Implementation}

We implement DoT in C using x86-64 SIMD intrinsics~\cite{intelIntrinsics},
trading peak microarchitectural tuning for portability across compilers and
microarchitectures. Addition and subtraction support saturated radix ($k=64$,
base $2^{64}$), using masked load/store intrinsics to handle
any limb counts, and can be adapted for SSE ($w=2$), AVX2 ($w=4$),
and AVX-512 ($w=8$) widths by selecting the corresponding intrinsics.

Multiplication uses two fixed-size implementations,
$5{\times}5$ and $4{\times}4$, both realizing Algorithm~\ref{algo:ut_mul}'s
five phases. The $5{\times}5$ routine operates on unsaturated radix ($k=52$, base $2^{52}$), 
since AVX-512 IFMA instructions operate on 52-bit limbs. 
The $4{\times}4$ routine is designed for compatibility with the saturated radix ($k=64$,
base $2^{64}$) used by GMP 
and OpenSSL, and pays the extra cost of radix conversion packing at entry ($k=64 \to k=52$)
and unpacking ($k=52 \to k=64$) at exit.

We integrate DoT as a static library into \remove{the C pathways of} GMP 6.3.0 (DoT\-MP)
and OpenSSL 3.5.2 (DoTSSL)\update{, in two configurations: \emph{asm} and \emph{generic}}, 
replacing each library's addition and subtraction primitives and
base-case multiplication\footnote{GMP: \texttt{mpn\_add\_n},
\texttt{mpn\_sub\_n}, \texttt{mpn\_mul\_basecase}. \\
OpenSSL:
\texttt{bn\_add\_words}, \texttt{bn\_sub\_words}, \texttt{bn\_mul\_comba4}.}. 
\update{The \emph{asm} builds retain the libraries' native, per-microarchitecture 
hand-tuned assembly wherever DoT is not invoked, while the \emph{generic} builds 
disable all assembly so that remaining routines compile from generic x86-64 C code\footnote{\update{For the \emph{generic} DoTMP and DoTSSL builds, we pin the
Karatsuba/Toom-Cook recursion thresholds so that large multiplications
bottom out in DoT's 4x4 base case. By contrast, baseline GMP and OpenSSL dynamically 
select these thresholds at compile time; across all comparisons, we retain their default
configurations for both \emph{generic} and \emph{asm}.
The \emph{asm} builds of DoTMP and DoTSSL retain the default thresholds, 
as their recursion is tightly coupled to hand-written assembly base cases: there, unless \texttt{mul()}
is called directly, DoT's 4x4 mul will not be invoked by other dependent higher operations.}}.
Since DoT is integrated as a static library, its own routines are
identical regardless of whether the surrounding library is built with generic or assembly; only
other code paths/operations that DoT does \emph{not} replace differ between the
two builds.}\remove{;
higher-level recursive multiplications automatically
invoke DoT as the base case.} 
Unaligned SIMD load/store intrinsics throughout
preserve compatibility with both libraries' non-aligned memory layouts.

We also test DoT on the integrated libraries' test suites,
which cover a wide range of operand sizes and values, and a hefty set of higher-level applications
(e.g., RSA, DSA) that rely on the primitives.
In total, DoT consists of $1013$ lines of C code, with \update{$35$}\remove{$38$} and \update{$99$}\remove{$43$} lines of
integration changes in GMP and OpenSSL, respectively.

\section{Evaluation}
\label{sec:evaluation}

We evaluate DoT along \update{six}\remove{four} axes: (i) reduction in carry\--prepar\-ation overhead
for addition and subtraction, (ii) scaling across SIMD widths, i.e.,  SSE ($w{=}2$), AVX2 ($w{=}4$) and AVX-512 ($w{=}8$), for addition and subtraction, (iii)
performance of base-case multiplication, \update{(iv) speedup of DoTMP and DoTSSL over GMP and OpenSSL (both asm and generic builds) for addition and subtraction, (v) the same comparison for multiplication, and (vi)} \remove{(iv)} impact on higher-level
operations.

We run the experiments on an Intel Xeon Gold 6548Y+ CPU (Emerald Rapids
microarchitecture) with 64 cores (two sockets, 32 cores per socket) @ 2.50 GHz,
256 GB of DRAM, and Ubuntu 22.04.2 LTS. We also run experiments on \update{two other CPUs:}\remove{an} Intel
Xeon Max 9462 (Sapphire Rapids, 64 cores @ 2.70 GHz, same OS and DRAM)\update{, and AMD EPYC 9354P 
(Zen 4 Genoa, 32 cores @ 3.25 GHz, same OS, 188 GiB DRAM)}. We
report results from the 6548Y+ in the main paper and include results from the \update{other CPUs}\remove{9462}
in \extdoc{} \update{for DoT's portability test.}\remove{.} We use \texttt{clang version 14.0.0-1ubuntu1.1} with
\texttt{-O2} and \texttt{-march=native} optimization flags to compile DoT.

To evaluate the accuracy and performance of DoT, we employ two sets of test
cases: \emph{random} and \emph{pathological}. Random cases mimic typical usage;
pathological cases target edge scenarios (full carry/borrow propagation,
maxed-out/zero limbs, frequent carries, frequent borrows, and mixed cases). We
use the Mersenne Twister~\cite{matsumoto1998mersenne} seeded with random
integers to generate the test cases. Every test run includes 100,000 random and
1,000 pathological test cases for twelve operand sizes from 512 to 32768 bits. 

Micro-benchmark results are reported as the mean over twenty runs per
experimental setup (a unique combination of operation and operand size), using
the $95\%$ confidence interval. We compute timing and throughput based on
GMPbench's methodology~\cite{gmplibGMPbenchResults} and CPU ticks via the RDTSC
instruction \cite{intel-white-paper-rdtsc}. Speedups, wherever reported, are
measured as the ratio of execution times, which are also verified using CPU
ticks and throughputs. Instruction counts, wherever reported, are measured
using the \texttt{perf\_event\_open} system
call~\cite{man7PerfEventOpen2Linux}. Unless otherwise noted, all compared
methods are evaluated in the same benchmarking harness with identical operand
sets, iteration counts, and measurement methodology (RDTSC and
\texttt{perf\_event\_open}); each implementation is validated against
precomputed expected outputs.

\subsection{Reduction of Carry Preparation Overhead}

Table~\ref{tab:ablation-existing} and Table~\ref{tab:ablation-dot} give the
full phase-wise breakdown for prior approaches and DoT, respectively. For
random inputs (Phase~4 never triggered across our entire test set), DoT
achieves a carry-to-add ratio of $\sim 4.9$, roughly half that of two-level KSA
($\sim 9.1$) and a third of Ren et al.\ ($\sim 12.4$). DoT also runs
$1.85\times$ faster overall, so the absolute carry cycles per limb are even
lower. For pathological inputs where Phase~4 fires, the ratio rises to $16.2$,
still well below naive SIMD ($52.1$).

Figure~\ref{fig:combined_microbench}(a) shows the performance comparison of DoT
(AVX-512), two-level KSA adapted from y-cruncher's approach\update{~\cite{numberworldIntegerAddition}}, and Ren et al.'s
method~\cite{Ren} for addition and subtraction on both randomly generated test
cases. Since Ren et al.\ did not release their implementation, we derived a
fair and AVX-512-intrinsics-based implementation of their \emph{ProposedAdd}
algorithm from \update{their provided pseudocode}\remove{the pseudocode provided in their paper}.

\begin{table}[!htbp]
\centering
\footnotesize
\setlength{\tabcolsep}{3.5pt} 
\begin{tabular*}{\linewidth}{@{\extracolsep{\fill}}lrrrr@{}}
\toprule
\textbf{Implementation} & \textbf{\update{Instr.}\remove{Instructions}} & \textbf{\update{Cycles}\remove{Avg.\ Cycles}} & \textbf{IPC} & \textbf{\update{Speedup}}$^\dagger$ \\
\midrule
\texttt{dot\_mul\_${5\times5}$}                  & 221 & 28.6 & 7.7  & \update{$0.81\times$} \\
\texttt{dot\_mul\_${4\times4}$}                  & 265 & 35.2 & 7.5  & \update{$1.00\times$} \\
Gueron \& Krasnov~\cite{gueron2016accelerating} & 345 & 81.5 & 4.2  & \update{$2.31\times$} \\
OpenSSL \texttt{BN\_mul} \update{(generic)}          & \update{739}\remove{655} & 68.7 & \update{10.8}\remove{9.5} & \update{$1.95\times$} \\
\update{OpenSSL \texttt{BN\_mul (asm)}}          & \update{436} & \update{42.6} & \update{10.2} & \update{$1.21\times$} \\
GMP \texttt{mpz\_mul} \update{(generic)}             & 870 & 88.4 & 9.8  & \update{$2.51\times$} \\
\update{GMP \texttt{mpz\_mul (asm)}}             & \update{304} & \update{29.1} & \update{10.4} & \update{$0.83\times$} \\
\update{OpenSSL \texttt{amm\_ifma}$^*$}          & \update{731} & \update{153.0} & \update{4.7}  & \update{$4.35\times$} \\
\bottomrule
\end{tabular*}
\caption{\small{Instruction \update{count}\remove{counts}, average cycle \update{count}\remove{counts}, and IPC for 256-bit multiplication. $^\dagger$\update{Speedup is computed as $\text{Cycles}_{\text{base}} / \text{Cycles}_{\text{dot\_mul\_4x4}}$} $^*$\update{Extracted from OpenSSL (computes almost Montgomery multiplication using AVX-512IFMA).}\remove{Cycle counts were obtained via \texttt{RDTSC}}}}
\label{tab:256_mul}
\end{table}

\remove{\noindent\textbf{Randomly Generated Test Cases.} }DoT addition achieves a speedup of $1\times$ to $1.9\times$. For subtraction,
DoT achieves a speedup of $0.9\times$ to $1.9\times$. DoT, on average, reduces
instruction count by $15\%$ for addition and $16\%$ for subtraction. Compared
to Ren et al., DoT achieves a speedup of $1.4\times$ to $2.2\times$ and
$1.1\times$ to $1.8\times$ for addition and subtraction, respectively. DoT
reduces the instruction count by an average of $28.8\%$ for addition and
$17.3\%$ for subtraction, respectively. For addition and subtraction, DoT shows
modest speedups at small operand sizes due to fixed function-call overhead. As
operands grow, the overhead amortizes, and the lower per-limb computation cost
dominates.  

\subsubsection{Pathological Test Cases} DoT achieves a speedup of $0.7\times$ to $2\times$
against two-level KSA and $0.8\times$ to $1.9\times$ against the work by Ren et al.~\cite{Ren}. Due to space constraints, the plot
comparing the performance for pathological test cases is shown in
\appsec{sec:appendix-er}.

\subsection{DoT's Performance over SIMD Widths}

We now compare the performance of three DoT SIMD variants 
with the scalar non-SIMD variant for addition and subtraction.
Figure~\ref{fig:combined_microbench}(b) shows these results. For comparing the
baseline scalar, we have used \texttt{\_addcarryx\allowbreak\_u64} and
\texttt{\_subborrow\allowbreak\_u64} intrinsics. With SSE, we observe
a geometric mean speedup of $0.7\times$ for both addition and subtraction; at
lower SIMD width, the overhead outweighs the parallelism benefits. With AVX2, the geometric mean speedups increase to $1.2\times$ for addition and
subtraction, respectively. The AVX-512 variant achieves the highest
speedups, with geometric mean speedups of $1.8\times$ for addition and
subtraction, respectively. As stated earlier, for the lower range of the
operands ($512$--$4096$ bits), the speedups are more modest, while for larger
operands ($6144$--$32768$ bits), the performance gains are more significant and
consistent.

SSE and AVX2 \emph{increase} instruction count by $125\%$ and $25\%$ for
addition and $129\%$ and $29\%$ for subtraction over scalar, respectively.
AVX-512 \emph{reduces} instruction count by $25\%$ for addition and $22\%$ for
subtraction, reaching $33\%$ and $30\%$ reductions for the largest operands.
These results reveal that SIMD parallelism and not instruction compactness
drives the speedup. The benefit of AVX-512 is that each SIMD instruction operates on $8$ limbs at once.
SSE and AVX2 do not have enough lanes to amortize the carry-management
overhead; AVX-512 crosses the threshold where lane parallelism
dominates the overhead.

\begin{figure*}[t]
    \centering
    \includegraphics[width=\linewidth]{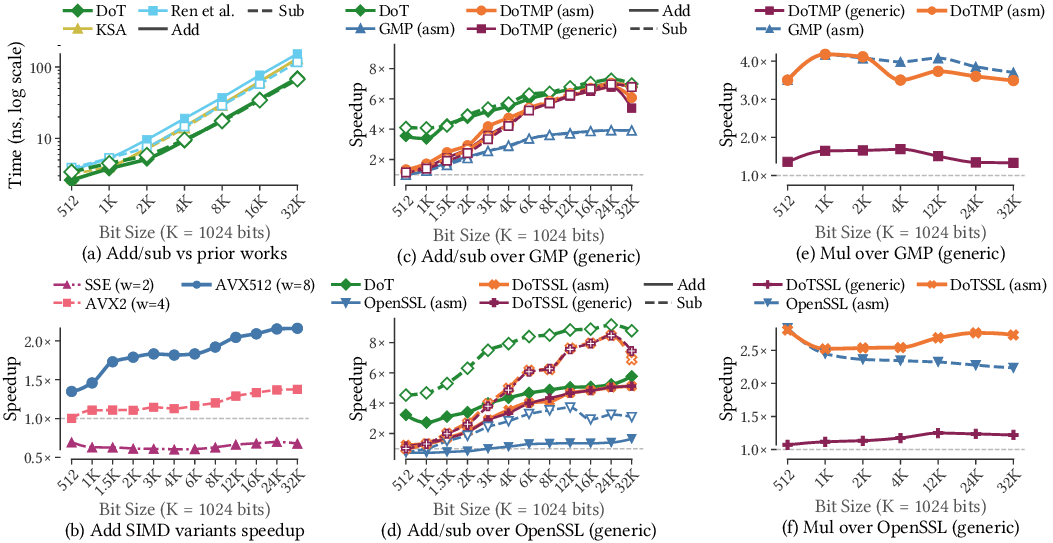}
    \caption{\small{\update{DoT micro-benchmarks (512–32,768-bit operands). (a) Execution time (log scale) vs. two-level KSA and Ren et al. (b) DoT SIMD speedup ($w{=}2,4,8$) over scalar add-with-carry. (c) Addition and subtraction speedup over generic GMP for GMP (asm), DoTMP (generic/asm), and standalone DoT. (d) Addition and subtraction speedup over generic OpenSSL for OpenSSL (asm), DoTSSL (generic/asm), and standalone DoT. (e) Multiplication speedup over generic GMP for GMP (asm) and DoTMP (asm). (f) Multiplication speedup over generic OpenSSL for OpenSSL (asm) and DoTSSL (asm).}\remove{Micro-benchmark evaluation of DoT across four axes. (a) Execution time (log scale) of DoT (AVX-512), two-level KSA, and Ren et al.\ for addition and subtraction across 512--32768-bit random operands. (b) Execution time speedup of DoT SIMD variants ($w{=}2,4,8$) over scalar add-with-carry. (c) Execution time speedup of DoTMP over GMP and DoTSSL over OpenSSL for addition and subtraction. (d) Execution time speedup of DoTMP over GMP and DoTSSL over OpenSSL for multiplication.}}}
    \label{fig:combined_microbench}
\end{figure*}

\subsection{256-bit Multiplication with DoT} \label{sec:eval-mul-256}

\subsubsection{Against prior SIMD work}
\remove{Table~\ref{tab:256_mul} compares our 256-bit base case multiplication routines ($4\times4$ and $5\times5$) against GMP, OpenSSL, and Gueron and Krasnov's AVX-512IFMA implementation~\cite{gueron2016accelerating}. DoT's $4\times4$ routine achieves a $2.31\times$ speedup over Gueron and Krasnov's method.} 
DoT's $4\times4$ routine issues $23\%$ fewer instructions
than Gueron and Krasnov's \update{AVX-512IFMA multiplication~\cite{gueron2016accelerating}}, achieving a $2.31\times$ speedup, driven primarily
by lower cycle count and higher IPC: $7.5$ for DoT compared to their $4.2$. Despite being faster, DoT spends significantly less time in the IFMA compute
phase ($19.4\%$) as compared to AVX-512IFMA method ($36.7\%$)
(Table~\ref{tab:ablation-dot}): DoT's reorganization of the multiplication
routine allows it to better utilize the CPU's execution resources and reduce
dependency chains, leading to higher IPC and overall performance gains.
\update{While Gueron and Krasnov speculated $2\times$ speedups
against GMP for 1024-bit and higher operands based solely on instruction counts,
at 256-bit operands, we observe their method to be similar to GMP (generic) in terms of execution time, 
%
despite issuing $\sim 2.5 \times$ fewer instructions; instruction count alone does not predict SIMD performance here. OpenSSL's AVX-512IFMA implementation~\cite{ghosslrsaz4kavxifma} for modular multiplication remains the slowest of all, primarily due to radix-conversion overheads.}\remove{While
Gueron and Krasnov target their multiplication routine for 1024-bit and above
operand sizes, we included their method for comparison since it is the only
publicly available AVX-512-based multiplication implementation;}

\subsubsection{Against GMP and OpenSSL} Over the \emph{generic} builds, \remove{Against OpenSSL and GMP,} the  \update{$4\times4$ }routine
delivers speedups of $1.95\times$ \update{over OpenSSL and $2.51\times$ over GMP}\remove{and $2.51\times$, respectively}. \update{Against the
\emph{asm} builds, DoT remains $1.21\times$ faster than OpenSSL, while
GMP's hand-tuned assembly is faster than DoT ($29.1$ vs.\ $35.2$ cycles). This is primarily due to radix-conversion overheads:} since DoT's $4\times4$ routine converts the 64-bit limbs to a 52-bit
representation to better fit the IFMA instruction's 52-bit operand limit, it
incurs some overhead for the radix conversion and alignment of the hi parts,
which is reflected in the timing breakdown in Table~\ref{tab:ablation-dot}.
\update{DoT's $5\times5$ routine, which avoids the $4\times4$ routine's radix-conversion overhead,
matches GMP's assembly outright ($28.6$ vs.\ $29.1$ cycles).} The
primary reason for implementing the $4\times4$ routine is to have direct API compatibility with the 256-bit multiplication base cases in GMP and OpenSSL,
which use 64-bit limb representation.

Other implementations~\cite{Edamatsu0,Edamatsu1,Edamatsu2} were not publicly
available and their papers did not provide sufficient details to reimplement
their algorithms for a fair comparison. Notably, all of the prior works on SIMD
multiplication keep GMP as the baseline for comparison, and also observe
speedups only beyond 1024-bit operands.

\begin{figure*}[t]
    \centering
    \includegraphics[width=\linewidth]{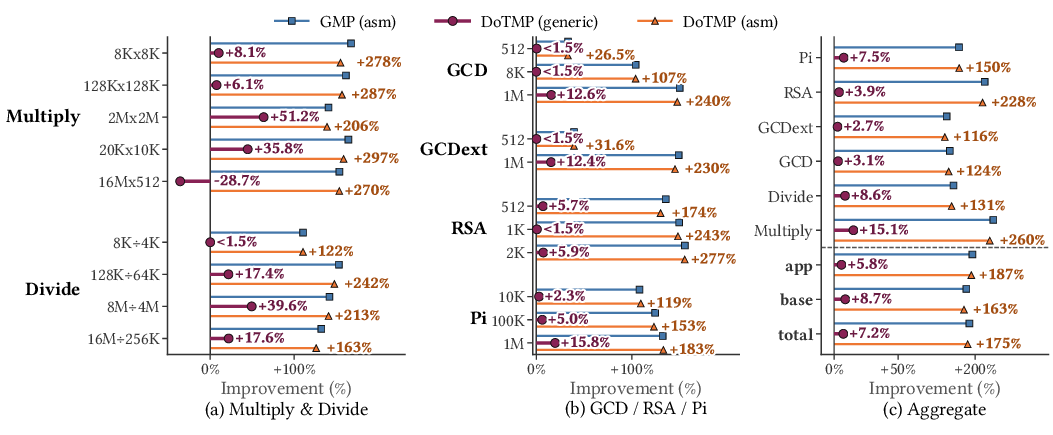}
    \caption{\small{DoTMP \update{and GMP (asm)} score (throughput) improvement over GMP \update{(generic)} in GMPbench. 
    }}
    \label{fig:gmpbench_c}
\end{figure*}

\subsection{DoT-integrated GMP and OpenSSL} \label{sec:eval-integrated}

\update{Figures~\ref{fig:combined_microbench}(c)--(f) compare the execution time speedup of standalone DoT, DoTMP, and DoTSSL across 512- to 32768-bit operands, over generic builds of GMP and OpenSSL. Two observations are common to both libraries. First, with DoT on a standalone basis, we observe higher speedups for subtraction since generic subtraction of GMP/OpenSSL spends more time than addition. Second, in GMP/OpenSSL, subtraction first compares operand magnitudes and computes absolute differences of the operands, returning a sign, which is sequential; SIMD does not accelerate it. Thus, for smaller operand ranges, DoTMP and DoTSSL observe lower speedups for subtraction than addition, but at larger operand sizes, speedups of subtraction converge with addition.}

\subsubsection{DoTMP} For addition and subtraction (Figure~\ref{fig:combined_microbench}(c)), relative to the \emph{generic} build, standalone DoT achieves geometric mean speedups of $5.37\times$ over GMP for addition (reducing instruction counts by $74.7\%$) and $5.66\times$ for subtraction (reducing instructions by $73.7\%$). When integrated into GMP, DoTMP (generic) achieves geometric mean speedups of $3.72\times$ for addition and $3.68\times$ for subtraction, with instruction count reductions of $64.2\%$ and $62.4\%$, respectively. Against the \emph{assembly} build, DoTMP (asm) delivers geometric mean speedups of $1.54\times$ for addition and $1.47\times$ for subtraction over GMP (asm), reducing instruction counts by $34.1\%$ and $31.1\%$, respectively.

\update{For multiplication (Figure~\ref{fig:combined_microbench}(e)), DoTMP (generic) achieves a geometric mean speedup of $1.50\times$ over GMP (ranging $1.33\times$ --- $1.69\times$), reducing instruction counts by $47.8\%$ on average. DoT replaces only the 256-bit base case ($2.5\times$ faster than GMP (generic)), which accounts for $\sim 50\%$ of total cycles at large operand sizes; the rest is spent in higher recursions, and speedups are better when bit-widths are powers of two and take the path to the 256-bit multiplication in GMP's Toom-Cook/FFT variants. Beyond 4096 bits, GMP switches to unequal partitioning for Toom-Cook, which results in fewer 256-bit calls and therefore smaller speedups from DoT. Against the \emph{assembly} build, DoTMP (asm) reaches $0.95\times$ of GMP (asm), primarily due to GMP's faster base-case mul implementation in assembly (Table~\ref{tab:256_mul}).}

\subsubsection{DoTSSL} For addition and subtraction (Figure~\ref{fig:combined_microbench}(d)), against the \emph{generic} build, standalone DoT achieves geometric mean speedups of $4.18\times$ over OpenSSL for addition (reducing instruction counts by $65.0\%$) and $7.19\times$ for subtraction (reducing instructions by $72.4\%$). When integrated, DoTSSL (generic) achieves geometric mean speedups of $3.00\times$ for addition and $4.02\times$ for subtraction, with instruction count reductions of $49.2\%$ and $53.2\%$, respectively. Over the \emph{assembly} build, DoTSSL (asm) achieves geometric mean speedups of $2.78\times$ for addition and $1.78\times$ for subtraction over OpenSSL (asm), reducing instruction counts by $41.8\%$ and $36.7\%$.

\update{For multiplication (Figure~\ref{fig:combined_microbench}(f)), DoTSSL (generic) achieves a $1.17\times$ geometric mean speedup over OpenSSL (generic) (ranging $1.07\times$ --- $1.25\times$), reducing instruction counts by $36.1\%$. 
DoTSSL (asm) mul is $1.11\times$ faster than OpenSSL (asm).}

\update{For a SIMD comparative baseline, we also evaluate Open\-SSL's Montgomery multiplication routine (Table \ref{tab:256_mul}) as a drop-in base-case multiplier for OpenSSL's \texttt{BN\_mul}. However, we observe severe performance degradation: $2.2\times$ slower than OpenSSL's generic mul and up to $6.7\times$ slower than the asm mul for operands up to 8192-bits.}


\begin{takeaway*}
\update{By algorithmic restructuring around independent, data-parallel operations, DoT leverages the acceleration of SIMD, and provides substantial speedups for addition and subtraction over state-of-the-art libraries, while matching the performance of hand-tuned assembly multiplication.}
\end{takeaway*}

\begin{figure*}[!t]
  \centering
  \includegraphics[width=\linewidth]{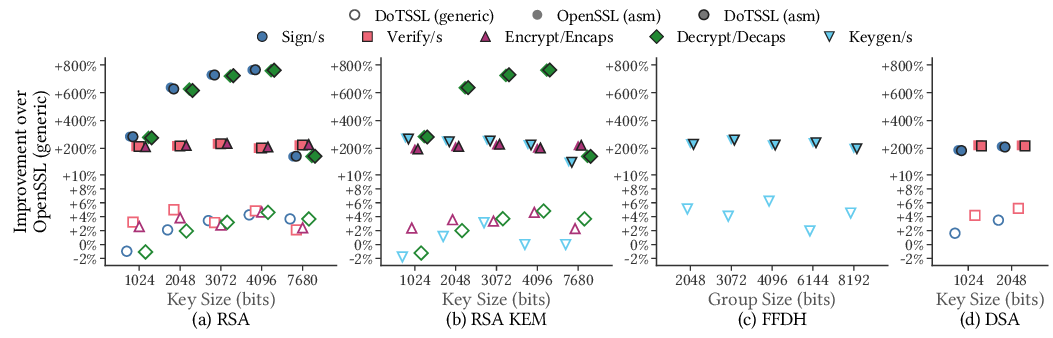}
  \caption{\small{DoTSSL \update{(generic/asm), OpenSSL (asm)} throughput improvement (\%) over OpenSSL \update{(generic)} for RSA (sign/verify/encrypt/decrypt), RSA KEM (encaps/decaps), FFDH (keygen), and DSA (sign/verify) across standard key and group sizes.}}
  \label{fig:openssl_speed}
  \label{fig:openssl_speed_combined}
\end{figure*}

\subsection{DoT's Impact on Higher-Level Applications} \label{sec:eval-apps}

\subsubsection{GMPbench} \update{Figure~\ref{fig:gmpbench_c} compares DoTMP and GMP (asm) across the GMPbench suite against GMP (generic).
DoTMP (generic) improves the overall GMPbench score by $7.2\%$ (base: $8.7\%$, app: $5.8\%$). 
Aggregate sub-scores improve across all benchmark categories: multiply improves by $15.1\%$ (up to $51.2\%$ for 2M$\times$2M-bit multiplication), 
divide by $8.6\%$ (up to $39.6\%$ for 8M$\div$4M-bit division), $\pi$ computation by $7.6\%$ (up to $15.8\%$ for 1M digits),
RSA by $3.9\%$ (up to $5.9\%$ on RSA-2048), and GCD/GCDext by $3.1\%$/$2.7\%$ (up to $12.6\%$ at 1M bits).
However, for a few operand sizes we observe regression (e.g., $-28.7\%$ for 16M$\times$512) due to lowered recursion thresholds (Section~\ref{sec:design}).
DoTMP (asm) achieves a $174.8\%$ score improvement over GMP (generic), tracking closely with GMP's assembly build ($180.2\%$). 
While DoT significantly accelerates addition/subtraction operations over GMP's assembly, the 256-bit multiplication base-case latency gap
(Section~\ref{sec:eval-mul-256}) slightly offsets these gains.}\remove{Figure \ref{fig:gmpbench_c} shows DoTMP's
improvement across the GMPbench suite. DoTMP improves the overall GMPbench
score by $7.8\%$. The multiply aggregate improves by $15.3\%$ (up to $48.7\%$
for $2,097,152 \times 2,097,152$-bit operands). Divide aggregate improves by
$8.4\%$ (up to $36.4\%$) even though DoT replaces no division routine, because
GMP's Newton-Raphson divider calls multiplication, and addition/subtraction in tight loops. 
Computation of $\pi$ gains $13.3\%$ (up to $19.3\%$ for 1M digits) because Chudnovsky
leans almost entirely on large-integer multiply and divide. Even GCD gains by
$3.1\%$ (up to $13.7\%$), since GMP's Lehmer-Euclid hybrid bottoms out in large
addition and subtraction. RSA improves by $4.2\%$ (up to $5.5\%$ for 512-bit
RSA). The GMPbench results expose a structural property of multi-precision
libraries: speeding up the primitives is not an isolated issue. A focused
improvement at the bottom of the stack propagates broadly across the suite.}

\subsubsection{OpenSSL \update{speed}} \update{Figure~\ref{fig:openssl_speed} 
compares DoTSSL and OpenSSL (asm) on the speed benchmark across various cryptographic workloads (RSA, RSA KEM, FFDH, and DSA) 
against OpenSSL (generic). DoTSSL (generic) improves throughput across nearly all operations 
and key sizes: RSA gains up to $5.0\%$ across sign, verify, encrypt, and decrypt; 
RSA KEM gains up to $4.8\%$; FFDH key generation improves by $2.0\%$ to $6.2\%$;
and DSA yields up to $5.4\%$. \Appsec{sec:appendix-er} details the latency experiments.
DoTSSL (asm) tracks OpenSSL's native assembly within $\pm 2.1\%$ across all workloads 
(e.g., FFDH within $\pm0.8\%$, RSA sign within $\pm1.6\%$). 
For these OpenSSL workloads, modular exponentiation operations (e.g., Table~\ref{tab:256_mul}'s \texttt{amm\_ifma}) dominates overall execution time.
Consequently, the routines replaced by DoT represent only a small fraction of execution cycles (Figure~\appref{fig:dot_composition} in \appsec{sec:appendix-er}),
leading to more modest speedups than in GMPbench.}

\subsubsection{Other scientific benchmarks} Figure~\ref{fig:app_bench} shows improvements from using DoTMP as the multi-precision backend for PARI/GP, MPFR, and Boost.Multiprecision. On PARI/\-GP, DoTMP (generic) accelerates \texttt{addii} by $578\%$ (10K digits) over GMP (generic), significantly outperforming GMP's native hand-tuned assembly ($292\%$).
In higher-level operations with indirect addition and subtraction invocations, DoTMP (generic) improves PARI \texttt{sqri} ($16.6\%$), \texttt{ratlift} ($5.8\%$--$10.0\%$), and \texttt{divii} ($3.9\%$--$8.9\%$). On MPFR, DoTMP (generic) improves \texttt{sub} by $13.4\%$ and \texttt{add} by $6.8\%$; DoTMP (asm) matches GMP (asm) within $0.2\%$--$1.8\%$ across other arithmetic routines (\texttt{div}, \texttt{exp}, \texttt{asin}, \texttt{acos}).
On Boost Multiprecision, DoTMP (generic) improves \texttt{gmp\_int +} by $194.3\%$ (2K), \texttt{gmp\_int -} by $101.3\%$ (2K),  \texttt{gmp\_\-float -} by $35.2\%$, and \texttt{gmp\_int *} by $52.8\%$. DoTMP (asm) achieves $122.7\%$ (2K) on \texttt{gmp\_int +}, comfortably surpassing GMP's native assembly ($89.1\%$).

\begin{figure*}[!t]
    \centering
    \includegraphics[width=\linewidth]{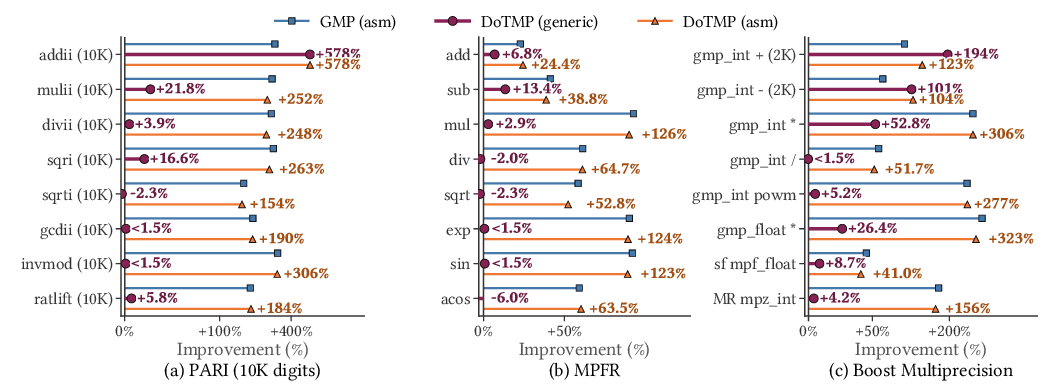}
    \caption{\small{\update{DoTMP's performance improvement over GMP (generic) for PARI/GP, MPFR, and Boost application benchmarks, each also compared against GMP (asm) and DoTMP (asm). 
    }}}
    \label{fig:app_bench}
\end{figure*}



\begin{takeaway*}
\update{For multi-precision workloads, complex mathematical computations depend fundamentally on primitive operations like addition or multiplication. Thus, accelerating the primitive operations cascades into broad performance gains across higher-level operations without modifying the high-level algorithmic logic.}
\end{takeaway*}

\section{Discussion}\label{sec:discussion}


\subsection{Portability} 
With portability also being a focus of our design, DoT is written in C with AVX-512 intrinsics and the performance benefits that we observe on the Xeon 6548Y+ CPU (microarch: Emerald Rapids) are portable across supporting microarchitectures: Xeon 9462 (Sapphire Rapids) and AMD EPYC 9354P (Zen 4) (see \extdoc); the same code runs as-is without any changes, trading peak per-microarchitecture performance for portability, similar to some of the other use-case libraries \cite{ermig1979SimdLibrary,BLIS1,eigenweb} that are written with intrinsics. However, we believe that with careful engineering and with the use of assembly, we can bridge the performance gap of multiplication between GMP's assembly implementation and DoT's base-case multiplication, but doing so requires complex per-microarchitecture hand-tuning.



\subsection{Future Hardware} Existing SIMD support has a few
limitations, around which DoT is designed to ensure the parallelism
offered by SIMD is fully utilized. (i) The scalar add-with-carry (ADC/SBB)
instruction eliminates all the explicit carry-handling steps.  However, in
SIMD, none of the ISAs have native carry propagation across lanes, so all the
approaches, including DoT, manage carries explicitly. ARM SVE2 in particular
exposes native carry-generation instructions
(\texttt{ADCLB}/\texttt{ADCLT}~\cite{armSVE2}) that map directly to DoT Add
Phase~2, suggesting an SVE2 port would make Phase 2 cheaper than the AVX-512
version. (ii) Currently, the SIMD multiplication instruction sets (e.g.,
AVX-512F), which work with 64-bit limbs, only provide the lower 64-bit product
of the 128-bit result. On the other hand, AVX-512IFMA provides a fused
multiply-add instruction that operates on lower 52-bit limbs, and requires two
separate instructions to get the full 104-bit product. 

As the hardware evolves and new features are added to SIMD, such as
add-with-carry chaining or a \texttt{mulx}-like wide-multiply instruction that
gives the full product in one step, approaches like DoT that restructure
algorithms around independent, data-parallel operations would further reduce
the overhead and make SIMD usage more attractive for large-number libraries.
As DoT uses C with AVX-512 intrinsics and no microarchitecture-specific assembly,
future compiler improvements would benefit DoT without code changes.
AVX10~\cite{intelAVX10}, a \update{successor}\remove{succesor} to AVX-512, preserves the intrinsics of
AVX-512, so DoT will run unchanged on upcoming Intel P- and E-cores.

\subsection{Security scope} DoT's addition and subtraction algorithms are data-dependent
 due to the separate slow-path handling of rare carry cascades and may be vulnerable to timing side-channel attacks. 
 For hardened cryptographic deployments, a constant-time variant is an interesting avenue which we leave for future work.


\section{Conclusion}
\label{sec:conclusion}
\update{In this work, we show that SIMD adoption in large-number arithmetic has been blocked by
algorithmic structure, not by hardware. By restructuring computations around independent and
data-parallel operations, we unlock the gains offered by SIMD.}\remove{Large-number addition and subtraction remain unaccelerated by SIMD in every
major multi-precision library, despite being the foundation of cryptography and
high-precision computing.}
We present \texttt{DigitsOnTurbo} (DoT): a 4-phase
SIMD addition/subtraction algorithm that breaks the sequential carry chain by deferring rare
cascades to a slow path, and a ``vertical and crosswise''-based multiplication
routine that exposes all $m^2$ partial products as independent IFMA work. Both
are implemented in portable C with AVX-512 intrinsics, and no hand-tuned assembly\update{, and
offer good acceleration across various x86-64 microarchitectures by leveraging SIMD.}\remove{.}

\remove{On Emerald Rapids ($w{=}8$), DoT add/sub reach $1.85\times$/$1.84\times$ over
scalar add-with-carry and DoT multiplication reaches $2.3\times$ over the
AVX-512IFMA baseline for 256-bit operands. Integrated into GMP (DoTMP) and
OpenSSL (DoTSSL), these gains propagate end-to-end: GMPbench's overall score
improves by $7.8\%$ ($\pi$ up to $+19.3\%$, multiply up to $+48.7\%$), and OpenSSL
throughput for RSA, FFDH, and DSA improves by up to $5.9\%$. The broader
takeaway is that SIMD adoption in large-number arithmetic has been blocked by
algorithmic structure, not by hardware. Restructuring around independent and
data-parallel operations is what unlocks the gains offered by SIMD.}

\begin{acks}
We thank the anonymous reviewers and our shepherd for their valuable feedback. We also thank CloudLab~\cite{cloudlab} for providing us access to the systems used in our experiments. This work was funded in part by the Director's PhD Fellowship at the Indian Institute of Technology Gandhinagar. 
\end{acks}







\appendix
\begingroup
\sloppy
\section{Artifact Appendix}

\subsection{Abstract}
The artifact is archived on Zenodo at \url{https://doi.org/10.5281/zenodo.23202956}. The DoT library can be found in the \texttt{dot-yc-ren/dot} directory, and the reference implementations of the prior work we compare against can be found in \texttt{dot-yc-ren/yc} (two-level KSA), \texttt{dot-yc-ren/ren} (Ren et al.), \texttt{dot-yc-ren/gk-ifma} (Gueron \& Krasnov) and \texttt{dot-yc-ren/ossl-avxifma} (OpenSSL's IFMA AMM kernel). The DoT integrations into GMP~6.3.0 and OpenSSL~3.5.2 can be found in the \texttt{mp-ssl} directory, as patches that \texttt{setup\_libraries.sh} applies to the bundled sources, and the application benchmarks (GMPbench, \texttt{openssl speed}, MPFR, Boost.Multiprecision and PARI/GP) can be found in \texttt{real-workloads}. The measurement data behind every figure in the paper is in \texttt{reference-data}, and the scripts that reproduce the results are in \texttt{ae}.

Since the paper includes a wide set of experiments that can take about two-to-three days per machine, we provide a separate benchmarking script (\texttt{01-key-results.sh}) that reproduces the main-body results with a reduced repetition count in about four to five hours (estimated) on a 28-core Intel Emerald Rapids machine. The workflow to reproduce the remaining results, including on other supporting CPUs like Sapphire Rapids and Zen~4, remains the same and is described in \texttt{README.md}, \texttt{REQUIREMENTS.md} and \texttt{INSTALL.md}.


\subsection{Artifact check-list (meta-information)}


{\small
\begin{itemize}
  \item {\bf Algorithm:} SIMD carry propagation for multi-limb addition and subtraction; IFMA52-based 256-bit multiplication.
  \item {\bf Compilation:} clang 14 (pinned; installer script included), \texttt{-O2 -march=native}.
  \item {\bf Binary:} Only gmpbench's \texttt{gexpr} helper (source included); everything else is built from source.
  \item {\bf Data set:}  Random and pathological operands generated by included scripts; $\sim$110\,MB of bundled measurement data.
  \item {\bf Run-time environment:} Ubuntu 22.04 (24.04 also works); \texttt{sudo} optional (without it, libraries install into \texttt{DOT\_PREFIX}, and runs are timing-only if hardware counters are unreadable); Internet access only to install Ubuntu packages; \texttt{perf} optional.
  \item {\bf Hardware:}  Intel Xeon Gold 6548Y+ CPU, Intel Xeon Max 9462, AMD EPYC 9354P. Any x86-64 CPU with AVX-512F/DQ/BW/VL and AVX512-IFMA (e.g., Intel Sapphire Rapids+ or AMD Zen~4+); $\geq$8 physical cores ($\geq$16 for \texttt{02-full.sh}), $\geq$32\,GB RAM should be able to reproduce.
  \item {\bf Run-time state:} Idle machine; benchmarks pin cores.
  \item {\bf Execution:} Bash drivers \texttt{ae/00-smoke.sh}, \texttt{ae/01-key-results.sh}, \texttt{ae/02-full.sh}; resumable.
  \item {\bf Metrics:} Cycles, time, speedup, instruction count, IPC, throughput, latency.
  \item {\bf Output:} CSV results, machine-readable correctness summary (JSON/CSV), figures, tables.
  \item {\bf How much disk space required (approximately)?:} 335\,MB unpacked; $\sim$60\,GB free for \texttt{01-key-results.sh} (46\,GiB transient peak).
  \item {\bf How much time is needed to prepare workflow (approximately)?: } $\sim$25 minutes.
  \item {\bf How much time is needed to complete experiments (approximately)?:} 7 minutes (\texttt{00-smoke.sh}), $\sim$4--5 hours (\texttt{01-key-results.sh}), $\sim$2--3 days (\texttt{02-full.sh}), all estimated except the smoke test, on 28 physical cores.
  \item {\bf Publicly available?:} Yes.
  \item {\bf Code licenses (if publicly available)?:} MIT; bundled third-party software under its own licenses.
  \item {\bf Data licenses (if publicly available)?:} MIT.
  \item {\bf Workflow automation framework used?:} None (Bash drivers).
  \item {\bf Archived (provide DOI)?:} \url{https://doi.org/10.5281/zenodo.23202956}
\end{itemize}
}

\subsection{Description}

\subsubsection{How to access} 
Download the artifact archive, \texttt{DigitsOnTurbo-c3b68e87.tar.gz}, from
Zenodo (\url{https://doi.org/10.5281/zenodo.23202956}) and unpack it with
\texttt{tar~-xzf}. If unpacking did not preserve execute permissions,
restore them from the artifact root:

\begin{verbatim}
find . -type f \( -name '*.sh' -o -name '*.py' \
    -o -name '*.pl' -o -name runbench \
    -o -name gexpr \) -exec chmod u+x {} +
\end{verbatim}

\subsubsection{Hardware dependencies}
DoT's multiplication requires AVX512-IFMA, which older Skylake-SP and Cascade Lake CPUs lack; without it, the multiplication results cannot be reproduced. The main-body results use an Intel Xeon Gold 6548Y+ (Emerald Rapids).

\subsubsection{Software dependencies}
Ubuntu packages \texttt{build-essential clang m4 zlib1g-dev numactl} and Python~3.10+ with \texttt{tqdm gmpy2 pandas numpy matplotlib seaborn scipy}; \texttt{linux-tools} for \texttt{perf} (optional). The GMP, OpenSSL, PARI/GP, MPFR and Boost sources are bundled and version-pinned, so nothing is downloaded beyond the Ubuntu packages. Instruction counts depend on the compiler, so \texttt{ae/check-hardware.sh} asserts clang 14 and \texttt{ae/install-clang14.sh} installs it alongside the system compiler.

\subsubsection{Data sets}
Benchmark operands (random and pathological dataset) are generated at run time. \texttt{reference-data/} holds the measurements behind every figure for all three CPUs.


\subsection{Installation}
After unpacking, \texttt{./ae/install.sh} installs the packages above, \texttt{perf} and clang~14. Without root, it skips installation if everything is already present, and otherwise prints the \texttt{apt-get} command for an administrator. Then run the preflight check, which verifies CPU features, compiler, Python stack and free disk, and reports the install prefix and the measurement mode (full or timing-only):
{\small\begin{verbatim}
./ae/check-hardware.sh      # expect: PREFLIGHT PASSED
\end{verbatim}}

\paragraph{Without sudo} DoT and the reference libraries install into \texttt{DOT\_PREFIX}, which defaults to \texttt{local/} in the artifact root (any writable directory works), and no \texttt{ldconfig} is needed. If \texttt{kernel.perf\_event\_paranoid} is above 2 and there is no passwordless \texttt{sudo}, runs are \emph{timing-only}: all timings, speedups and figures are still produced; the 256-bit multiplication table reports Cycles and Speedup, with instruction counts and IPC shown as ``--''; and the cycle-attribution figure of \texttt{02-full.sh} is skipped. \texttt{DOT\_TIMING\_ONLY=1} forces this mode. \texttt{INSTALL.md} has the details.

\paragraph{Basic test}
{\small\begin{verbatim}
./ae/00-smoke.sh
\end{verbatim}}
\noindent This builds DoT and the reference libraries, runs the correctness suite, runs a one-iteration micro-benchmark, and regenerates all 42 figures
from the bundled data. It prints \texttt{SMOKE TEST PASSED} after about 7 minutes, with figures in \texttt{results/<timestamp>/figures/}.


\subsection{Experiment workflow}
{\small\begin{verbatim}
nohup ./ae/01-key-results.sh >t1.log 2>&1 &
nohup ./ae/02-full.sh        >t2.log 2>&1 &
\end{verbatim}}

\noindent Both scripts build the GMP/OpenSSL variants (vanilla and DoT-patched, asm and generic; \texttt{02-full.sh} adds a ninth), run their benchmarks,
gather outputs into the per-CPU layout (\texttt{ae/collate-results.sh}), and render the figures. Logs, collated results and figures go to \texttt{results/<timestamp>/}, with per-step logs in \texttt{logs/}. Re-running a script skips completed steps; \texttt{-{}-force} redoes them. \texttt{01-key-results.sh} reduces repetitions to 1 (paper: 20, or 15 for \texttt{openssl speed}) and each \texttt{openssl speed} test to 10\,s (paper: 120\,s), and runs PARI/GP at 5K digits only (paper: 1K, 5K and 10K). The carry-to-add ratios need a separate run of a few minutes, once \texttt{00-smoke.sh} or \texttt{01-key-results.sh} has installed DoT: \texttt{ae/run-ablation.sh}, then \texttt{ae/make-table-ablation.py}. Each ablation binary checks its results against a reference implementation, and exits with an error on a mismatch or if the pathological inputs do not exercise the overflow path.

\subsection{Evaluation and expected results}
Compare the figures in \texttt{results/<timestamp>/figures/} with the paper's by running the benchmark scripts; one can further check the generated CSVs to verify and validate the speedups mentioned in the paper. The ``Evaluating Results'' section of \texttt{README.md} maps each figure and table in the paper to the script and output file that produce it.

{\raggedright
\paragraph{Correctness} The correctness suite writes \texttt{correctness\_summary.json} and \texttt{correctness\_summary.csv} to \texttt{dot-yc-ren/test/correctness-results/}, and these are also collated into \texttt{results/<timestamp>/}. DoT must pass every case. Failures of the prior-work baselines are classified automatically, and all of them are expected: the two-level KSA at operand sizes that are not a multiple of 512 bits, and Ren et al.'s subtraction when the result is negative.\par}

\paragraph{Tolerances} Speedup ratios and geometric means should fall within $\pm$5\%, and end-to-end percentage improvements (GMPbench, \texttt{openssl speed}, application benchmarks, latency) within $\pm$10\% relative. The sign and ordering of every comparison must hold, and the GMPbench regression at 16M$\times$512 must appear as a regression. With clang 14, the 256-bit multiplication table's instruction counts match exactly (they are shown as ``--'' in timing-only runs). \texttt{01-key-results.sh} runs a single repetition, so judge it on sign, ordering and the $\pm$10\% band; exact geometric means need \texttt{02-full.sh}.

\paragraph{Known deviation} A few cycle counts in the bundled Emerald Rapids \texttt{mul-microbench} data come from a different run than the published 256-bit multiplication table ($\sim$3\% run-to-run variance), so the table regenerated from bundled data differs slightly; judge it from a measured run instead, which \texttt{01-key-results.sh} and \texttt{02-full.sh} write to \texttt{results/<timestamp>/figures/tab\_256\_mul.tex}.

\paragraph{Without AVX512-IFMA} \texttt{./ae/plot-all.sh -{}-data reference} regenerates every figure and table from the bundled data in $\sim$45\,s on any machine. This checks the data-to-figure pipeline, not the measurements.


\subsection{Experiment customization}

Environment variables control the compiler (\texttt{CC}), install prefix (\texttt{DOT\_PREFIX}), timing-only mode (\texttt{DOT\_TIMING\_ONLY}), data and output directories (\texttt{DOT\_DATA\_DIR}, \texttt{DOT\_PLOT\_DIR}, \texttt{DOT\_RESULTS\_DIR}), and repetition counts (\texttt{ITERATIONS}, \texttt{NUM\_RUNS}, \texttt{SECONDS\_PER\_TEST}, \texttt{PARI\_DIGITS}, \texttt{APP\_RUNS}). With two or more repetitions, each application benchmark also writes a \texttt{*\_stats.csv} with the mean, standard deviation and a 95\% confidence interval. Operand sizes are set in \texttt{mp-ssl/test/run\_benchmark.sh}. DoT's API mirrors GMP's \texttt{mpn\_add\_n}/\texttt{mpn\_sub\_n} and OpenSSL's \texttt{bn\_add\_words}/\texttt{bn\_mul\_comba4}, so it can be linked into other libraries.

\subsection{Notes}
Every library variant links against one \texttt{libdot.a}, so a stale copy from an earlier run silently affects results; \texttt{check-hardware.sh} warns about this, and \texttt{INSTALL.md} shows how to remove it.



\endgroup



\bibliographystyle{ACM-Reference-Format}
\ifx\pversion\dotmainversion\balance\fi
\bibliography{citations}

\ifx\pversion\dotextendedversion
  \clearpage
  \appendix

\section{Hardware and Library Context for Limb-level Arithmetic}
\label{sec:appendix-bg-hw}

This appendix expands on the scalar hardware support and library-level context that Section~\ref{sec:bg_simd_struggle} only summarizes.

\paragraph{ALU width and word-local optimizations.}
On most CPUs the arithmetic logic unit (ALU) handles fixed-size operands of 32 or 64 bits. Within a single word, hardware designs such as carry-lookahead~\cite{Macsorley61} and carry-select logic~\cite{carrySelect}, as well as parallel-prefix trees (Kogge-Stone~\cite{KoggeStone}, Sklansky~\cite{sklansky1960conditional}, Brent-Kung~\cite{brent1982regular}), reduce the bit-level carry delay from ripple-carry $O(n)$ toward prefix-style $O(\log n)$~\cite{hennessy2011computer}. Modern integer pipelines hide much of this work in practice. The catch is that all of these optimizations apply only \emph{within} a single machine word, not across the multiple words that make up a large-number operand.

\paragraph{Limb-by-limb addition with hardware carry chains.}
For large numbers, the operation has to be decomposed into a sequence of word-sized limb additions~\cite{brent2010modern}. For example, adding two 512-bit numbers on a 64-bit architecture takes eight 64-bit limb additions. Each one is fast on its own, but a sequential dependency creeps in: the carry-out of each addition has to flow into the next as carry-in, which brings the linear $O(m)$ delay back at the limb level. The textbook software approach computes $S_i = A_i + B_i + C_{i-1}$~\cite{knuth1997art,brent2010modern}; Algorithm~\ref{algo:gmp_add} (MPN-ADD-M) shows GMP's realization of this idea. To avoid manual carry detection and propagation, low-level assembly uses the add-with-carry chaining technique~\cite{intelSDM,armADDC}: the least significant limb addition starts with the \texttt{ADD} instruction (which sets the hardware carry flag), and subsequent limbs use \texttt{ADC} instructions that add both the limb operands and the incoming carry flag. Subtraction uses \texttt{SBB} instructions in the same way. This pattern is widely deployed in hand-optimized GMP and OpenSSL assembly, because mainstream compilers still do not reliably synthesize equivalent carry chains from straightforward high-level C loops~\cite{gnuGCCCompiler,llvmaddwithoverflow}.

\begin{algorithm}
\caption{Sequential Addition (GMP-style)}
\label{algo:gmp_add}
\footnotesize
\SetAlgoLined

\KwData{Integers $A$, $B$ each with $m$ limbs}
\KwResult{Sum $R$ and final carry $c_{out}$}

\textcolor{gray}{\tcp{Sequential limb addition with carry propagation}}
\textbf{Function:} \textsc{MPN-ADD-M}$(R, A, B, m)$\\
\Begin{
  $c_{in} \leftarrow 0$\;

  \For{$i \leftarrow 0$ \textbf{to} $m-1$ \textbf{step} $1$}{
    $R[i] \leftarrow A[i] + c_{in}$\; \label{line:add_carry_start}
    $c_{out} \leftarrow (R[i] < c_{in})$ ? $1$ : $0$\;
    $R[i] \leftarrow R[i] + B[i]$\;
    $c_{out} \leftarrow c_{out} + ((R[i] < B[i])$ ? $1$ : $0)$\; \label{line:add_carry_end}
  }
  \Return{$c_{out}$}\;
}
\end{algorithm}

\begin{algorithm}
\caption{Karatsuba Multiplication (or, Toom-Cook 2-way)}
\label{algo:karatsuba}
\footnotesize
\SetAlgoLined

\KwData{Integers $A, B$ with $m$ limbs}
\KwResult{Product $P = A \cdot B$}

\textbf{Function:} \textsc{MPN-KARATSUBA}$(A, B, m)$\\
\Begin{
  \If{$m \le \theta$}{ \label{line:karatsubA_basecase}
    \Return{\textsc{MPN-MUL-BASECASE}$(A, B, m)$}\;
  }
  $k \leftarrow m/2$, $b \leftarrow 2^{k \cdot \text{bits\_per\_limb}}$\; \label{line:karatsubA_split1}
  $A_H, A_L \leftarrow \text{split } A \text{ at } k$;\quad $B_H, B_L \leftarrow \text{split } B \text{ at } k$\; \label{line:karatsubA_split2}
  $P_1 \leftarrow \textsc{MPN-KARATSUBA}(A_H, B_H, k)$\;\label{line:karatsubA_high}
  $P_0 \leftarrow \textsc{MPN-KARATSUBA}(A_L, B_L, k)$\; \label{line:karatsubA_low}
  $s_x, s_y \leftarrow \text{sign of} (A_H - A_L), \text{sign of} (B_H - B_L)$\; \label{line:karatsubA_sign}
  $P_{\Delta} \leftarrow \textsc{MPN-KARATSUBA}(|A_H - A_L|, |B_H - B_L|, k)$\; \label{line:karatsubA_mid}
  $P_2 \leftarrow (s_x \cdot s_y) P_{\Delta}$\; \label{line:karatsubA_mid_sign}
  \Return{$(b^2 + b)P_1 - b P_2 + (b + 1)P_0$}\; \label{line:karatsubA_combine}
}
\end{algorithm}

\paragraph{Hardware multipliers and base-case multiplication routines.}
Arbitrary-precision multiplication builds on the same fixed-size hardware multipliers. A modern 64$\times$64-bit unit uses partial-product generation, Booth's encoding \cite{booth1951signed}, compressor-tree reduction (typically Wallace or Dadda trees \cite{wallace2006suggestion,dadda1965some}), and a final carry-propagate addition, all within a few cycles of latency~\cite{hennessy2011computer,agnerInstructionTables}. Many implementations also use Booth-style recoding to cut down on partial products before reduction~\cite{hennessy2011computer}. These datapath optimizations help single-limb multiply throughput but do nothing about the cross-limb accumulation dependencies that arise in arbitrary-precis\-ion multiplication. For small sizes ($m = 2, \ldots, 8$), libraries primarily use schoolbook and Comba~\cite{knuth1997art,comba90,brent2010modern}: schoolbook accumulates row-wise, while Comba accumulates column-wise to shorten carry chains. Assembly implementations commonly use \texttt{mulx}, which computes the full 128-bit product of two 64-bit limbs and returns both halves in a single instruction.

\paragraph{Recursive multiplication thresholds.}
For larger operan\-ds, most libraries fall back to Karatsuba's~\cite{karatsuba1963multiplication,brent2010modern} divide-and-conquer approach (Algorithm~\ref{algo:karatsuba}), which uses three half-size multiplications instead of four and brings the complexity down to $O(m^{\log_2 3}) \approx O(m^{1.585})$. GMP further switches to Toom-Cook~\cite{toom1963complexity} (a generalization of Karatsuba) and to FFT-based multiplication~\cite{schonhage1971fast} for very large operands, while OpenSSL typically only goes as far as Karatsuba beyond the base case. The exact switch points depend on the implementation, the operand size, and the target architecture. When the recursion reaches its base case (Line~\ref{line:karatsubA_basecase}), the library invokes schoolbook or Comba. Crucially, because Karatsuba and Toom-Cook trade multiplications for additions and subtractions at every recursive node, the aggregate add/sub work grows quickly with recursion depth. So a faster base-case multiply \emph{and} a faster add/sub both feed into every higher-level operation built on top: division, modular exponentiation, and computations like $\pi$.

\section{Proof of Correctness and Carry Probability Analysis}
\label{sec:appendix-carry}

This appendix formalizes the correctness of DoT's addition and multiplication. Also, it analyzes the probability of carry generation in large-number addition, which motivates DoT's design choice to isolate carry handling into a separate phase that only triggers on rare carry cascades.

\paragraph{Proof of correctness of DoT's addition algorithm}
Let two large integers $A$ and $B$ be represented in radix $X = {2^k}$ (where $k$ is the bit size of a limb) and suppose that the number of limbs for the two integers is $n$:
$$A = \sum_{i=0}^{n-1} A_i X^i = A_{n-1}X^{n-1} + \dots + A_1X^1 + A_0$$
$$B = \sum_{j=0}^{n-1} B_j X^j = B_{n-1}X^{n-1} + \dots + B_1X^1 + B_0$$
where $0 \le A_i, B_j < X$.
The sum of these two numbers using schoolbook addition is given by:
$$S = \sum_{i=0}^{n-1} (s_i X^i) + c_{n}X^{n}$$ such that 
$s_i = (A_i + B_i + c_i) \pmod{X}, c_0 = 0$ and
$c_{i+1} = \left\lfloor (A_i + B_i + c_i)/X \right\rfloor$

To show that Algorithm~\ref{algo:dot_add} is correct, we need to show that the output of the algorithm is the same as the sum computed above or $$S=\sum_{i=0}^{n-1} (s_i X^i)$$ and $c_{out} = c_n$.
Since the maximum value of a limb ($A_i, B_i$)  is $X-1$, the maximum sum of two limbs and a carry is: $(X-1) + (X-1) + 1 = 2X - 1$.
This means that $c_{i+1}$ can never exceed $\lfloor (2X-1)/X \rfloor = 1$. Also, when $s_i < A_i$ or $s_i < B_i$, it means that a carry is generated because the sum of two numbers is always greater than the individual numbers. 
\begin{lemma}[Carry generation exceeds limb size]
 \label{lem1}
  The sum of two $n$-bit integers $a$ and $b$ generates a carry, if and only if the least significant $n$ bits of the sum are less than both $a$ and $b$.
  In other words, if $a+b \geq a$, then $a+b \geq b$ and no carry-bit is generated. 
\end{lemma}
\begin{proof}
Let $a$ and $b$ be $n$-bit integers such that $0 \leq a,b < 2^n$, and $s = a+b$. 
A carry occurs if  $s \ge 2^n$.
Since $a,b < 2^n$, the maximum value of $s$ is $2^n - 1 + 2^n -1 = 2^{n+1} - 2$, which is strictly less than $2^{n+1}$.
As the sum is $2^n \leq s < 2^{n+1}$, the last $n$-bits of the sum are given by $s \bmod 2^n$, which is $s - 2^n$.
Since $a$ and $b$ are less than $2^n$, the sum is less than both $a$ and $b$.

To show that if $s < a$ and $s < b$, then a carry is generated:\\
Assume $s < a$ and no carry is generated. As no carry is generated, $0 \leq s = a+b < 2^n$. Since $b \geq 0$, $s \geq a$, which is a contradiction.
Similarly, if $s < b$, a carry has to be generated. 
  \end{proof}

Suppose $X = {2^k}$ and we have two integers  $a_{n-1}..a_1a_0$ and $b_{n-1}..b_1b_0$ such that $a_i, b_i$ are $k$-bits in size. The sum of these two integers is given by 
    $s = s_{n-1}..s_1s_0$ = $a_{n-1}..a_1a_0$ + $b_{n-1}..b_1b_0$, i.e.,\\
    $s_0 = a_0 + b_0 \pmod X$
    $s_1 = (a_1 + b_1 + (a_0 + b_0)/X) \pmod X$
    $s_2 = (a_2 + b_2 + ((a_1 + b_1 + (a_0 + b_0)/X)/X) \pmod X$ and so on. \\
    Equivalently, if 
    $s'_0 = r_0 + p0_0$, then  $s_0 = s'_0 \pmod X$ and \\
    if $s'_i = (r_i + p0_i + s'_{i-1}/X)$, then $s_i = s'_i \pmod X$\\
    
\begin{lemma}[Correctness of Phase 3 and Phase 4]
   \label{lem2}
    The result of adding a vector register, $r_{n-1}..r_1r_0$, with a mask register containing carry bits, $p0_{n-1}..p0_1p0_0$ per the phase 3 and 4 of Algorithm~\ref{algo:dot_add} is the same as the sum computed through normal addition.
\end{lemma}
\begin{proof}
We know that $p0_i$ is either $0$ or $1$. If all $p0_i = 0$, there is no effect on the values of $r_i$, which is equal to the value of $s_i$ as the values are all less than $2^k - 1$. Thus, Phase 4 is not triggered. \\
We show how $p0_i =1$ affects the addition, i.e., Phase 4 of the algorithm. \\
Base case: $r'_0 = r_0 + p0_0$. If carry is generated then $r'_0 = r_0 + p0_0 \pmod X$, $p1_0 = (r_0 + p0_0) / X$ and $c_{out} = p1_0$. $s_0 = r_0 + p_0 \pmod X$ and $s_1 = (0 + 0 + (r_0 + p0_0)/X) \pmod X = (r_0 + p0_0)/X$ as this value is either $0$ or $1$ and always smaller than X.\\
Inductive case: Suppose the lemma holds for $r_i$. We show that for $r_{i+1}$ the case holds.\\
Each $r_i$ and $r_{i+1}$ can have any of these three values:\\
1. $r_i < X - 2$ - if this is the case then for either value of $p0_i$, no further carry is generated. \\
2. $r_i = X - 2$ - for $p0_i=0$, no further carry is generated but for $p0_i=1$, the value of $r'_i = {2^k} - 1$, is equal to mask \\
3. $r_i = X - 1$ - for $p0_i=0$, no further carry is generated but the value of $r'_i = {2^k} - 1$, which is equal to mask. - for $p0_i=1$, a carry is generated which is captured in $p1_{i+1}=1$ and the value of $r'_i = 0$.\\
Case-analysing possible cases of $r_{i+1}$ and $r_i$, respectively, with $p0_i = 1$; for $r_i < X - 2$, we consider $r_i = X - 3$ w.l.o.g. as the largest possible number to demonstrate carry propagation (the remaining values will have the same behavior with respect to carry propagation):
\begin{compactitem}
    \item if ${2^k}$ - 3 and ${2^k}$ - 3, then $r'_{i+1} = {2^k} -2$ and $r'_i = {2^k} - 2$ as there are no carry-bits generated. Also, $s_{i+1} = {2^k} - 2$ and $s_i = {2^k}-2$
    \item if ${2^k}$ - 2 and ${2^k}$ - 3, then $r'_{i+1} = {2^k} - 1$ and $r'_i = {2^k} -2$ are one more as there are no carry-bits generated. Also, $s_{i+1} = {2^k} - 1$ and $s_i = {2^k}-2$
    \item if ${2^k}$ - 1 and ${2^k}$ - 3, then $r'_{i+1} = {2^k}$ and $r'_i = {2^k} - 2$. $p1_{i+2} = 1$, $p1_{i+1} = 0$, $p1_i = 0$ or $p1_i = 1$. $r_i'$ is not max so $p2_i = 0$. Hence, $p1_i$ is directly added to $r'_i$ with $c_{out} = 1$. 
    $s_{i+1} = {2^k} \pmod {2^k}$ and $s_i = {2^k}-2$ (or) $2^k-1$ depending on the value of carry and $s_{i+2} = {2^k}/{2^k}=1$. 
    \item if ${2^k}$ - 3 and ${2^k}$ - 2, then $r'_{i+1} = {2^k} -2$ and $r'_i= {2^k}-1$. $p2_i = 1$ as $r'_i$ is equal to mask. $p4_i = 0$ if $p1_i = 0$. No further carry is generated and $s_{i+1} = {2^k} - 2$ and $s_i = {2^k}-1$.
    If $p1_i = 1$, then $p3_i = 0$ and $p3_{i+1}= 1$. Adding $p4_i = 1$ to $r'_i$ and $p4_{i+1} = 1$ to $r'_{i+1}$, we get $r'_i = 0$ and $r'_{i+1} = 2^k -1$ . $s_i = 2^k - 1 + 1 = 2^k \pmod X$ and $s_{i+1} = 2^k - 1 \pmod X$. 
    \item if ${2^k}$ - 2 and ${2^k}$ - 2, then $r'_{i+1}= {2^k}-1$ and $r'_i= {2^k}-1$. $p2_{i+1} = p2_i = 1$ as $r'_i$ is equal to mask. $p4_{i+1} = p4_i = 0$ if $p1_{i+1} = p1_i = 0$. So the values are the same as the actual result. $s_{i+1} = {2^k} - 1$ and $s_i = {2^k}-1$
    If $p1_i = 1$, then $p3_i = 0$ and $p3_{i+1}=0$ but $c_{out} = 1$. Adding $p4_i = 1$ to $r'_i$ and $p4_{i+1} = 1$ to $r'_{i+1}$, we get $r'_i = 0$ and $r'_{i+1} = 0$ . $s_i = 2^k - 1 + 1 = 2^k \pmod X$ and $s_{i+1} = 2^k - 1 + 1 \pmod X$. 
    \item if ${2^k}$ - 1 and ${2^k}$ - 2, then $r'_{i+1} = {2^k}$ and $r'_i = {2^k}-1$. $p1_{i+2} = 1$ and $p2_i = 1$. Thus, $r'_{i+1} = 0$ and $r'_i = {2^k}-1$ and $c_{out} = 1$. $s_{i+1} = {2^k} \pmod {2^k}$ and $s_i = {2^k}-1$ and $s_{i+2} = {2^k}/{2^k} = 1$.
    If $p1_i = 1$, then $p3_i = 0$ and $p3_{i+1}=0$ but $c_{out} = 1$. Adding $p4_i = 1$ to $r'_i$ and $p4_{i+1} = 1$ to $r'_{i+1}$, we get $r'_i = 0$ and $r'_{i+1} = 1$ . $s_i = 2^k - 1 + 1 = 2^k \pmod X$ and $s_{i+1} = 2^k + 1 \pmod X$. 
    \item if ${2^k}$ - 3 and ${2^k}$ - 1, then $r'_{i+1} = {2^k}-2$ and $r'_i = {2^k}$. $p1_{i+1} = 1$, $p2_{i+1} = 0$. $p4_{i+1} = 1$. So $r'_{i+1} = {2^k} -1$ and $r'_i = 0$. $s_i = {2^k} \pmod {2^k} = 0$ and $s_{i+1} = {2^k} - 2 + ({2^k}/{2^k}) = {2^k} - 1$ because of $s_i$. As $r'_{i} = 2^k$, the case of $p1_i = 0$ or $1$ is the same.   
    \item if ${2^k}$ - 2 and ${2^k}$ - 1, then $r'_{i+1} = {2^k}-1$ and $r'_i = {2^k}$. $p2_{i+1} = p1_{i+1} = 1$. $p3_{i+1} = 0$ and $p3_{i+2} = 1$, $p4_{i+1} = 1$ and $p4_{i+2} = 1$. So $c_{out} = 1$, $r'_{i+1} = 0$ and $r'_i = 0$.
    $s_i = {2^k} \pmod {2^k} = 0$ and $s_{i+1} = {2^k} - 1 + ({2^k}/{2^k}) \pmod {2^k} = {2^k} \pmod {2^k} = 0$ because of $s_i$. $s_{i+2} = ({2^k} - 1 + ({2^k}/{2^k})) /{2^k} = 1$. As $r'_{i} = 2^k$, the case of $p1_i = 0$ or $1$ is the same.   
    \item if ${2^k}$ - 1 and ${2^k}$ - 1, then $r'_{i+1} = {2^k}$ and $r'_i = {2^k}$. $p1_{i+2} = p1_{i+1} = 1$ and $p2_i = 0$. So, $p4_{i+2} = p4_{i+1} = 1$. Thus, $r'_{i+1} = 1$ and $r'_i = 0$ and $c_{out} = 1$.
    $s_i = {2^k} \pmod {2^k} = 0$ and $s_{i+1} = {2^k} + ({2^k}/{2^k}) \pmod {2^k} = {2^k} + 1\pmod {2^k} = 1$ because of $s_i$. $s_{i+2} = ({2^k} + ({2^k}/{2^k})) /{2^k} = 1$. As $r'_{i} = 2^k$, the case of $p1_i = 0$ or $1$ is the same.   
\end{compactitem}

\end{proof}

\begin{theorem}[Correctness of DoT-Addition]
  Suppose two large integers $A$ and $B$,
  $$A = \sum_{i=0}^{n-1} A_i X^i = A_{n-1}X^{n-1} + \dots + A_1X^1 + A_0X^0$$
  $$B = \sum_{j=0}^{n-1} B_j X^j = B_{n-1}X^{n-1} + \dots + B_1X^1 + B_0X^0$$
  where $X = {2^k}$ (where $k$ is the bit size of a limb), the number of limbs for the two integers is $n$,  and $0 \le A_i, B_j < X$.

  Then, if $\textsc{DoT-Add-Words}(A,B,n) = (S, c_{\text{out}})$, then
 $$S = \sum_{i=0}^{n-1} (r_i X^i)$$ such that 
$r_i = (A_i + B_i + c_i) \pmod{X}, c_0 = 0$ and
$c_{i+1} = \left\lfloor (A_i + B_i + c_i)/X \right\rfloor$, 
and $c_{\text{out}} = c_n$
\end{theorem}
\begin{proof}
  The proof proceeds by induction on the \textsc{DoT-Add-Words} algorithm.\\
  \textbf{Base case} (for n = 1 limb): For the base case, this reduces to the addition of two 64-bit integers.
  In Phase 1, we compute $r_0 = A_0 + B_0$.\\
  In Phase 2, if $r_0 < A_0$, then $c = 1$ else $c=0$. As $w = 1$, $c_{out} = c \gg 0$, i.e., $c_{out} = c$.\\
  In Phase 3, $r_0 = A_0 + B_0 + c_0$. As $c_0 = 0$, $r_0 = A_0 + B_0$. Because $r_0$ is unchanged, Phase 4 is not triggered.\\
  Thus, the final value is $S = A_0 + B_0$ and $c_{out} = c$. \\
  From Lemma~\ref{lem1}, we know that if carry is generated $S < A_0$ and $S < B_0$.\\
  \textbf{Inductive case} (w.l.o.g.,  we show for $0<n<8$). \\
  IH: for $n$ limbs, 
  $$S' = \sum_{i=0}^{n-1} (r_i X^i)$$ such that $r_i = (A_i + B_i + c_i) \pmod{X}, c_0 = 0$ and $c_{i+1} = \left\lfloor (A_i + B_i + c_i)/X \right\rfloor$,  and $c_{\text{out}} = c_n$. \\
  To show: for $(n+1)^{th}$ limb, 
  $$S = \sum_{i=0}^{n} (r_i X^i)$$ such that $r_i = (A_i + B_i + c_i) \pmod{X}, c_0 = 0$ and $c_{i+1} = \left\lfloor (A_i + B_i + c_i)/X \right\rfloor$,  and $c_{\text{out}} = c_{n+1}$. \\
  $$S = (r_n X^n) + \sum_{i=0}^{n-1} (r_i X^i)$$
  $$S = (r_n X^n) + S'$$ such that $c_n = \left\lfloor (A_{n-1} + B_{n-1} + c_{n-1})/X \right\rfloor$ \\
  From Phase 1, we have $\forall\, (0 \leq i \leq n).\, r_i = (A_i + B_i)$ \\
  In Phase 2, if $r_i < A_i$, then $c_i = 1$ else $c_i = 0$. As $w =
  n$, $c_{out} = \mathbf{p0} \gg  n$, i.e., $c_{out} = c_i$. \\
  $\mathbf{p0} = c_{n-1}...c_0$.\\
  In Phase 3, we have $r_i = A_i + B_i + \mathbf{p0}_i$.  If this is more than
  $X$, we get $r_n = (A_n + B_n + c_n) \pmod{X}$ and a carry
  is  further generated, i.e., $(A_n + B_n + c_n)\pmod{X} < (A_n +
  B_n)$ or $\left\lfloor (A_n + B_n + c_n)/X \right\rfloor = 1$ (from
  Lemma~\ref{lem1}). \\
  The correctness of the final phases and the carry generation follows from Lemma~\ref{lem2}. 
\end{proof}

\paragraph{Proof of correctness of vertical and crosswise multiplication}
The \textsc{DoT-Mul-Words} algorithm provides the pseudocode for implementing vertical and crosswise multiplication on SIMD machines. As the registers can accommodate $32$- or $64$-bit numbers, the results are stored separately in \emph{lo} and \emph{hi} registers corresponding to the $128$-bit product. We show the correctness of the vertical and crosswise multiplication below: 
\begin{theorem}[Correctness of vertical and crosswise multiplication]
  \label{thm:mul}
  Suppose two large integers $A$ and $B$,
  $$A = \sum_{i=0}^{n-1} A_i X^i = A_{n-1}X^{n-1} + \dots + A_1X^1 + A_0X^0$$
  $$B = \sum_{j=0}^{n-1} B_j X^j = B_{n-1}X^{n-1} + \dots + B_1X^1 + B_0X^0$$
  where $X = {2^k}$ (where $k$ is the bit size of a limb), the number of limbs for the two integers is $n$,  and $0 \le A_i, B_j < X$.
  Then $A \times B = \textsc{DoT-Mul-Words}(A, B, n, k)$
\end{theorem}
\begin{proof}
  Without loss of generality, we assume that both $A$ and $B$ have the same number of limbs (this can be done by padding the smaller number with required number of 0s on the MSBs. 
  \begin{align}
    \begin{aligned}
    A \times B =&  \sum_{i=0}^{n-1} A_i X^i \times \sum_{j=0}^{n-1} B_j X^j\\
               =&(A_{n-1}X^{n-1} + \dots + A_1X^1 + A_0X^0) \times \\
               & B_{n-1}X^{n-1} + \dots + B_1X^1 + B_0X^0\\
               =&(A_{n-1}X^{n-1}. B_{n-1}X^{n-1} + A_{n-1}X^{n-1}. B_{n-2}X^{n-2} + ... + \\
                &A_0X^0. B_0X^0\\
    =&\sum_{i=0}^{n-1} \sum_{j=0}^{n-1} A_i.B_j X^{i+j} \\
  \end{aligned}
  \end{align}
  \begin{align}
    \begin{aligned}
      &\textsc{DoT-Mul-Words}(A, B, n, k) \\
      &\quad= A_0X^0. B_0X^0 + (A_1X^1. B_0X^0 + A_0X^0. B_1X^1) \\
      &\quad\phantom{{}=} + \dots + (A_{n-1}X^{n-1}.B_{n-2} X^{n-2} \\
      &\quad\phantom{{}=} \quad + A_{n-2}X^{n-2}.B_{n-1} X^{n-1}) \\
      &\quad\phantom{{}=} + (A_{n-1}X^{n-1}.B_{n-1} X^{n-1})\\
      &\quad=\sum_{j=0}^{n-1}\sum_{i=0}^{n-1}  A_i X^i B_j X^j\\
    \end{aligned}
  \end{align}
  Thus, $\textsc{DoT-Mul-Words}(A, B, n, k) =  A \times B$
  \end{proof}

\paragraph{Probability of carry generation in large-number addition}

\begin{lemma}[Probability of a maxed-out limb sum]
\label{prop:maxsum}
Let $X_i$ and $Y_i$ be independent random variables drawn uniformly from $\{0, 1, \ldots, 2^k - 1\}$. The probability that their sum equals the maximum representable value without overflow is
\[
  \Pr\!\left[X_i + Y_i = 2^k - 1\right] = 2^{-k}.
\]
\end{lemma}

\begin{proof}
The sample space is $\{0, \ldots, 2^k - 1\}^2$ with $2^{2k}$ equally likely outcomes. For the event $X_i + Y_i = 2^k - 1$, fixing any $X_i = j$ uniquely determines $Y_i = 2^k - 1 - j$. Since $0 \le j \le 2^k - 1$ implies $0 \le 2^k - 1 - j \le 2^k - 1$, this value of $Y_i$ is always valid. There are therefore exactly $2^k$ favorable pairs, one per value of $X_i$, giving
\[
  \Pr\!\left[X_i + Y_i = 2^k - 1\right] = \frac{2^k}{2^{2k}} = 2^{-k}. \qedhere
\]
\end{proof}

For $k = 64$ this is $\approx 5.4 \times 10^{-20}$: a maxed-out limb sum is astronomically rare. Carries themselves, however, are common:

\begin{lemma}[Probability of a carry out of a single limb]
\label{prop:carry}
Under the same conditions,
\[
  \Pr\!\left[X_i + Y_i \geq 2^k\right] = \frac{1}{2} - 2^{-(k+1)}.
\]
\end{lemma}

\begin{proof}
Count non-carry pairs, i.e., those with $X_i + Y_i \le 2^k - 1$. For each value $s \in \{0, \ldots, 2^k - 1\}$, there are $s+1$ pairs that sum to exactly $s$, so the total number of non-carry pairs is
\[
  \sum_{s=0}^{2^k-1}(s+1) = \frac{2^k(2^k + 1)}{2}.
\]
The carry-producing pairs number $2^{2k} - \frac{2^k(2^k + 1)}{2} = \frac{2^k(2^k - 1)}{2}$. Dividing by $2^{2k}$:
\[
  \Pr\!\left[X_i + Y_i \geq 2^k\right] = \frac{2^k - 1}{2^{k+1}} = \frac{1}{2} - 2^{-(k+1)}. \qedhere
\]
\end{proof}

As $k \to \infty$ this approaches $\frac{1}{2}$: roughly every other limb addition produces a carry. The crucial distinction from Lemma~\ref{prop:maxsum} is that while carries are frequent, carry \emph{propagation chains} require the intermediate sum $R_i$ to be maxed out, which is exponentially rare.

\begin{corollary}[Phase~4 is negligible for random inputs]
\label{cor:cascade}
Phase~4 of \textsc{DoT-Add} (Algorithm~\ref{algo:dot_add}) is triggered at
limb $i$ only when the Phase~1 output $R_i = 2^k - 1$ and the propagated
carry into that limb is $1$, so that Phase~3's addition overflows. By Lemma~\ref{prop:maxsum}, $\Pr[R_i = 2^k - 1] = 2^{-k}$. Since limbs are
independent, the occurrence of a carry cascade for $n$-limb addition is at most $(2^{-k})^n$,
which is exponentially small.
\end{corollary}

\section{Additional \remove{Plots and }Results on Intel Xeon Gold 6548Y+ (Emerald Rapids)}
\label{sec:appendix-er}
\noindent\update{\textbf{Pathological inputs.}} Figures~\ref{fig:add_sub_speedup_dotyc_pathological}
and~\ref{fig:add_sub_dot_simd_speedup_sub_pathological} show DoT's performance
on pathological inputs, complementing the random-input results in the main
text. DoT performs comparably to the two-level KSA and Ren et al.'s method on
pathological inputs, confirming that its speedup on random inputs does not come
at the cost of worse performance on worst-case inputs. The SIMD speedup plot
confirms that AVX512 still achieves a significant speedup over scalar even on
pathological inputs, while SSE and AVX2 underperform due to carry-management
overhead dominating at narrower widths.

\begin{figure*}[!t]
  \centering
  \begin{subfigure}[t]{0.48\linewidth}
    \centering
    \includegraphics[width=\linewidth]{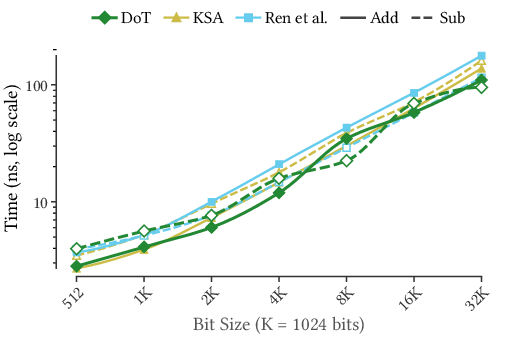}
    \caption{\small{Execution time (normalized, lower is better) of DoT (AVX512), two-level KSA (\texttt{add512}/\texttt{sub512}), and Ren et al.'s ProposedAdd/ProposedSub for addition and subtraction across 512--32768-bit pathological operands.}}
    \label{fig:add_sub_speedup_dotyc_pathological}
  \end{subfigure}\hfill
  \begin{subfigure}[t]{0.48\linewidth}
    \centering
    \includegraphics[width=\linewidth]{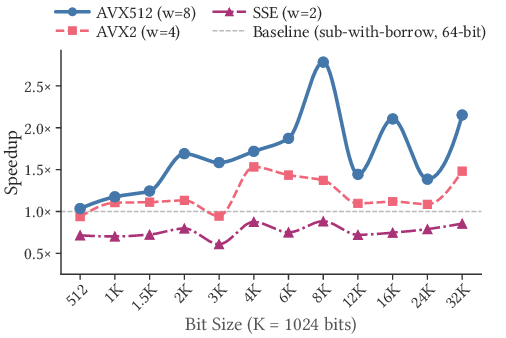}
    \caption{\small{Speedup of DoT SIMD variants ($w{=}2$ SSE, $w{=}4$ AVX2, $w{=}8$ AVX512) over scalar \texttt{\_subborrowx\allowbreak\_u64} for subtraction across 512--32768-bit pathological operands.}}
    \label{fig:add_sub_dot_simd_speedup_sub_pathological}
  \end{subfigure}
  \caption{\small{\update{DoT performance on pathological inputs on the Intel Xeon Gold 6548Y+ (Emerald Rapids).}}}
  \label{fig:er_pathological_combined}
\end{figure*}

\noindent\textbf{\update{Latency for OpenSSL workloads.}\remove{DoT's Impact on Higher-Level Applications.}} Since OpenS\-SL's
speed benchmark reports throughput in terms of operations per second, we
further compare the latency distributions of DoTSSL \update{(generic)} and OpenSSL \update{(generic)} for RSA
sign/verify, FFDH derive, and DSA sign/verify
(Figure~\ref{fig:openssl_speed_latency_cdf}). DoTSSL consistently achieves
lower latency across all operations. For instance, we observe median (P50)
latency improvements reaching up to \update{$4.7\%$}\remove{$7.9\%$} for DSA (2048-bit verify), \update{$5.6\%$}\remove{$6.1\%$}
for FFDH derive (2048-bit), and \update{$4.7\%$}\remove{$5.5\%$} for RSA sign (4096-bit). These latency
distributions confirm that the throughput gains observed earlier stem directly
from faster individual operations, showing consistent performance even at tail
percentiles. \remove{(e.g., $7.8\%$ and $7.7\%$ improvements at the 95th and 99th
percentiles for DSA 2048-bit verify)}

\begin{figure*}[!t]
    \centering
    \includegraphics[width=\linewidth]{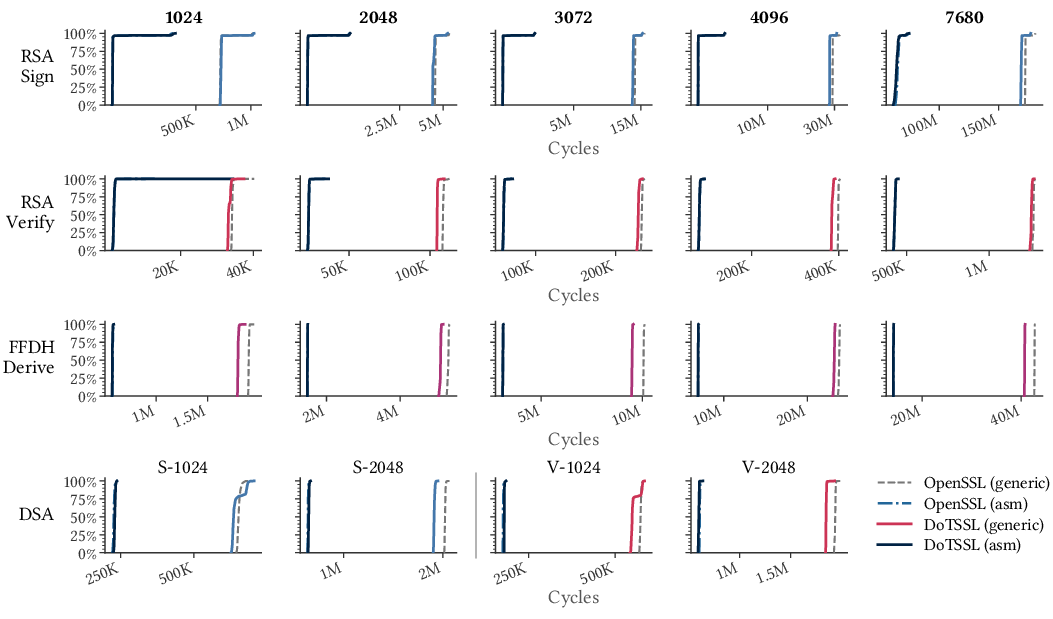}
    \caption{\small{\update{Latency CDFs} of DoTSSL vs. OpenSSL for RSA sign/verify, FFDH derive, and DSA sign/verify across the evaluated key sizes \update{on Intel Xeon Gold 6548Y+ (ER)}. \update{Cycles are measured via \texttt{RDTSC}.}}}
    \label{fig:openssl_speed_latency_cdf}
\end{figure*}

\noindent\update{\textbf{DoT's Contribution.}} \update{To understand the contribution of DoT's addition, subtraction, and multiplication routines to the
observed application-level speedups, we performed an analysis of the cycles spent (\%) across GMPbench, OpenSSL speed, PARI/GP, MPFR, and Boost benchmark workloads for DoT's replaced routines: \texttt{dot\_add\_words}, \texttt{dot\_sub\_words}, and \texttt{dot\allowbreak \_mul\allowbreak \_4x4} (refer to Figure~\ref{fig:dot_composition}). For \emph{generic} builds, DoT's routines are invoked for a major fraction of total execution time across various operations: $42.3\%$--$68.6\%$ in large GMPbench multiplications, $26.3\%$--$37.6\%$ in division, $17.0\%$--$21.0\%$ in PARI/GP, and $5.0\%$--$13.0\%$ in OpenSSL public-key workloads. Accelerating these critical-path operations directly translates into the observed end-to-end throughput gains. For \emph{assembly} builds, while DoT's multiplication was not invoked due to the tightly coupled assembly implementations of higher operations, DoT's addition and subtraction still get frequently invoked in add/sub-intensive routines, consuming $23.7\%$--$40.7\%$ of cycles in GMPbench, $19.4\%$--$22.0\%$ in PARI/GP, and $12.9\%$ in MPFR. This substantial cycle share explains why DoTMP (asm) surpasses hand-tuned scalar assembly in workloads like PARI \texttt{addii}. In contrast, OpenSSL's assembly workloads spend $>90\%$ of cycles in dedicated Montgomery modular exponentiation, limiting DoT's cycle share to $<1.5\%$ and explaining why DoTSSL (asm) closely tracks OpenSSL (asm).}

\begin{figure*}[!t]
    \centering
    \includegraphics[width=\linewidth]{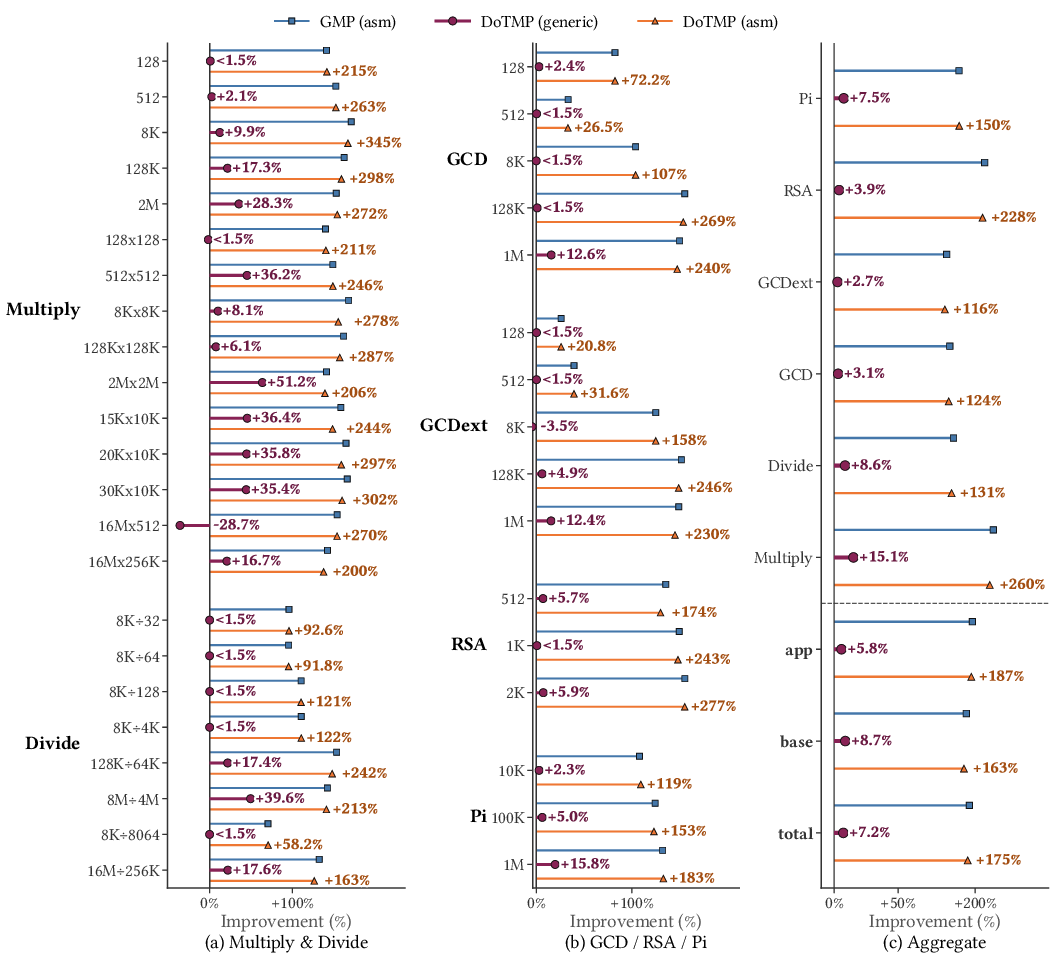}
    \caption{\small{\update{DoTMP (generic/asm) and GMP (asm) score (throughput) improvement over GMP (generic) in GMPbench on Intel 6548Y+ (ER) (full version).}}}
    \label{fig:gmpbench_full}
\end{figure*}

\begin{figure*}[!t]
    \centering
    \includegraphics[width=\linewidth]{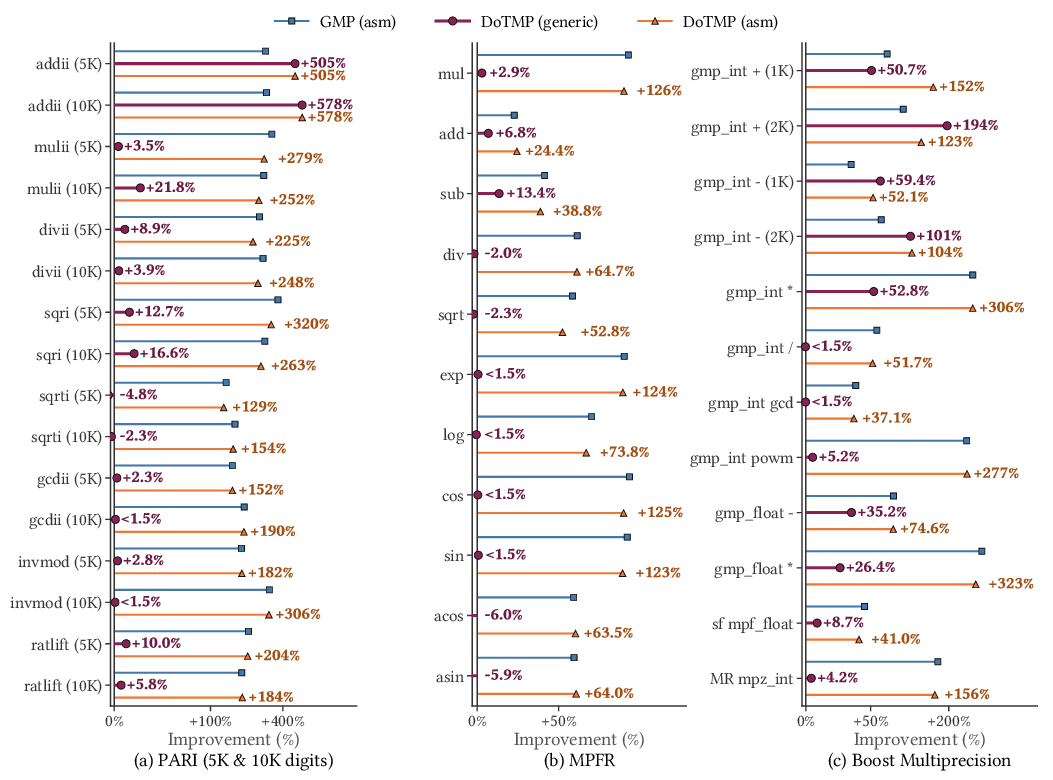}
    \caption{\small{\update{DoTMP's performance improvement over GMP (generic) for PARI/GP (5K \& 10K digits), MPFR, and Boost application benchmarks, each also compared against GMP (asm) and DoTMP (asm) on Intel Xeon 6548Y+ (ER). (full version)}}}
    \label{fig:app_bench_full}
\end{figure*}

\begin{figure*}[!t]
    \centering
    \includegraphics[width=\linewidth]{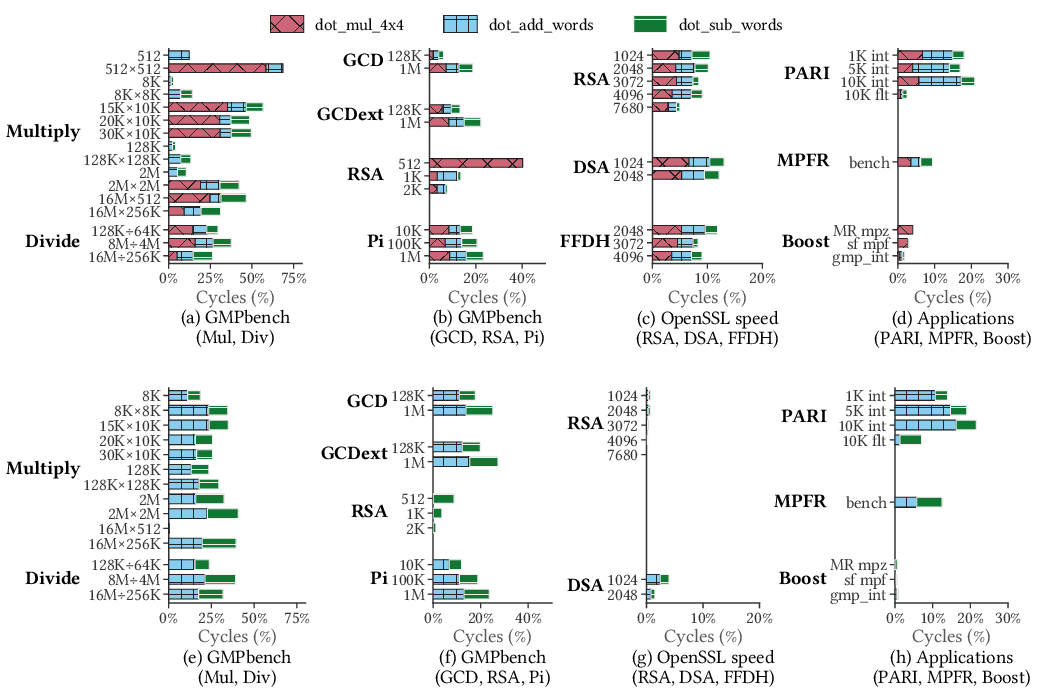}
    \caption{\small{\update{Cycles spent (\%) in DoT's \texttt{dot\_add\_words}, \texttt{dot\_sub\_words}, and \texttt{dot\_mul\_4x4} across  GMPbench, OpenSSL speed, PARI/GP, MPFR, and Boost workloads for generic ((a)--(d)) and assembly ((e)--(h)) builds (measured via \texttt{perf}; zero-contribution cases omitted).}}}
    \label{fig:dot_composition}
\end{figure*}

\section{Results on Intel Xeon Max 9462 (Sapphire Rapids)}
\label{sec:appendix-spr}

This appendix presents the full evaluation results on the Intel Xeon Max 9462 CPU (Sapphire Rapids microarchitecture, SPR), complementing the 6548Y+ (Emerald Rapids, ER) results reported in the main paper. The experimental setup, workloads, and methodology are identical to those described in Section~\ref{sec:evaluation}. All SPR experiments were compiled with \texttt{clang 14.0} using \texttt{-O2 -march=native}.

\begin{figure*}[!t]
    \centering
    \includegraphics[width=\linewidth]{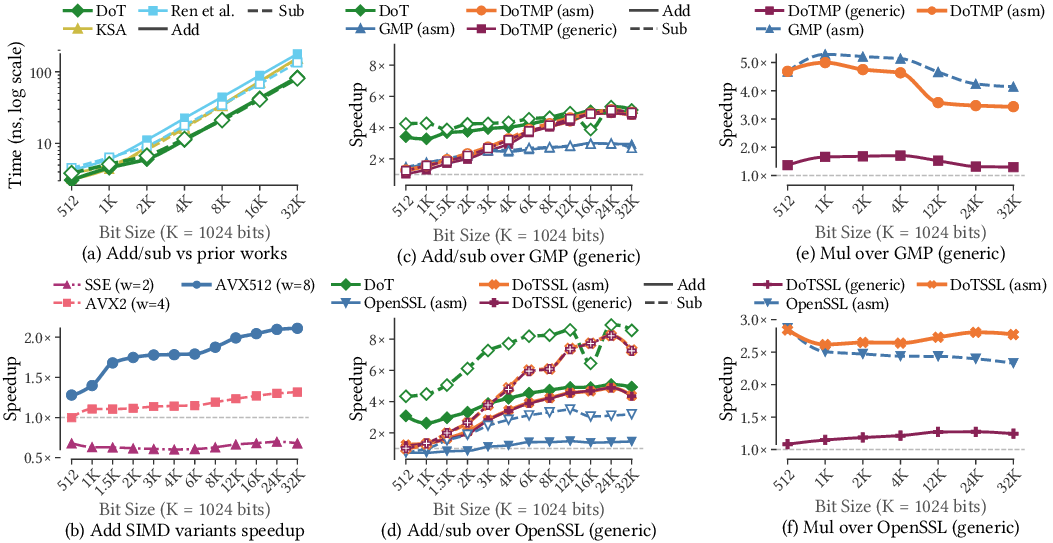}
    \caption{\small{\update{Micro-benchmark evaluation of DoT on the Intel Xeon Max 9462 (SPR), across six axes (512--32768-bit random operands): (a) DoT vs.\ prior add/sub work, (b) DoT SIMD width scaling, (c) add/sub speedup over GMP, (d) add/sub speedup over OpenSSL, (e) multiplication speedup over GMP, (f) multiplication speedup over OpenSSL.}\remove{Micro-benchmark evaluation of DoT on the Intel Xeon Max 9462 (SPR). (a) Execution time (log scale) of DoT (AVX512), two-level KSA, and Ren et al. across 512--32768-bit random operands. (b) Speedup of DoT SIMD variants ($w{=}2$ SSE, $w{=}4$ AVX2, $w{=}8$ AVX512) over scalar \texttt{\_addcarryx\_u64} for addition. (c) Timing speedup of DoTMP over GMP and DoTSSL over OpenSSL for addition and subtraction. (d) Timing speedup of DoTMP over GMP and DoTSSL over OpenSSL for multiplication.}}}
    \label{fig:spr_combined_microbench}
\end{figure*}

\subsection{DoT Addition and Subtraction vs. Prior Works}

\begin{figure*}[!t]
    \centering
    \includegraphics[width=\linewidth]{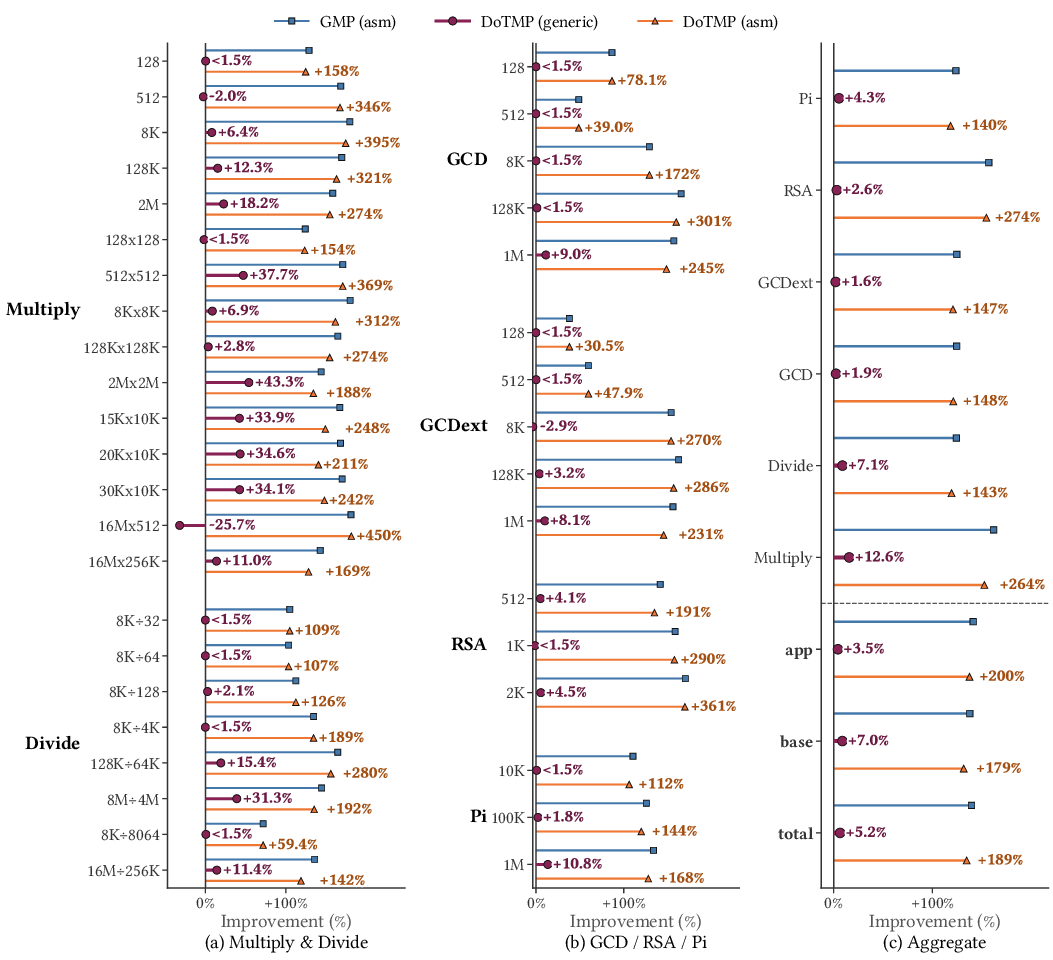}
    \caption{\small{DoTMP's percentage improvement over GMP across GMPbench workloads on the Intel Xeon Max 9462 (SPR).\remove{ Overall score improves by $6.2\%$, with multiply ($+12.7\%$) and pi ($+10.1\%$) leading, following the same workload-dependent pattern as ER but at modestly lower absolute gains.}}}
    \label{fig:spr_gmpbench_full}
    \label{fig:spr_gmpbench}
\end{figure*}

\begin{figure*}[!t]
    \centering
    \includegraphics[width=\linewidth]{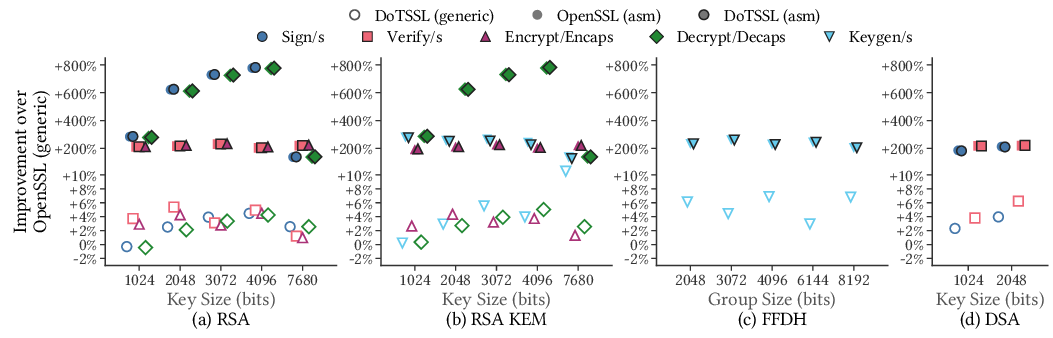}
    \caption{\small{DoTSSL throughput improvement (\%) over OpenSSL for RSA, RSA KEM, FFDH, and DSA on the Intel Xeon Max 9462 (SPR).\remove{ Improvements are generally higher than on ER: FFDH reaches up to $+7.2\%$ and DSA verify up to $+6.9\%$, reflecting SPR's higher base frequency amplifying the relative cost of scalar carry chains.}}}
    \label{fig:spr_openssl_speed}
\end{figure*}

Figure~\ref{fig:spr_combined_microbench}(a) shows the execution time comparison of DoT (AVX512) against the two-level KSA and Ren et al.'s method on SPR for random test cases. The trends closely mirror those on ER.

Compared to the two-level KSA, DoT (AVX512) achieves a geomean speedup of $1.4\times$ for addition ($1.23\times$ for smaller operands, $1.73\times$ for larger) and $1.4\times$ for subtraction ($1.12\times$ for smaller, $1.73\times$ for larger). Compared to Ren et al., DoT achieves a $1.82\times$ geomean speedup for random addition and $1.45\times$ for random subtraction. Instruction count reductions over Ren et al.\ are $29.1\%$ for random addition and $17.2\%$ for random subtraction, consistent with ER results. \update{As observed on ER, DoT exhibits modest speedups at smaller operand sizes ($512$--$2048$ bits) due to fixed invocation and mask-setup overheads. As operands scale to $32768$ bits, these fixed overheads amortize across vector iterations, allowing DoT's lower per-limb carry preparation cost to dominate the overall execution time. Furthermore, Ren et al.'s approach incurs sequential propagation delays between 64-bit blocks, whereas DoT resolves inter-lane carries in parallel during Phase~2 and Phase~3, sustaining higher throughput as operand lengths grow.}

\subsection{DoT's Performance over SIMD Widths}

\begin{figure*}[!t]
  \centering
  \begin{subfigure}[t]{0.48\linewidth}
    \centering
    \includegraphics[width=\linewidth]{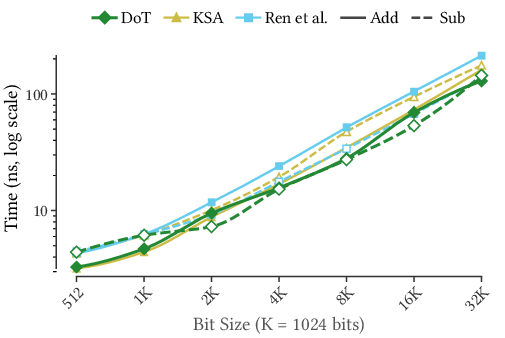}
    \caption{\small{Execution time (normalized, lower is better) of DoT (AVX512), two-level KSA, and Ren et al.'s method for addition and subtraction across 512--32768-bit pathological operands on the Intel Xeon Max 9462 (SPR).}}
    \label{fig:spr_dot_yc_ren_timing_pathological}
  \end{subfigure}\hfill
  \begin{subfigure}[t]{0.48\linewidth}
    \centering
    \includegraphics[width=\linewidth]{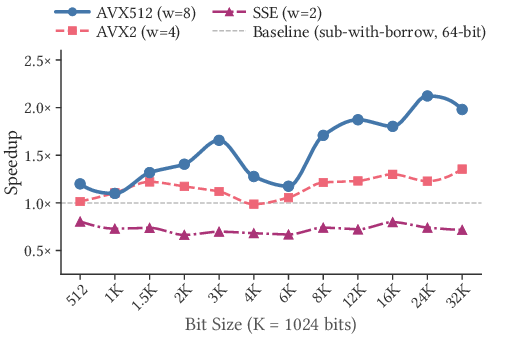}
    \caption{\small{Speedup of DoT SIMD variants ($w{=}2$ SSE, $w{=}4$ AVX2, $w{=}8$ AVX512) over scalar \texttt{\_subborrowx\allowbreak\_u64} for subtraction across 512--32768-bit pathological operands on SPR.}}
    \label{fig:spr_dot_simd_speedup_sub}
  \end{subfigure}
  \caption{\small{\update{DoT performance on pathological inputs on the Intel Xeon Max 9462 (SPR).}}}
  \label{fig:spr_pathological_combined}
  \vspace{0.3cm}
  \includegraphics[width=\linewidth]{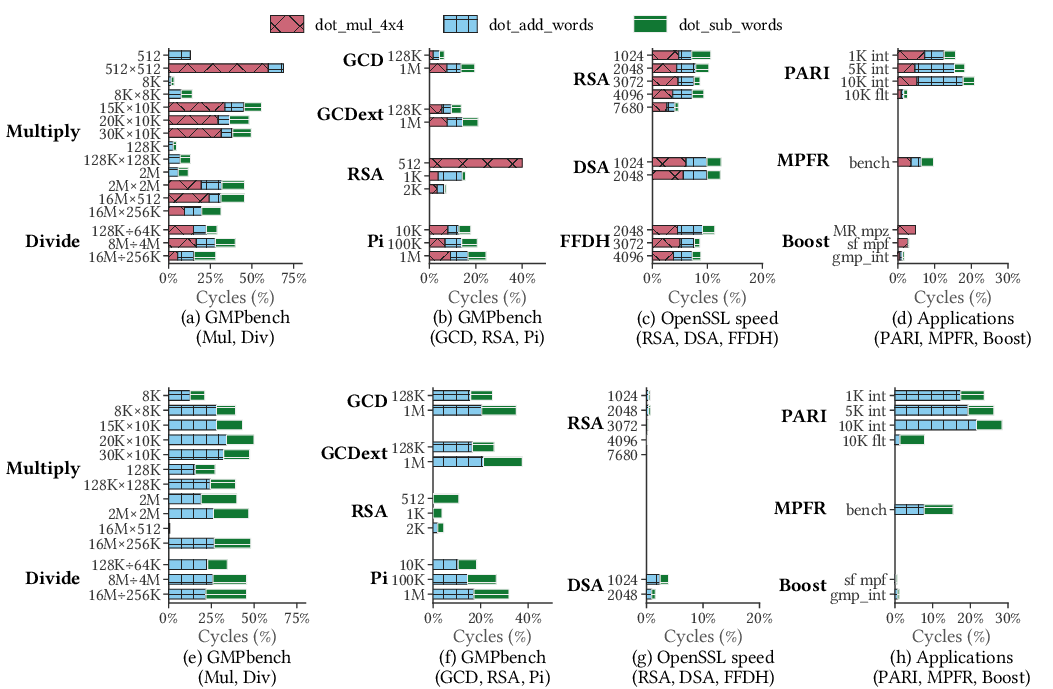}
  \caption{\small{\update{Cycles spent (\%) in DoT's \texttt{dot\_add\_words}, \texttt{dot\_sub\_words}, and \texttt{dot\_mul\_4x4} across  GMPbench, OpenSSL speed, PARI/GP, MPFR, and Boost workloads for generic ((a)--(d)) and assembly ((e)--(h)) builds (measured via \texttt{perf} on the Intel Xeon Max 9462 (SPR); zero-contribution cases omitted).}\remove{Cycle spent (\%) by DoT's \texttt{dot\_add\_words}, \texttt{dot\_sub\_words}, and \texttt{dot\_mul\_4x4} routines in GMPbench and OpenSSL speed workloads, measured via \texttt{perf} on the Intel Xeon Max 9462 (SPR). We omitted handful of cases in the GMPbench (e.g., lower sized mul, div and gcd) since they spend zero cycles in DoT routines. Additionally, OpenSSL speed benchmarks keygen, sign, encrypt, decrypt, etc. in aggregate for each key size; thus we report the cycle share for DoT routines in the aggregate of all operations for each key size.}}}
  \label{fig:dot_composition_spr}
\end{figure*}

Figure~\ref{fig:spr_combined_microbench}(b) shows the timing speedup of DoT SIMD variants over scalar add-with-carry on SPR. With SSE ($w=2$), geomean speedups are $0.6\times$ for addition and $0.7\times$ for subtraction. With AVX2 ($w=4$), speedups are $1.2\times$ for both. AVX512 ($w=8$) achieves $1.78\times$ for addition and $1.77\times$ for subtraction, slightly below ER's \update{$1.83\times$/$1.81\times$}\remove{$1.85\times$/$1.84\times$}, consistent with SPR's higher base clock and different memory subsystem characteristics. \update{At narrower vector widths ($w=2$), the instruction overhead of extracting carry masks outweighs SIMD parallelism benefits. AVX-512 ($w=8$) successfully crosses the threshold where lane-level parallelism amortizes the carry management logic, providing robust acceleration across large operands. Instruction counts for SSE and AVX2 increase over scalar addition by over $120\%$ and $25\%$, whereas AVX-512 achieves net instruction count reductions up to $30\%$, proving that wide vector execution is essential to overcome scalar carry-chain inefficiency.}

Figures~\ref{fig:spr_dot_yc_ren_timing_pathological} and~\ref{fig:spr_dot_simd_speedup_sub} show DoT's performance and SIMD width scaling on SPR for pathological inputs. \update{Under pathological inputs where full carry and borrow propagation chains are triggered, DoT's Phase~4 resolves propagation across SIMD lanes without falling back to sequential scalar operations, maintaining consistent speedups over naive SIMD carry propagation. Even in the worst-case carry ripple scenarios, DoT avoids pipeline stalls, verifying that its correctness guarantees do not compromise worst-case performance.}

\subsection{256-bit Multiplication with DoT} 

\begin{table}[!t]
  \centering
  \footnotesize
  \begin{tabular*}{\linewidth}{@{\extracolsep{\fill}}lrrr@{}}
  \toprule
  \textbf{Implementation} & \textbf{Instructions} & \textbf{Avg.\ Cycles} & \textbf{IPC} \\
  \midrule
  \texttt{dot\_mul\_${5\times5}$} & 221 & \update{35.8}\remove{36.0} & \update{7}\remove{6.1} \\
  \texttt{dot\_mul\_${4\times4}$} & 265 & \update{48.1}\remove{47.1} & \update{5}\remove{5.6} \\
  Gueron \& Krasnov~\cite{gueron2016accelerating} & 350 & \update{103.4}\remove{103.5} & \update{4}\remove{3.3} \\
  OpenSSL \texttt{BN\_mul} \update{(generic)}     & \update{739}\remove{655} & \update{90.0}\remove{86.0} & \update{8.2}\remove{7.6} \\
  \update{OpenSSL \texttt{BN\_mul (asm)}}       & \update{436} & \update{53.1} & \update{8.2} \\
  GMP \texttt{mpz\_mul} \update{(generic)}        & 872 & \update{111.6}\remove{111.4} & 7.8 \\
  \update{GMP \texttt{mpz\_mul (asm)}}          & \update{286} & \update{35.6} & \update{8.0} \\
  \update{OpenSSL \texttt{amm}$^*$\texttt{52x5\_ifma256}} & \update{733} & \update{195.6} & \update{3.8} \\
  \bottomrule
  \end{tabular*}
  \caption{\small{Instruction counts, average cycle counts, and IPC for 256-bit multiplication on the Intel Xeon Max 9462 (SPR). Cycle counts were obtained via \texttt{RDTSC}. \remove{\texttt{RDTSCP} (minimum 700 million total cycles per reported value, averaged over millions of iterations.)} \update{$^*$Extracted from OpenSSL which performs Almost Montgomery Multiplication and uses AVX-512IFMA.}}}
  \label{tab:256_mul_max}
\end{table}

\begin{figure*}[!t]
    \centering
    \includegraphics[width=\linewidth]{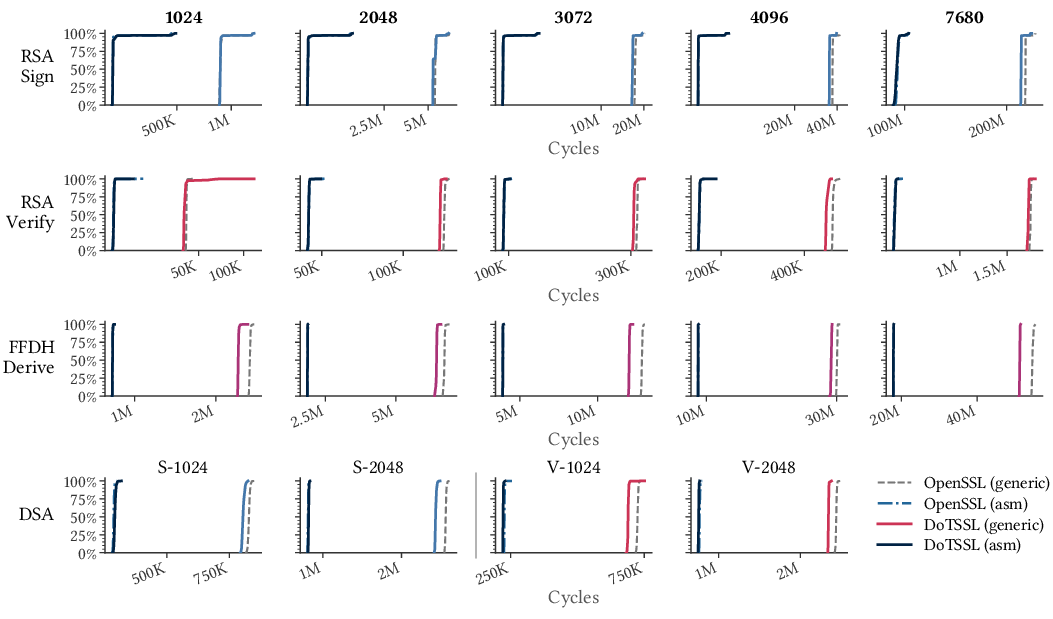}
    \caption{\small{Latency comparison (CDF) of DoTSSL vs. OpenSSL for RSA sign/verify, FFDH derive, and DSA sign/verify on the Intel Xeon Max 9462 (SPR)\remove{. Cycles are measured via \texttt{RDTSC} on the Intel Xeon Max 9462 (SPR) and plotted on a log scale.}.}}
    \label{fig:spr_openssl_latency_cdf}
\end{figure*}

\begin{figure*}[!t]
    \centering
    \includegraphics[width=\linewidth]{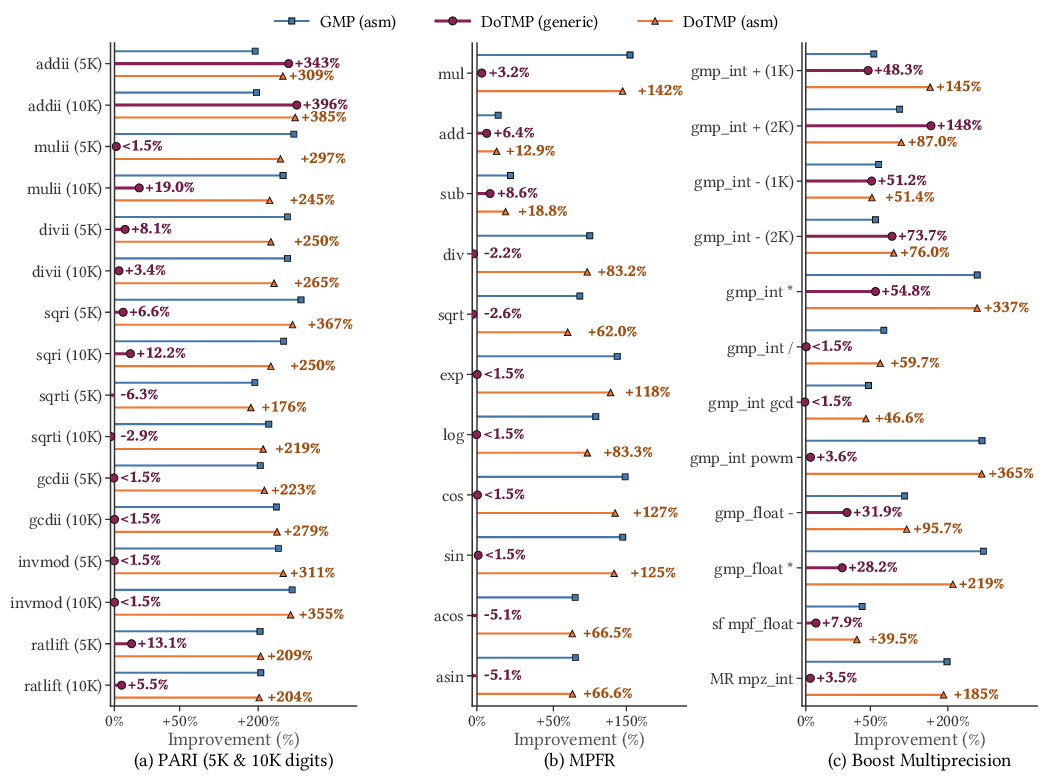}
    \caption{\small{\update{DoTMP's performance improvement over GMP for PARI/GP (5K \& 10K digits), MPFR, and Boost application benchmarks on the Intel Xeon Max 9462 (SPR), comparing GMP (asm), DoTMP (generic), and DoTMP (asm).}}}
    \label{fig:spr_appbench_full}
\end{figure*}

Compared to the ER results, SPR shows a similar pattern of DoT multiplication outperforming both GMP and OpenSSL baselines: Table~\update{\ref{tab:256_mul_max}}\remove{\ref{tab:256_mul}} shows that DoT's $4\times4$ multiplication routine achieves similar speedups over these baselines on SPR as on ER, confirming that the multiplication gains are architecture-independent and primarily driven by reduced dependency chains rather than memory subsystem differences.

\update{As detailed in Table~\ref{tab:256_mul_max}, DoT's $4\times4$ routine completes in 48.1 average cycles (IPC 5.0), achieving a $2.15\times$ speedup over Gueron and Krasnov's AVX-512IFMA routine (103.4 cycles, IPC 4.0) and a $1.87\times$ speedup over OpenSSL's generic \texttt{BN\_mul} (90.0 cycles). Over GMP's generic \texttt{mpz\_mul} (111.6 cycles), DoT is $2.32\times$ faster. When radix conversion overhead is eliminated in \texttt{dot\_mul\_5x5} (35.8 cycles, IPC 7.0), DoT matches GMP's hand-tuned assembly multiplication (35.6 cycles) outright. In contrast, OpenSSL's AVX-512IFMA Almost Montgomery Multiplication routine (\texttt{amm\_52x5\_ifma256}) requires 195.6 cycles due to radix conversion between 64-bit and 52-bit representations, confirming that dedicated modular multipliers do not provide general-purpose base-case multiplication acceleration. DoT's restructured partial-product accumulation eliminates intra-lane dependencies, enabling the CPU's out-of-order execution engine to sustain high IPC throughout the multiplication core.}

\subsection{\update{DoT-integrated GMP and OpenSSL}\remove{DoT's Performance over GMP and OpenSSL}}

\update{Figures~\ref{fig:spr_combined_microbench}(c)--(f) compare the execution time speedup of DoTMP and DoTSSL across 512--32768 bit operands against the generic builds of GMP and OpenSSL on SPR. As observed on ER, standalone DoT exhibits higher speedups for subtraction due to the higher relative overhead in generic subtraction routines, while at smaller operand sizes, sequential operand magnitude comparison moderates the speedup before converging with addition at larger sizes.}

\noindent\textbf{Addition and Subtraction.}
Figure~\ref{fig:spr_combined_microbench}\remove{(c)} shows the speedup of DoTMP over GMP \update{(panel (c))} and DoTSSL over OpenSSL \update{(panel (d))} on SPR. For addition, DoTMP achieves a \update{$3.02\times$}\remove{$2.80\times$} geomean speedup over GMP (up to $4.94\times$)\remove{($1.84\times$ for smaller operands, $4.25\times$ for larger, peak $4.98\times$)}, with instruction count reductions averaging $64.7\%$ (up to $78.5\%$). DoTSSL achieves \update{$2.94\times$}\remove{$2.92\times$} geomean over OpenSSL (up to $4.88\times$)\remove{($1.92\times$ small, $4.45\times$ large, peak $4.93\times$)}, with instruction count reductions of $49.8\%$ on average. For subtraction, DoTMP achieves \update{$3.10\times$}\remove{$3.07\times$} geomean (up to $5.15\times$)\remove{($2.06\times$ small, $4.55\times$ large, peak $5.17\times$)} and DoTSSL achieves \update{$3.97\times$}\remove{$4.00\times$} geomean (up to $8.23\times$)\remove{($2.25\times$ small, $7.10\times$ large, peak $8.24\times$)}. These results are lower than ER for DoTMP (ER: \update{$4.00\times$/$3.80\times$}\remove{$3.81\times$/$3.73\times$}) but comparable for DoTSSL, consistent with SPR's architectural differences in SIMD throughput.

\noindent\textbf{Multiplication.}
Figure~\ref{fig:spr_combined_microbench}\remove{(d)} shows multiplication speedups on SPR. DoTMP achieves \update{$1.50\times$}\remove{$1.40\times$} geomean over GMP \update{(generic)} (\update{$1.30\times$--$1.71\times$}\remove{$1.30\times$--$1.48\times$}, instruction count reduction $47.2\%$ on average). DoTSSL achieves $1.20\times$ geomean over OpenSSL \update{(generic)} (\update{$1.08\times$--$1.27\times$}\remove{$1.10\times$--$1.32\times$}, instruction count reduction \update{$36.1\%$}\remove{$41.6\%$}). Both are consistent with ER results\remove{($1.41\times$ and $1.20\times$)}, confirming that multiplication gains are architecture-independent. \update{The speedup peaks when operand bit lengths are powers of two and follow the Toom-Cook path to the 256-bit base multiplier, while OpenSSL's recursive Karatsuba structure consistently invokes DoT's multiplier across all sizes.}

\subsection{DoT's Impact on Higher-Level Applications}

\noindent\textbf{GMPbench.}
Figure~\update{\ref{fig:spr_gmpbench_full}}\remove{\ref{fig:spr_gmpbench}} shows DoTMP's improvement on GMPbench on SPR. The overall score improves by $6.2\%$ \update{(base score $+7.1\%$, total score $+5.2\%$)}. The multiply aggregate improves by \update{$12.6\%$}\remove{$12.7\%$} (individual cases up to $40.0\%$), divide \update{aggregate }by \update{$7.1\%$}\remove{$6.9\%$} (up to \update{$31.3\%$}\remove{$30.3\%$}), and pi \update{aggregate }by \update{$4.4\%$}\remove{$10.1\%$} (up to \update{$10.8\%$}\remove{$13.7\%$} for the 1M-digit case). GCD and GCDext improve by \update{$1.9\%$}\remove{$2.6\%$} and \update{$1.6\%$}\remove{$2.2\%$} respectively. RSA improves by \update{$2.6\%$}\remove{$2.8\%$}. These gains are consistently lower than ER (ER: $7.8\%$ overall) but follow the same workload-dependent pattern. \update{The propagation of primitive gains through the benchmark suite confirms that accelerating base arithmetic benefits composite operations such as Newton-Raphson division and Chudnovsky $\pi$ computation, which call addition, subtraction, and multiplication in tight inner loops.}

\noindent\textbf{\update{OpenSSL Speed and Latency.}\remove{OpenSSL Speed}}
\update{Figure~\ref{fig:spr_openssl_speed} shows DoTSSL throughput improvements over OpenSSL for RSA, RSA KEM, FFDH, and DSA on SPR. Furthermore, Figure~\ref{fig:spr_openssl_latency_cdf} shows the latency CDF comparison.}\remove{Figure~\ref{fig:spr_openssl_speed} shows DoTSSL's throughput improvements on OpenSSL speed on SPR. Improvements are generally higher than on ER: RSA averages $3.9\%$ across sign, verify, encrypt, and decrypt (up to $6.0\%$ for 4096-bit encrypt). FFDH averages $5.4\%$ (up to $7.2\%$ for 4096-bit groups). DSA sign averages $4.4\%$ and verify $5.4\%$ (up to $6.9\%$ for 2048-bit verify). The stronger gains on SPR reflect its higher base frequency making the relative cost of scalar carries more pronounced.
Similarly, the latency distributions (Figure~\ref{fig:spr_openssl_latency_cdf}) show that DoTSSL consistently achieves lower latency across all operations on SPR. For instance, we observe median (P50) latency improvements reaching up to $6.1\%$ for DSA (2048-bit verify), $6.7\%$ for FFDH derive (4096-bit), and $5.8\%$ for RSA sign (4096-bit). These latency improvements maintain consistent performance even at tail percentiles (e.g., $6.2\%$ and $6.1\%$ improvements at the 95th and 99th percentiles for DSA 2048-bit verify).}
\update{Median (P50) latency improvements reach up to $6.1\%$ for DSA (2048-bit verify), $6.7\%$ for FFDH derive (4096-bit), and $5.8\%$ for RSA sign (4096-bit), maintaining consistent gains even at tail percentiles ($6.2\%$ and $6.1\%$ improvements at P95 and P99 for DSA 2048-bit verify).}

\noindent\textbf{\update{Application Benchmarks.}\remove{DoT's Contribution to these gains}}
\update{Figure~\ref{fig:spr_appbench_full} shows the execution time speedup of DoTMP over GMP for PARI/GP (5K \& 10K digits), MPFR, and Boost Multiprecision on SPR. DoTMP demonstrates substantial speedups across these real-world workloads, following a similar pattern to the ER results: on add/sub-intensive routines such as PARI \texttt{addii} and Boost \texttt{gmp\_int +}, DoTMP significantly accelerates performance over GMP's generic build, while the assembly builds track native assembly closely and surpass it on large addition and subtraction. In composite operations like PARI \texttt{sqri} and \texttt{ratlift}, DoTMP achieves consistent acceleration by speeding up internal limb additions. On Boost Multiprecision, \texttt{gmp\_int +} and \texttt{gmp\_int -} show speedups exceeding $100\%$ at 2K-bit operands.}

\noindent\textbf{\update{DoT Contribution.}}\remove{Similar to ER, we used \texttt{perf} to analyze the cycle composition of DoT's routines in GMPbench and OpenSSL speed on SPR. Figure~\ref{fig:dot_composition_spr} shows that DoT's \texttt{dot\_add\_words}, \texttt{dot\_sub\_words}, and \texttt{dot\_mul\_4x4} routines account for a significant share of cycles in both GMPbench and OpenSSL speed workloads, confirming that DoT's improvements in these core operations are driving the overall performance gains observed in higher-level applications on SPR as well.}
\update{Figure~\ref{fig:dot_composition_spr} shows the cycle percentage spent in DoT routines across GMPbench, OpenSSL speed, PARI/GP, MPFR, and Boost workloads on SPR. Similar to ER, DoT's \texttt{dot\_add\_words}, \texttt{dot\_sub\_words}, and \texttt{dot\_mul\_4x4} routines account for a significant share of execution cycles in add/sub- and multiply-intensive workloads, driving the overall application speedups on SPR and confirming that DoT's benefits translate directly into higher-level scientific computing and cryptographic tasks.}

\clearpage
\section{\update{Results on AMD EPYC 9354P (Zen 4)}}
\label{sec:appendix-zg}

\update{This appendix presents full evaluation results on the AMD EPYC 9354P CPU (Zen 4 microarchitecture, ZG), complementing the Intel Emerald Rapids (ER) and Sapphire Rapids (SPR) evaluations. The experimental setup, workloads, and methodology are identical to those described in Section~\ref{sec:evaluation}. All ZG experiments were compiled with \texttt{clang 14.0} using \texttt{-O2 -march=native}.}

\begin{figure*}[!t]
    \centering
    \includegraphics[width=\linewidth]{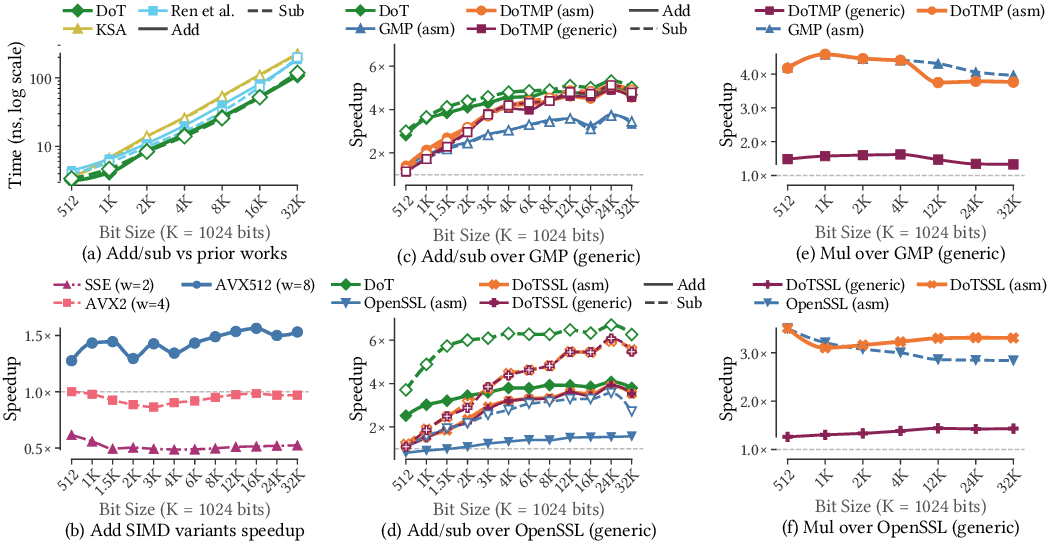}
    \caption{\small{\update{Micro-benchmark evaluation of DoT on the AMD EPYC 9354P (ZG), across six axes (512--32768-bit random operands): (a) DoT vs.\ prior add/sub work, (b) DoT SIMD width scaling, (c) add/sub speedup over GMP, (d) add/sub speedup over OpenSSL, (e) multiplication speedup over GMP, (f) multiplication speedup over OpenSSL.}}}
    \label{fig:zg_combined_microbench}
\end{figure*}

\subsection{\update{DoT Addition and Subtraction vs. Prior Works}}

\begin{figure*}[!t]
    \centering
    \includegraphics[width=\linewidth]{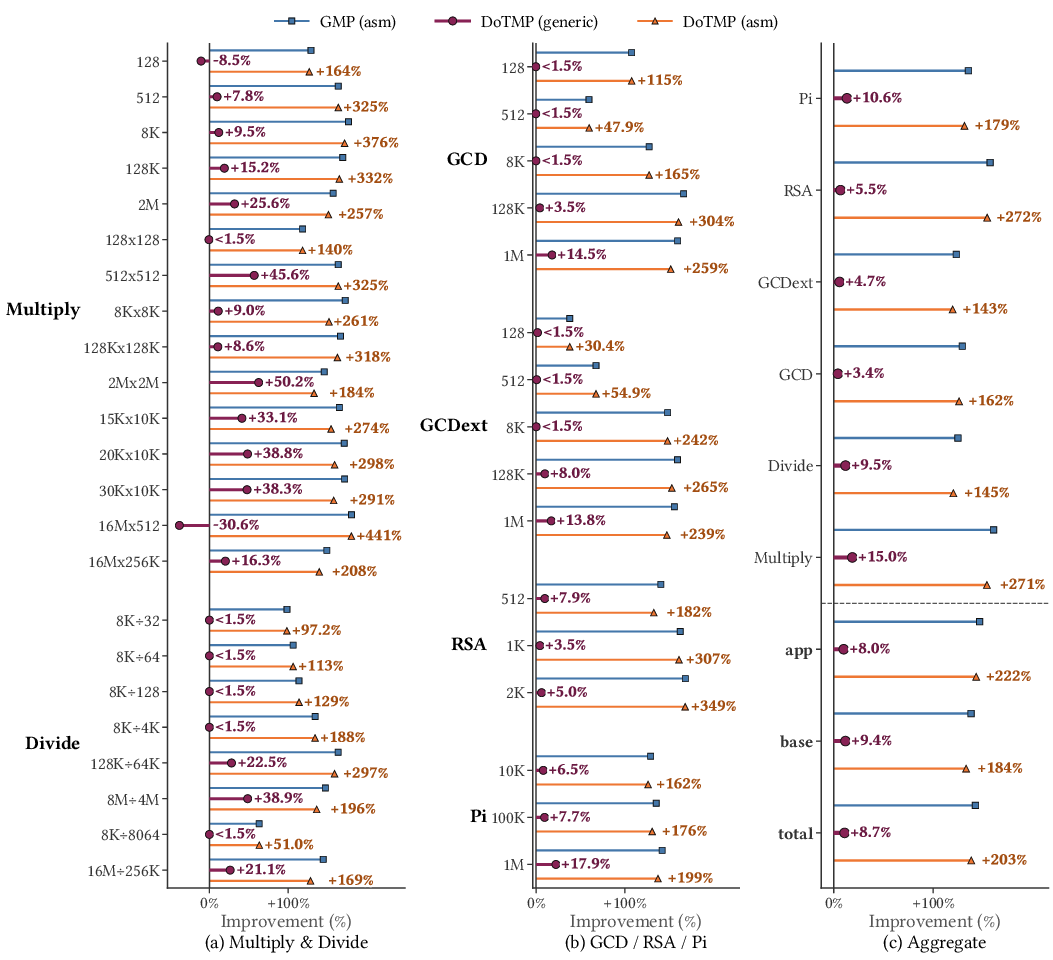}
    \caption{\small{\update{DoTMP's percentage improvement over GMP across GMPbench workloads on the AMD EPYC 9354P (ZG).}}}
    \label{fig:zg_gmpbench_full}
\end{figure*}

\begin{figure*}[!t]
    \centering
    \includegraphics[width=\linewidth]{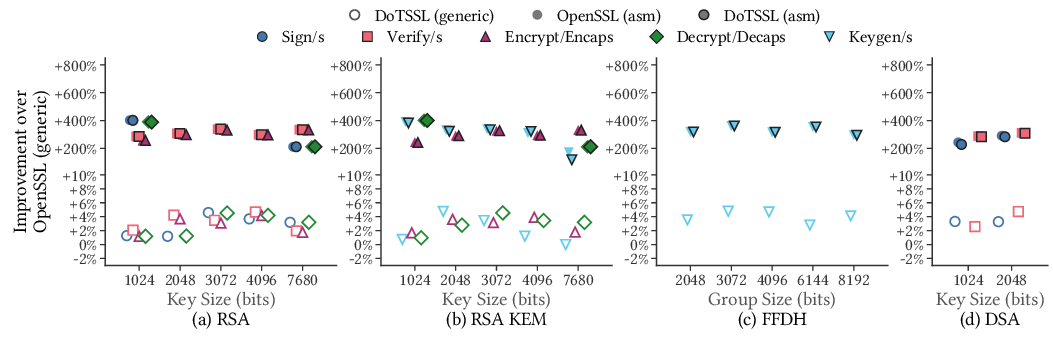}
    \caption{\small{\update{DoTSSL throughput improvement (\%) over OpenSSL for RSA, RSA KEM, FFDH, and DSA on the AMD EPYC 9354P (ZG).}}}
    \label{fig:zg_openssl_speed}
\end{figure*}

\update{Figure~\ref{fig:zg_combined_microbench}(a) shows the execution time comparison of DoT (AVX512) against prior vector addition and subtraction approaches on ZG for random test cases. On Zen 4, AVX-512 instructions are executed across dual 256-bit execution units. Despite this split-execution model, DoT consistently outperforms both the two-level KSA and Ren et al.'s method across all evaluated operand lengths, demonstrating that DoT's algorithmic elimination of sequential carry propagation is effective across processor microarchitectures from different vendors.}

\subsection{\update{DoT's Performance over SIMD Widths}}

\update{Figure~\ref{fig:zg_combined_microbench}(b) shows the timing speedup of DoT SIMD variants over scalar add-with-carry on ZG. As observed on Intel architectures, DoT demonstrates clear scaling across vector widths: AVX-512 ($w=8$) provides the highest performance across all operand sizes, while narrower vector widths ($w=2$ SSE, $w=4$ AVX2) exhibit lower speedups due to the higher relative overhead of mask extraction and carry propagation per limb. AVX-512 ($w=8$) amortizes the carry management overhead across 8 lanes, delivering substantial performance gains on AMD Zen 4.}

\subsection{\update{256-bit Multiplication with DoT}}

\begin{table}[!t]
  \centering
  \footnotesize
  \begin{tabular*}{\linewidth}{@{\extracolsep{\fill}}lrrr@{}}
  \toprule
  \textbf{Implementation} & \textbf{Instructions} & \textbf{Avg.\ Cycles} & \textbf{IPC} \\
  \midrule
  \update{\texttt{dot\_mul\_${5\times5}$}} & \update{222} & \update{46.2} & \update{4.8} \\
  \update{\texttt{dot\_mul\_${4\times4}$}} & \update{266} & \update{61.0} & \update{4.4} \\
  \update{Gueron \& Krasnov~\cite{gueron2016accelerating}} & \update{342} & \update{167.8} & \update{2.0} \\
  \update{OpenSSL \texttt{BN\_mul (generic)}}     & \update{739} & \update{131.3} & \update{5.6} \\
  \update{OpenSSL \texttt{BN\_mul (asm)}}       & \update{436} & \update{69.4} & \update{6.3} \\
  \update{GMP \texttt{mpz\_mul (generic)}}        & \update{884} & \update{152.0} & \update{5.8} \\
  \update{GMP \texttt{mpz\_mul (asm)}}          & \update{301} & \update{53.4} & \update{5.6} \\
  \update{OpenSSL \texttt{amm}$^*$\texttt{52x5\_ifma256}} & \update{735} & \update{182.4} & \update{4.0} \\
  \bottomrule
  \end{tabular*}
  \caption{\small{\update{Instruction counts, average cycle counts, and IPC for 256-bit multiplication on the AMD EPYC 9354P (ZG). Cycle counts were obtained via \texttt{RDTSC}/\texttt{RDTSCP}. $^*$Extracted from OpenSSL which performs Almost Montgomery Multiplication and uses AVX-512IFMA.}}}
  \label{tab:256_mul_zg}
\end{table}

\begin{figure*}[!t]
    \centering
    \includegraphics[width=\linewidth]{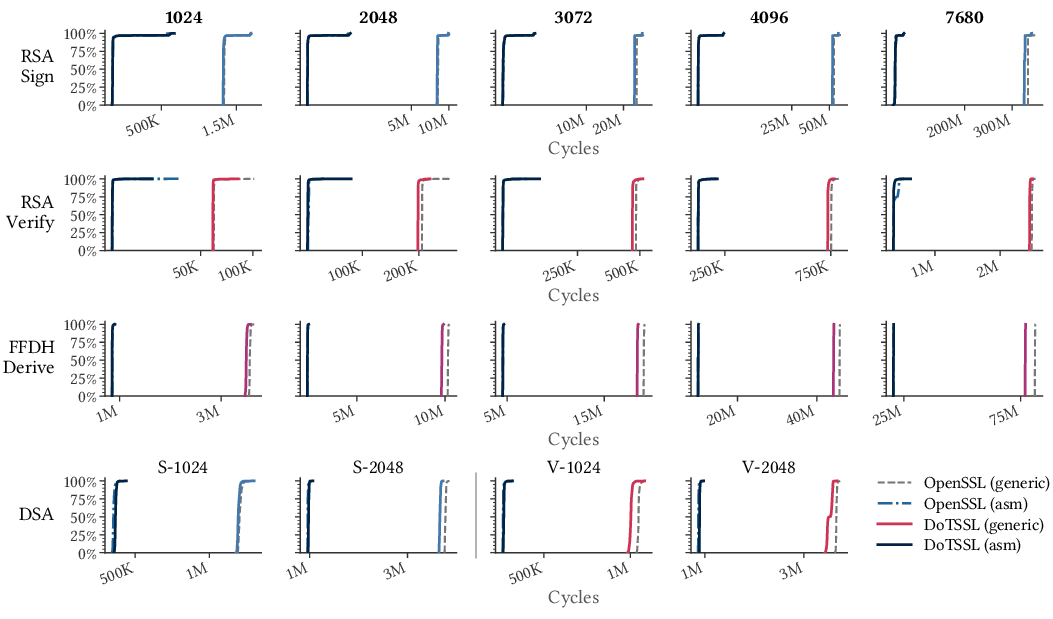}
    \caption{\small{\update{Latency comparison (CDF) of DoTSSL vs. OpenSSL for RSA sign/verify, FFDH derive, and DSA sign/verify on the AMD EPYC 9354P (ZG).}}}
    \label{fig:zg_openssl_latency_cdf}
\end{figure*}

\begin{figure*}[!t]
    \centering
    \includegraphics[width=\linewidth]{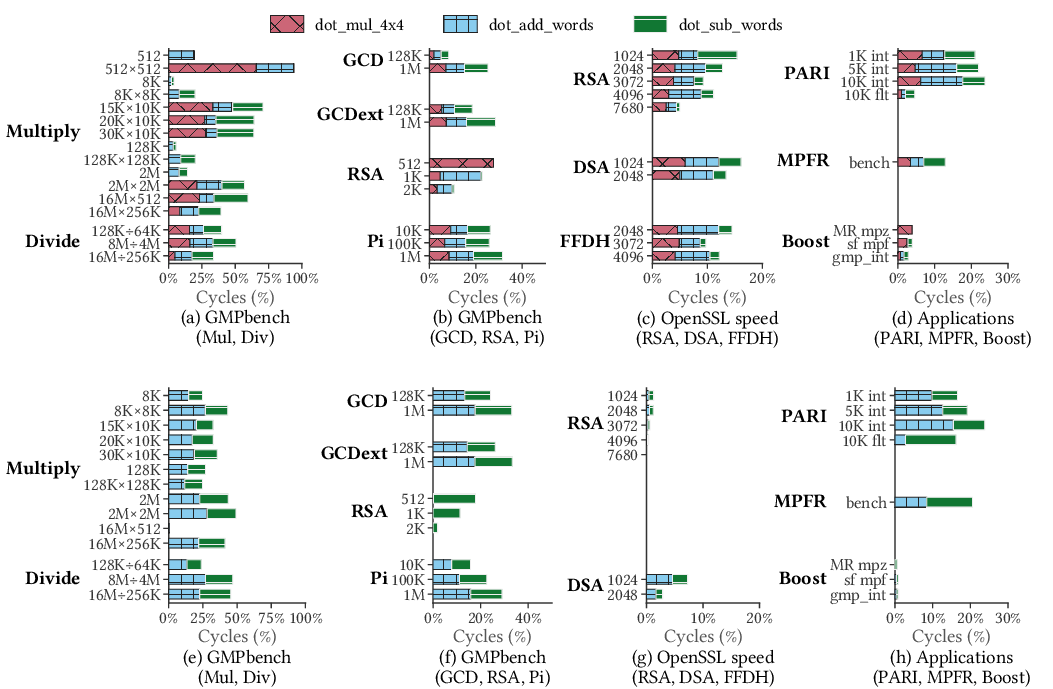}
    \caption{\small{\update{Cycles spent (\%) in DoT's \texttt{dot\_add\_words}, \texttt{dot\_sub\_words}, and \texttt{dot\_mul\_4x4} across  GMPbench, OpenSSL speed, PARI/GP, MPFR, and Boost workloads for generic ((a)--(d)) and assembly ((e)--(h)) builds (measured via \texttt{perf} on the AMD EPYC 9354P (ZG); zero-contribution cases omitted).}}}
    \label{fig:zg_dot_composition}
\end{figure*}

\update{Table~\ref{tab:256_mul_zg} lists instruction counts, average cycle counts, and IPC for 256-bit multiplication routines on the AMD EPYC 9354P (ZG). DoT's $4\times4$ multiplication routine achieves 61.0 average cycles (IPC 4.4), outperforming OpenSSL \texttt{BN\_mul (generic)} (131.3 cycles, IPC 5.6) by $2.15\times$ and Gueron \& Krasnov (167.8 cycles, IPC 2.0) by $2.75\times$, while remaining competitive with hand-tuned assembly baselines (OpenSSL asm: 69.4 cycles, IPC 6.3; GMP asm: 53.4 cycles, IPC 5.6). When radix conversion overhead is avoided in \texttt{dot\_mul\_5x5} (46.2 cycles, IPC 4.8), DoT surpasses GMP's assembly multiplication ($1.16\times$ faster). As on the Intel platforms, the AVX-512IFMA Almost Montgomery Multiplication routine from OpenSSL incurs higher latency (182.4 cycles), showing that specialized modular multipliers do not provide general-purpose acceleration.}

\subsection{\update{DoT-integrated GMP and OpenSSL}}

\begin{figure*}[!t]
    \centering
    \includegraphics[width=\linewidth]{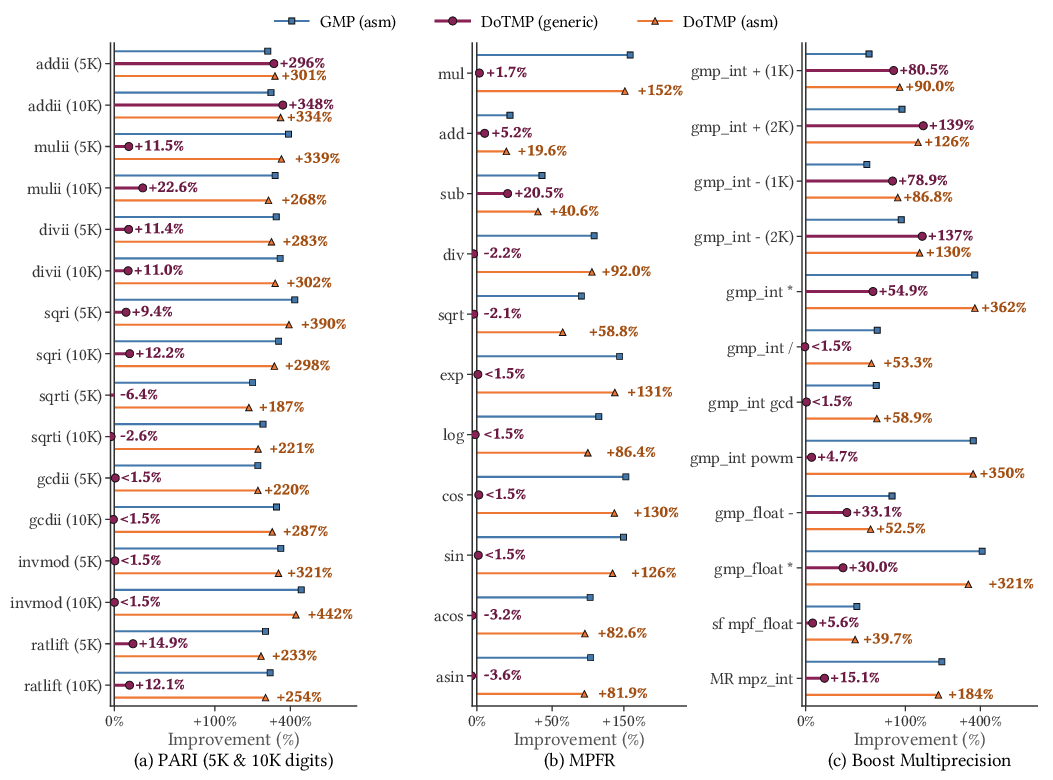}
    \caption{\small{\update{DoTMP's performance improvement over GMP (generic) for PARI/GP (5K \& 10K digits), MPFR, and Boost application benchmarks on the AMD EPYC 9354P (ZG), each also compared against GMP (asm) and DoTMP (asm).}}}
    \label{fig:zg_appbench_full}
\end{figure*}

\update{Figures~\ref{fig:zg_combined_microbench}(c)--(f) compare the execution time speedup of DoTMP and DoTSSL across 512--32768 bit operands against GMP and OpenSSL on ZG. As on ER and SPR, standalone DoT achieves higher speedups for subtraction than addition due to generic subtraction overheads, while library-level integration converges at larger operand sizes as fixed magnitude-comparison overheads are amortized.}

\noindent\update{\textbf{DoTMP.}} \update{For addition and subtraction (Figure~\ref{fig:zg_combined_microbench}(c)), against generic GMP, DoTMP achieves geometric mean speedups of $3.32\times$ for addition and $3.38\times$ for subtraction. Against assembly, DoTMP (asm) delivers $3.52\times$ and $3.51\times$ speedups. For multiplication (Figure~\ref{fig:zg_combined_microbench}(e)), DoTMP (generic) achieves a $1.48\times$ geometric mean speedup over GMP (generic), consistent with ER and SPR.}

\noindent\update{\textbf{DoTSSL.}} \update{For addition and subtraction (Figure~\ref{fig:zg_combined_microbench}(d)), against generic OpenSSL, DoTSSL achieves $2.64\times$ (addition) and $3.63\times$ (subtraction). Against assembly, DoTSSL (asm) delivers $2.68\times$ and $3.69\times$. For multiplication (Figure~\ref{fig:zg_combined_microbench}(f)), DoTSSL (generic) achieves a $1.37\times$ geometric mean speedup over OpenSSL (generic).}

\begin{takeaway}
\update{On Zen 4, DoT accelerates both GMP and OpenSSL across generic and assembly builds: $3.32\times$--$3.52\times$ on GMP add/sub, $2.64\times$--$3.69\times$ on OpenSSL add/sub, and $1.48\times$/$1.37\times$ on multiplication, confirming cross-vendor portability.}
\end{takeaway}

\subsection{\update{DoT's Impact on Higher-Level Applications}}

\noindent\update{\textbf{GMPbench.}}
\update{Figure~\ref{fig:zg_gmpbench_full} shows DoTMP's performance improvement across the GMPbench benchmark suite on ZG for operand sizes from 128 to 16M bits. DoTMP delivers consistent speedups across multiplication, division, $\pi$ calculation, and GCD workloads, following the workload-dependent pattern seen on Intel platforms.}

\noindent\update{\textbf{OpenSSL speed and latency.}}
\update{Figure~\ref{fig:zg_openssl_speed} shows DoTSSL throughput improvements over OpenSSL on the speed benchmark for RSA, RSA KEM, FFDH, and DSA across evaluated key sizes on ZG. Furthermore, Figure~\ref{fig:zg_openssl_latency_cdf} shows the latency CDF comparison, confirming that DoTSSL consistently achieves lower latency across all cryptographic operations.}

\noindent\update{\textbf{Other scientific benchmarks.}}
\update{Figure~\ref{fig:zg_appbench_full} shows the execution time speedup of DoTMP over GMP for PARI/GP (5K \& 10K digits), MPFR, and Boost Multiprecision on ZG. On PARI/GP, DoTMP (generic) achieves $+295.7\%$ (5K) and $+348.0\%$ (10K) improvement for \texttt{addii} over GMP (generic), surpassing GMP's native assembly build ($+264.0\%$ and $+280.0\%$). Boost \texttt{gmp\_int +} gains $+80.5\%$ and $+139.2\%$ for 1K and 2K bit operands, demonstrating that DoT's primitive speedups translate into real-world application improvements on AMD Zen 4.}

\begin{takeaway}
\update{Figure~\ref{fig:zg_dot_composition} breaks down the cycle share of DoT's routines across generic and assembly builds on ZG. In add/sub- and multiply-intensive routines, DoT accounts for a substantial percentage of total execution time, explaining the broad end-to-end performance gains observed across scientific computing and cryptographic applications on the AMD platform.}
\end{takeaway}

\fi
\end{document}